\documentclass[journal,twoside,web]{ieeecolor}
\usepackage{generic}
\usepackage{cite}
\usepackage{amsmath,amssymb,amsfonts}
\usepackage{verbatim} 
\usepackage{algorithmic}
\usepackage{graphicx}
\graphicspath{{figures/}{generalfig}}
\usepackage{subfigure} 
\allowdisplaybreaks[4]

\newtheorem{assumption}{\bf{Assumption}}   
\newtheorem{remark}{\bf{Remark}}		   
\newtheorem{theorem}{\bf{Theorem}}		   
\newtheorem{lemma}{\bf{Lemma}}			   
\newtheorem{proposition}{\bf{Proposition}} 

\usepackage{algorithm,algorithmic}
\usepackage{stfloats}

\makeatletter
\let\NAT@parse\undefined
\makeatother

\usepackage{hyperref}
\hypersetup{
    colorlinks=true,        
    linkcolor=blue,         
    citecolor=red          
}
\usepackage{textcomp}  
\def\BibTeX{{\rm B\kern-.05em{\sc i\kern-.025em b}\kern-.08em
    T\kern-.1667em\lower.7ex\hbox{E}\kern-.125emX}}
\begin{document}
\bstctlcite{IEEEexample:BSTcontrol} 

\title{Stabilization of Unidirectional First-Order PDE-ODE Coupled Systems with Boundary and Distributed Input Delays}
\author{Sanguan Zhong, \IEEEmembership{Student Member, IEEE}, Jie Qi, \IEEEmembership{Member, IEEE}
\thanks{This work was supported by the National Natural Science Foundation of China (62173084), the Project of Science and Technology Commission of Shanghai Municipality, China (23ZR1401800), and the Fundamental Research Funds for the Central Universities (CUSF-DH-T-2023026). ({\slshape Corresponding author: Jie Qi}.)}
\thanks{Sanguan Zhong and Jie Qi are with the School of Information and Intelligent Science, Donghua University, Shanghai, 201620, China (e-mail: zsgguan@mail.dhu.edu.cn; jieqi@dhu.edu.cn).}
}
\maketitle

\begin{abstract}
The paper considers a system of unidirectional first-order hyperbolic partial differential equations (PDEs) coupled with ordinary differential equations (ODEs), modeling two-phase advective transport processes subject to both boundary and distributed input delays. To address the system defined over two spatial domains, we design a set of three backstepping transformations involving both Volterra- and Fredholm-type integral operators, which introduce eight kernel functions: five associated with the PDE states and three with the ODE states. Since the distributed delayed variable spans the two spatial domains, its transformation must integrate over all state variables, leading to three-dimensional kernel functions with singular initial conditions involving Dirac delta functions that serve as sampling operators to establish mappings between variables of different dimensions.
 By employing the method of characteristics and successive approximations, we prove the well-posedness of the kernel equations. On this basis, we prove the invertibility of the Fredholm transformation and establish the exponential stability of the closed-loop system in the $L^\infty$ norm under the delay-compensated controller. Simulation results for a linear example and a nonlinear screw extrusion model are presented to demonstrate the effectiveness of the proposed method.
\end{abstract}


\section{Introduction}\label{sec:introduction}
\IEEEPARstart{U}{nidirectional} (homodirectional) first-order hyperbolic partial differential equations (PDEs) provide a canonical framework for advective transport in multi-region processes, where several spatial segments are coupled in cascade along the direction of propagation. Representative examples include melt-spinning screw extrusion, where solid conveying and melt flow evolve on different spatial regions and are coupled through an interface at the solid-melt transition \cite{dos2013introduction}; supply-chain networks, where each processor is posed on its own spatial segment and governed by a homodirectional first-order hyperbolic conservation law, with multiple regions coupled through routing rules at network nodes \cite{gottlich2005network}; and supercritical open-channel networks, in which each reach is described by the Saint-Venant shallow-water balance laws and, in the supercritical regime, both characteristic speeds are downstream-directed \cite{briani2018fluvial}. In such applications, the distributed transport dynamics often interact with finite-dimensional dynamics (modeled by ordinary differential equations, ODEs) through boundary conditions, forming interconnected PDE-ODE systems \cite{diagne2011well,yu2019bilateral}.
These applications motivate the present work on multi-region, homodirectional hyperbolic PDEs coupled with ODE dynamics.

Control and stabilization of coupled PDE-ODE systems have received sustained attention, largely driven by applications where finite-dimensional actuator/sensor/load dynamics interact with distributed-parameter transport processes.
From a modeling viewpoint, PDE-ODE couplings arise in several canonical architectures. One class represents distributed actuator dynamics via PDEs that drive an ODE plant, with designs covering stabilization, output regulation, semilinear and uncertain features \cite{di2018stabilization, irscheid2023output, strecker2022predictive,ghousein2023adaptive}. Another class captures finite-dimensional actuator/sensor dynamics described by ODEs coupled to a distributed process  \cite{zhang2019input,deutscher2019output}. More complex sandwiched configurations occur when both the actuator and the load (or sensor) are modeled by ODEs surrounding a distributed process, leading to ODE-PDE-ODE systems with robustness enhancements and output-feedback extensions  
\cite{anfinsen2018stabilization, auriol2023robustification, deutscher2018output, wang2021output,di2020robust}. Conversely, PDE-ODE-PDE architectures emerge when a finite-dimensional plant is actuated and sensed through distributed dynamics \cite{aarsnes2019delay, wang2022delay}. 
Beyond fixed-domain settings, PDE-ODE couplings also appear in moving-boundary problems, such as Stefan-type phase-change models \cite{koga2017backstepping}, traffic shockwave dynamics \cite{yu2019bilateral}, and neuron growth models \cite{diehl2014one}, where ODEs govern the evolution of the spatial domain itself. Moreover, by mapping the moving spatial domain onto a fixed reference interval, moving-boundary dynamics can often be reformulated as equivalent PDE-ODE systems in fixed-domain settings \cite{dos2013introduction}. Comprehensive overviews of backstepping developments for PDE control, including interconnected PDE-ODE systems, can be found in a recent survey \cite{vazquez2026backstepping}.

Time delays are inherent in such systems due to the finite propagation speeds of mass, energy, and information. These delays may degrade control performance and can even destabilize the closed-loop system. To address this issue, various control methodologies have been developed for PDE systems with delays \cite{selivanov2018delayed, selivanov2019delayed, katz2021delayed, prieur2019feedback, lhachemi2021feedback}. A prominent class of approaches is based on Lyapunov methods. For instance, in \cite{selivanov2019delayed}, an $H_\infty$ control problem is addressed for two-dimensional diffusion systems with delayed pointwise measurements by constructing a Lyapunov-Krasovskii functional combined with Friedrichs’ inequality. In addition, observer-based LMI approaches are proposed in \cite{selivanov2018delayed} and \cite{katz2021delayed} to stabilize reaction-diffusion systems subject to constant and time-varying delays, respectively.
Another effective methodology relies on spectral analysis for parabolic PDEs. In \cite{prieur2019feedback}, a reaction-diffusion PDE with delayed boundary actuation is decomposed into a finite-dimensional unstable subsystem and an infinite-dimensional stable subsystem via spectral reduction, based on which a finite-dimensional delayed controller is designed to stabilize the overall system.

The PDE backstepping method has exhibited its effectiveness in the design of the compensator for delayed systems \cite{krstic2009control,auriol2018delay,sano2019boundary, qi2021compensation}. The method is pioneered in \cite{krstic2009control} for unstable reaction-diffusion equations and has since been extended to first-order hyperbolic PDEs \cite{aarsnes2019delay,ahmed2021sampled,auriol2018delay2,sano2019boundary,qi2021output}, wave equations \cite{krstic2011dead, wang2011wave,than2021output}, higher-dimensional PDEs \cite{qi2019control,guan2023bilateral}, and systems with distributed input delays \cite{qi2019stabilization,xu2024backstepping,qi2023delay}, spatially-varying distributed input delays \cite{qi2021compensation} as well as  unknown distributed input delays \cite{wang2022delayadaptive}.

However, research on PDE systems involving heterogeneous delays remains limited. Existing results include output-feedback control with boundary actuation delays and measurement delays for coupled linear parabolic PDEs \cite{deutscher2020fredholm} and for first-order hyperbolic PDEs \cite{qi2021output}.

In this paper, we consider an unstable PDE-ODE system that involves both boundary and distributed input delays with distinct delay values. The joint compensation of heterogeneous boundary and distributed delays remains largely unexplored in the existing literature.
The only related work is our previous study on traffic flow \cite{qi2023delay}, where the boundary and distributed actuation share the same physical control mechanism, resulting in identical delays.  
In contrast, the present work considers a more general and practically relevant configuration in which the boundary and distributed actuators are physically distinct, leading to different delay values.

To address this challenge, we propose a set of three backstepping transformations 
that combine Volterra- and Fredholm-type integrals, involving five PDE-associated and three ODE-associated kernel functions. 
Two transformations are designed to compensate for the boundary delay, whereas the third compensates the distributed delay, making its construction more intricate. The kernels arising from the Fredholm-type terms are defined on rectangular domains, and their governing first-order PDEs are subject to two boundary conditions (equivalently, posed as initial-boundary-value problems when one variable is interpreted as time). Moreover, Dirac delta functions enter the boundary conditions and act as sampling operators, thereby forcing the kernel solutions to be defined piecewise. This piecewise structure poses additional challenges for the well-posedness analysis of the resulting kernel equations.

We use the method of characteristics and successive approximations to derive explicit piecewise-defined solutions for a subset of the kernel functions while establishing the well-posedness of the remaining kernel equations. Furthermore, we prove the invertibility of the mixed Fredholm-Volterra transformations with respect to the $L^\infty$ norm. Based on this, we establish the exponential stability of the delay-compensated PDE-ODE system.
Finally, two numerical examples, including an application to a nonlinear screw extrusion process, are presented to demonstrate the effectiveness of the proposed control law.

The contributions of this paper are as follows:
\begin{enumerate}
\item
The method proposed can simultaneously address distinct delays of boundary control and distributed control without requiring ordering of the delays, unlike existing methods that require specific delay ordering assumptions \cite{kharitonov2017prediction,tsubakino2016exact,deutscher2020fredholm,qi2021output}.

\item
To resolve the ill-posedness arising from insufficient boundary conditions, which leads to non-uniqueness of the solution on part of the rectangular domain, we introduce a modified boundary condition that replaces the original one while maintaining consistency. This reformulation preserves the admissible solution set and renders the problem well-posed.

\item  
A unified backstepping transformation is proposed for cross-domain first-order hyperbolic systems with distributed input delay. 
By combining the Dirac delta function and the Heaviside function in the kernel function, the proposed transformation eliminates unstable terms while preserving a triangular target structure, which highlights a distinctive structural feature of the kernel equations.
\end{enumerate}

This paper is organized as follows. In Section~\ref{problem statement}, we present the  two-phase unidirectional PDE-ODE system. Section~\ref{control design} details the backstepping control design  and states the main result. Section~\ref{well-posedness} proves the well-posedness of the kernel equations. The exponential stability of the closed-loop system is analyzed in Section~\ref{stability}. Two simulation examples are provided in Section \ref{simulation}. Conclusions are given in Section~\ref{conclusion}.  
    
    \textit{Notation:} Let $\Omega \subset \mathbb{R}^n$ be a measurable domain, and define
    \begin{equation*}
        L^\infty(\Omega) := \bigl\{
f:\Omega \to \mathbb{R} \;\big|\;
\|f\|_{\infty} < \infty
\bigr\}.
    \end{equation*}
The corresponding norm is denoted by $\|f\|_{\infty} := \operatorname*{ess\,sup}_{x\in\Omega} |f(x)|$.
For any $f \in L^2(\Omega)$, we also use the notation 
$\|f\|:= 
\left( \int_\Omega |f(x)|^2 \, dx \right)^{1/2}$.   
Here, $|\cdot|$ stands for the Euclidean norm. We use the shorthand notation $u_t = \frac{\partial u}{\partial t}$,~ $u_z = \frac{\partial u}{\partial z}$ and define  
    \begin{align*}
        \mathcal{D}_1 =& \{(z,y)\in [0,1]^2:0\le z\le y \le 1\}, \\
        \mathcal{D}_2 =& \{(s,y)\in[0,1]^2\}, \\
        \mathcal{D}_3 =& \{(z,s,y)\in[0,2]\times [0,1]^2\}, \\
        \mathcal{D}_4 =& \{(z,s,y)\in [0,2]\times [0,1]\times [1,2]\}.
    \end{align*}
     Let $\delta(\cdot)$ represent the Dirac delta function and $H(\cdot)$ be the Heaviside function. Moreover, we define the state spaces 
    \begin{align*}
       \mathcal{S} =~& 
    L^\infty(0,1)^3
    \times L^\infty(1,2) 
    \times L^\infty((0,2)\times(0,1) )
    \times \mathbb{R}^n,\\
       \mathcal{S}_c =~& 
   C([0,1])^3
    \times C([1,2]) 
    \times C([0,2]\times[0,1] )
    \times \mathbb{R}^n.
    \end{align*}

\section{Problem Statement}\label{problem statement}
    We consider a coupled unidirectional first-order hyperbolic PDE-ODE system:
    \begin{subequations}\label{main model}
    \begin{align}
		u_t(z,t) =& -au_z(z,t)+\theta u(1,t), ~\quad \quad ~\quad~ \:z\in(0,1],   \label{system 1}\\
		\alpha_t(z,t) =&-a\alpha_z(z,t)+c_1\alpha(z,t)+d_1\alpha(1,t) \notag \\
    &+f(z)u(z,t)+g_1(z)u(1,t)+h_1(z)^{\mathrm{T}}X(t)
    \notag \\
    &+\kappa_1 U_2(z,t-D), \qquad \qquad \qquad z\in(0,1], \label{system 2}  \\
		\beta_t(z,t)=&-b\beta_z(z,t) +c_2\beta(z,t)+d_2\beta(2,t) \notag \\
		&+g_2(z)u(1,t)+h_2(z)^{\mathrm{T}}X(t)\notag \\
        &+\kappa_2 U_2(z,t-D), \qquad \qquad \qquad z\in(1,2],
        \label{system 3}\\
		\dot{X}(t) =&AX(t)+Bu(1,t), \label{system 4}\\
		u(0,t) =& U_1(t-\tau), \label{system 5}\\
		\alpha(0,t) =& 0, \label{system 6}\\
		\beta(1,t) =& \alpha(1,t), \label{system 7}
    \end{align}
    \end{subequations}
   with initial conditions
    \begin{equation}\label{u0, alpha0, beta0, m0}
    \begin{aligned}
        u(z,0) =& u_0(z),  &
        \alpha(z,0) &= \alpha_0(z),
        \\
        X(0) =& X_0, &
        \beta(z,0) & =\beta_0(z), \\
        U_1(t)=&\varpi_1(t), &   t\in& [-\tau, 0], \\
        U_2(z, t)=&\varpi_2(z,t), &  t\in& [-D, 0],
    \end{aligned}
    \end{equation}
    where $X(t)\in\mathbb{R}^n$, $(A,B)$ is a controllable pair with $A\in\mathbb{R}^{n\times n}$ and $B\in \mathbb{R}^{n}$. The constants $a>0$ and $b>0$ denote propagation velocities. The coefficients are $\theta, d_i>0$, $c_i, \kappa_i\in \mathbb{R}$ for $i=1,2$, and the coefficient functions are $f$, $g_1\in C([0,1]; \mathbb{R})$, $g_2 \in C([1,2]; \mathbb{R})$, $h_1\in C([0,1];\mathbb{R}^n)$, and $h_2\in C([1,2];\mathbb{R}^n)$. The open-loop system would be unstable if the coefficients $\theta$ and $d_i$ were sufficiently large \cite{krstic2008backstepping}. We can derive the open-loop transfer function from the input $u(0,t)$ to the output $u(1,t)$ as $G(s)=\frac{s}{(s-\theta)\mathrm{e}^{\frac{s}{a}}+\theta}$, which has at least one pole in the right half-plane when $\theta>a$. This system is subject to both boundary and distributed input delays, where the boundary input is represented by  $U_1(t-\tau)$ and the distributed input by $U_2(z,t-D)$, with $\tau>0$ and $D>0$ denoting the delay parameters.

    \begin{remark}
       Equations \eqref{main model} describe a general linearized model for a unidirectional advective  process governed by first-order hyperbolic PDEs, where transport is dominant, alongside phase transitions or chemical reactions. The upstream ``conveying" zone in this model is characterized by two state variables, commonly representing material concentration and/or thermal energy.       
       After a phase transition or reaction occurs, the downstream zone is modeled using a single state variable, often representing post-reaction temperature. An ODE couples these two regions, capturing dynamic interactions such as the movement of the interface between zones. The system is actuated via boundary input from the feeding mechanism and distributed input from the heating unit, both of which are subject to time delays. Typically, the delay associated with the heater is more significant than that of the feeding system due to thermal inertia and energy propagation characteristics. Classical applications include extrusion energy dynamics \cite{vergnes2006modeling, diagne2011well} and chemical kinetics \cite{suvarov2012simple,ratnakar2017bifurcation,michelsen1996comprehensive}.
    \end{remark}
    \begin{assumption}\label{assumption on delay}
    The distributed input delay $D$ is assumed to be smaller than the transport time in both zones, i.e., $D < \min\left\{\frac{1}{a}, \frac{1}{b}\right\}$, where $\frac{1}{a}$ and $\frac{1}{b}$ represent the time required for the state to propagate from one boundary to the other along the characteristic curves. 
    \end{assumption}    
    \begin{remark}     
     Assumption~\ref{assumption on delay} is specific to first-order hyperbolic systems with distributed input delays, where information propagates along characteristic curves with finite speed. It requires the characteristic transport time to exceed the control lag; otherwise, the delayed distributed input becomes causally ineffective, since the control action arrives only after the state has reached the boundary and can no longer influence the system dynamics. 
\end{remark}

\section{Backstepping Control Design}\label{control design}
   First, we introduce a 1-D and a 2-D transport equation to model the delayed boundary input and the distributed input, respectively, yielding a cascading system:
    \begin{subequations}\label{transport_cascadingsystem}
	\begin{align}
		u_t(z,t) =& -au_z(z,t)+\theta u(1,t), \quad\quad~ \! z\in (0,1],  \label{cascade system 1}\\
		u(0,t) =& r(0,t),\label{cascade system 5}\\
		\tau r_t(s,t) =& r_s(s,t),  \qquad\qquad\qquad\qquad   ~~ s\in [0,1), \label{cascade system 6}
        \\  r(1,t) = &U_1(t), \label{boundary control}\\
		\alpha_t(z,t) =&-a\alpha_z(z,t)+c_1\alpha(z,t)+d_1\alpha(1,t) \notag \\
    &+f(z)u(z,t)+g_1(z)u(1,t)+h_1(z)^{\mathrm{T}}X(t)
    \notag \\
    &+\kappa_1 v(z,0,t), \qquad \qquad \qquad \!z\in (0,1],  \label{cascade system 2}  \\
		\beta_t(z,t)=&-b\beta_z(z,t) +c_2\beta(z,t)+d_2\beta(2,t) \notag \\
		&+g_2(z)u(1,t)+h_2(z)^{\mathrm{T}}X(t)\notag \\
        &+\kappa_2 v(z,0,t), \qquad\qquad \qquad \! z\in (1,2],  
        \label{cascade system 3}\\
         \alpha(0,t) = &0, \quad \beta(1,t) = \alpha(1,t), \\ 
		Dv_t(z,s,t) =& v_s(z,s,t), \qquad ~~~z\in [0,2],~~ \! s\in [0,1),\label{cascade system 7}\\
		v(z,1,t)=&U_2(z,t),\label{cascade system 8}\\
		\dot{X}(t) =&AX(t)+Bu(1,t), \label{cascade system 4}
	\end{align}
    \end{subequations}
    with initial conditions $r(s,0)=r_0(s)=\varpi_1(\tau(s-1))$ and $v(z,s,0) = v_0(z,s)=\varpi_2(z,D(s-1))$.
    The system diagram is illustrated in Fig.~\ref{system structure}.

    \subsection{Backstepping Transformations}  \label{sec_transformation}
	We introduce the following transformations: 
	\begin{align}
		w(z,t)& = u(z,t)\!-\!\!\int_{z}^{1}\!\eta(z,y)u(y,t)dy \!-\!\varphi(z)^{\mathrm{T}}X(t), \label{u trans}
		\\
		\mu(s,t) &=r(s,t)-\int_{0}^{s}\tau aq(s-y,0)r(y,t)dy \notag \\ 
		&~~~-\int_{0}^{1}q(s,y)u(y,t)dy-\psi(s)^{\mathrm{T}}X(t),  \label{r trans}
	\\
		\gamma(z,s,t) &=v(z,s,t)\!-\!\mathcal{T}^v[v](z,s)\!-\!\mathcal{T}^{\alpha}[\alpha](z,s) \!-\! \mathcal{T}^{\beta}[\beta](z,s)\nonumber \\&~~~
		 -\mathcal{T}^{u}[u](z,s)-\phi(z,s)^{\mathrm{T}}X(t), \label{v trans}
	\end{align}
    where
	\begin{subequations}\label{definition of components in trans}
	\begin{align}
		\mathcal{T}^v[v](z,s)\stackrel{\triangle}{=}&D\!\!\int_0^1\!\!\!\!\int_0^s \!\!k(z,s\!-\!\epsilon,y)v(y,\epsilon,t)d\epsilon dy \notag \\
		&+D\!\!\int_1^2\!\!\!\!\int_0^s \!\!l(z,s\!-\!\epsilon,y)v(y,\epsilon,t)d\epsilon dy,\\
		\mathcal{T}^{\alpha}[\alpha](z,s)\stackrel{\triangle}{=}&\frac{1}{\kappa_1}\int_0^1 k(z,s,y)\alpha(y,t)dy, \\
		\mathcal{T}^{\beta}[\beta](z,s)\stackrel{\triangle}{=}&\frac{1}{\kappa_2}\int_1^2 l(z,s,y)\beta(y,t)dy, \\
		\mathcal{T}^{u}[u](z,s)\stackrel{\triangle}{=}&\int_0^1 M(z,s,y)u(y,t)dy,
	\end{align}
	\end{subequations}
    where the kernel functions $\eta$ and $q$ are defined on $\mathcal{D}_1$ and $\mathcal{D}_2$, respectively; $k$ and $M$ are defined on $\mathcal{D}_3$; $l$ is defined on $\mathcal{D}_4$. The vector functions satisfy $\varphi,\psi: [0,1] \to \mathbb{R}^n$, and $\phi : [0,2]\times [0,1] \to \mathbb{R}^n$.

The transformations \eqref{u trans} and \eqref{r trans} jointly eliminate the unstable term $\theta u(1,t)$ in \eqref{cascade system 1}, stabilize the ODE subsystem \eqref{cascade system 4}, and compensate for the delayed boundary input via the boundary control law \eqref{boundary control}.
The third transformation \eqref{v trans}  is designed to remove the unstable terms $d_1\alpha(1,t)$ and $d_2\beta(2,t)$ in \eqref{cascade system 2} and \eqref{cascade system 3}, respectively, and decouple $X$ therein, while simultaneously compensating for the distributed input delay $D$, resulting in a delay-compensated distributed control law $U_2$. 
    \begin{figure*}[hbt]
		\begin{center}
		\includegraphics[width=0.90\linewidth]{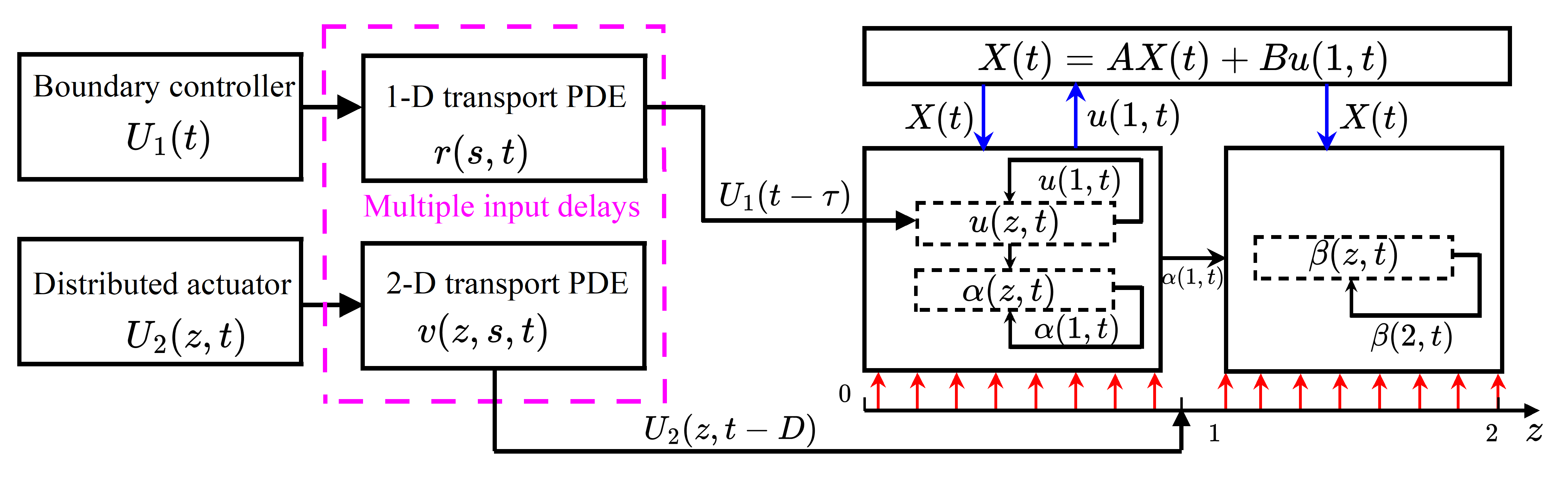}
        \end{center}
        \vspace{-15pt}
		\caption{Schematic diagram of the coupled unidirectional first-order PDE-ODE system with a boundary input delay and a distributed input delay.}
		\label{system structure}
    \end{figure*}
The invertibility of the affine Volterra integral transformations \eqref{u trans} and \eqref{r trans} is straightforward, whereas the invertibility of the Fredholm integral transformation \eqref{v trans} is nontrivial and thus requires a proof. This is stated as follows.
  \begin{proposition}(Invertibility of transformation)\label{invertibility of v trans}
       The transformations \eqref{u trans}--\eqref{v trans} define a bounded linear operator $\mathcal{T}: \mathcal{S} \to \mathcal{S}$, which is boundedly invertible on $\mathcal{S}$.       
  \end{proposition}  
  
    The proof is given in Appendix~\ref{proof of invertibility}.
  
    The inverse transformations are coordinately given by   
	\begin{align}
		u(z,t)=&w(z,t)\!+\!\!\int_{z}^{1}\!\!\breve{\eta}(z,y)w(y,t)dy \!+\! \breve{\varphi}(z)^{\mathrm{T}}X(t),  \label{w trans}
		\\
		r(s,t)=&\mu(s,t)+\int_{0}^{s}\tau a\breve{q}(s-y,0)\mu(y,t)dy \notag 
        \\
		&+\int_{0}^{1}\breve{q}(s,y)w(y,t)dy+\breve{\psi}(s)^{\mathrm{T}}X(t),\label{mu trans} 
        \\
        v(z,s,t) =&\gamma(z,s,t)\!+\!\mathcal{\breve{T}}^{\gamma}[\gamma](z,s)\!+\!\mathcal{\breve{T}}^{\alpha}[\alpha](z,s)\!+\!\mathcal{\breve{T}}^{\beta}[\beta](z,s) \nonumber \\
		  &\!+\!\mathcal{\breve{T}}^{w}[w](z,s)\!+\!\mathcal{\breve{T}}^{\mu}[\mu](z,s)\!+\!\breve{\phi}(z,s)^{\mathrm{T}}\! X(t), \label{gamma trans}
	\end{align}
    where
    \begin{subequations}\label{definition of components in inverse trans}
    \begin{align}
		\mathcal{\breve{T}}^{\gamma}[\gamma](z,s)\stackrel{\triangle}{=}&D\!\!\int_0^1\!\!\!\!\int_0^s \!\!\breve{k}(z,s\!-\!\epsilon,y)\gamma(y,\epsilon,t)d\epsilon dy  \notag\\
		&+D\!\!\int_1^2\!\!\!\!\int_0^s \!\!\breve{l}(z,s\!-\!\epsilon,y)\gamma(y,\epsilon,t)d\epsilon dy,\\
		\mathcal{\breve{T}}^{\alpha}[\alpha](z,s)\stackrel{\triangle}{=}&\frac{1}{\kappa_1}\int_0^1 \breve{k}(z,s,y)\alpha(y,t)dy, \\
		\mathcal{\breve{T}}^{\beta}[\beta](z,s)\stackrel{\triangle}{=}&\frac{1}{\kappa_2}\int_1^2 \breve{l}(z,s,y)\beta(y,t)dy, \\
		\mathcal{\breve{T}}^{w}[w](z,s)\stackrel{\triangle}{=}&\int_0^1 \breve{M}(z,s,y)w(y,t)dy, \\
		\mathcal{\breve{T}}^{\mu}[\mu](z,s)\stackrel{\triangle}{=}&\int_0^1 \breve{N}(z,s,y)\mu(y,t)dy,\label{T mu definition}
	\end{align}  
    \end{subequations}
    where the inverse kernel functions $\breve{\eta}$ and $\breve{q}$ are defined on $\mathcal{D}_1$ and $\mathcal{D}_2$, respectively; $\breve{k}$, $\breve{M}$, $\breve{N}$ are defined on $\mathcal{D}_3$; $\breve{l}$ is defined on $\mathcal{D}_4$. The vector functions satisfy $\breve{\varphi}, \breve{\psi}: [0,1]\to \mathbb{R}^n$, and $\breve{\phi}: [0,2]\times[0,1]\to \mathbb{R}^n$.

    Applying the transformations \eqref{u trans}--\eqref{v trans}, we convert the original system \eqref{transport_cascadingsystem} into the following target system: 
    \begin{subequations}\label{target system}
    \begin{align}
		w_t(z,t) =& -aw_z(z,t), \quad ~\quad \quad ~\quad\quad~ \:z\in(0,1],\label{target 1} \\
		w(0,t) =& \mu (0,t), \label{target 5}\\
		\tau \mu_t(s,t) =& \mu_s(s,t), \qquad\qquad\qquad\qquad   ~ s\in [0,1), \label{target 6} \\
        \mu(1,t) =& 0, \\
		\alpha_t(z,t) =& -a\alpha_z(z,t) - \bar{c}_1\alpha(z,t) \notag \\
		&+ \kappa_1\gamma(z,0,t)+F_1[w], \label{target 2}  \\
		\beta_t(z,t) =& -b\beta_z(z,t) - \bar{c}_2\beta(z,t)-\frac{\kappa_2 d_1}{\kappa_1}\alpha(1,t) \notag \\
		&+\kappa_2 \gamma(z,0,t) + F_2[w], \label{target 3} \\
       \alpha(0,t) =& 0, \quad \beta(1,t) = \alpha(1,t), \label{target boundary alpha beta}  
        \\
		D\gamma_t(z,s,t) =& \gamma_s(z,s,t), \qquad ~~~z\in [0,2],~~ \! s\in [0,1),\label{target 7} \\
        \gamma(z,1,t) =& 0, \label{target U2 boundary} \\
		\dot{X}(t) =& (A+BK)X(t) + Bw(1,t), \label{target 4}
	\end{align}
    \end{subequations}
    with initial conditions $w(z,0)=w_0(z)$, $\mu(s,0)=\mu_0(s)$, and $\gamma(z,s,0)=\gamma_0(z,s)$, where $\bar{c}_i>|c_i|$, $i=1,2$, are design parameters used to tune the convergence rates, with $c_i$ being the reaction coefficients in \eqref{system 2} and \eqref{system 3}. Choose $K\in\mathbb{R}^{1\times n}$ such that $A+BK$ is Hurwitz. 
    The functions $F_1[w]$ and $F_2[w]$ depend on the inverse kernel function $\breve \eta$ in \eqref{w trans}, given by  
    \begin{align}
		F_1[w] =& \breve{f}(z)\left(w(z)\!+\!\!\!\int_z^1 \!\!\breve{\eta}(z,y)w(y)dy\right)\!+\!g_1(z)w(1), \label{F1 rewrite}\\
		F_2[w] =& g_2(z)w(1), \label{F2 rewrite}
    \end{align}
    where 
    \begin{align}
        \breve{f}(z) =& f(z)\!+\!D\zeta_1\int_0^{\frac{z}{Da}}\!\!\mathrm{e}^{Dc_1\varsigma}f(z\!-\!Da\varsigma)d\varsigma H(Da\!-\!z). \label{breve f}
    \end{align}
    The parameter $\zeta_1=\bar c_1+c_1$ combines the convergence-rate design parameter $\bar c_1$ introduced in the target system~\eqref{target system} with the given reaction coefficient $c_1$.

    \subsection{Control Law}\label{section main result}
    Therefore, we obtain the controllers:
    \begin{align}
	&\begin{aligned}\label{U1}
		U_1(t) =& \int_{0}^{1}q(1,y)u(y,t)dy+\psi(1)^{\mathrm{T}}X(t) \\
		&+\int_{t-\tau}^{t}aq(\frac{t-\varsigma}{\tau}, 0)U_1(\varsigma)d\varsigma, 
	\end{aligned}\\
	&\begin{aligned}\label{U2}
		U_2 (z,t) =&  \frac{1}{\kappa_1}\int_0^1 k(z,1,y)\alpha(y,t)dy
		 \\
		&+\frac{1}{\kappa_2}\int_1^2 l(z,1,y)\beta(y,t)dy  \\        &+\int_{0}^{1}\!\!M(z,1,y)u(y,t)dy\!+\!\phi(z,1)^{\mathrm{T}}X(t)  \\
        &+\int_0^1\int_{t-D}^t k(z,\frac{t-\varsigma}{D},y)U_2(y,\varsigma)d\varsigma dy  \\
		&+\int_1^2\int_{t-D}^t l(z,\frac{t-\varsigma}{D},y)U_2(y,\varsigma)d\varsigma dy.
	\end{aligned}
    \end{align}

    \subsection{Main Result}
\begin{theorem}\label{main result}
Consider the system \eqref{transport_cascadingsystem} under the control laws \eqref{U1} and \eqref{U2}, with initial conditions $(u_0, r_0, \alpha_0, \beta_0, v_0, X_0)\in \mathcal{S}_c$ satisfying the corresponding compatibility conditions. Then, the closed-loop system admits a unique weak solution $(u, r, \alpha, \beta, v, X) \in C([0,\infty); \mathcal{S})$ and its zero equilibrium is exponentially stable in $\mathcal{S}$, i.e., there exist constants $\varGamma_0,~ \nu_0> 0 $ such that
\begin{align}
   V_1(t) \leq \varGamma_0 \, \mathrm{e}^{-\nu_0 t} V_1(0),
\end{align}
where  
\begin{align}
    V_1(t) = &\| u \|_\infty^2+ \| r \|_\infty^2 + \| \alpha \|_\infty^2 + \| \beta \|_\infty^2  + \| v \|_\infty^2 + |X|^2.\label{infty norm of the original states}
\end{align}
\end{theorem}

\section{Kernel Equations and Their Well-posedness Analysis}\label{well-posedness}
  \subsection{Kernel Equations} \label{kernel functions}
    To map the original system \eqref{transport_cascadingsystem} to the target system \eqref{target system}, we derive three groups of coupled kernel equations.
    First, the kernel functions of transformations \eqref{u trans} and \eqref{r trans}
     satisfy
    \begin{align}
        \varphi(z)^{\mathrm{T}} = & K\mathrm{e}^{A\frac{1-z}{a}}, ~~ \quad \psi(s)^{\mathrm{T}} = \varphi(0)^{\mathrm{T}}\mathrm{e}^{A\tau s}, \quad\label{varphi, psi}        
    \end{align}
    \begin{subequations}\label{eta, p, q}
        \begin{align}
        \eta_z(z,y) =& -\eta_y(z,y), \\
        \eta(z,1) = &\int_z^1\frac{\theta}{a} \eta(z,y)dy+\frac{1}{a}\left(\varphi(z)^{\mathrm{T}}B-\theta\right),        
         \\
        q_s(s,y) =& \tau aq_y(s,y), \\
        q(0,y) =& \eta(0,y), \\
        q(s,1) =& \int_0^1\frac{\theta}{a} q(s,y)dy+\frac{1}{a}\psi(s)^{\mathrm{T}}B.
    \end{align}
    \end{subequations}
       
    Second, the kernel functions of transformation  \eqref{v trans} satisfy 
    \begin{align}
      \left\{\begin{aligned}\label{k, l}
    k_s(z,s,y) =& Dak_y(z,s,y) +Dc_1k(z,s,y),\\
    k(z,s,1) =&\frac{d_1}{a}\int_0^1k(z,s,y)dy+\frac{b\kappa_1}{a\kappa_2}l(z,s,1), \\
    k(z,0,y) =& -d_1\delta(y-1)-\zeta_1\delta(y-z), \\
    l_s(z,s,y) =& Dbl_y(z,s,y) +Dc_2l(z,s,y),\\
    l(z,s,2) =& \frac{d_2}{b}\int_1^2 l(z,s,y)dy, \\
    l(z,0,y) =& -d_2\delta(y-2)H(z-1)-\zeta_2\delta (y-z),
    \end{aligned}\right.
    \end{align}
    where $\zeta_i=\bar c_i+c_i>0$, and the following coupled PDE-ODE system:
    \begin{align}
    		M_{s}(z,s,y) =& DaM_{y}(z,s,y) +\frac{D}{\kappa_1}k(z,s,y)f(y),  \label{M}\\
		M(z,s,1) =& \int_0^1\!\frac{\theta}{a} M(z,s,y)dy \!+\! \frac{1}{a}(\varPi[k,l]\!+\!\!\phi(z,s)^{\mathrm{T}}\!B),  \label{M(z,s,1)}
         \\
		M(z,s,0) =& 0, \label{M(z,s,0)}   \\
    \label{phi(z,s)}
       \phi_s(z,s) =& DA^{\mathrm{T}}\phi(z,s) + \varPhi[k,l],\\   
		\phi(z,0) =& -\frac{1}{\kappa_1}\left[f(z)\breve{\varphi}(z)+g_1(z)K^{\mathrm{T}}+h_1(z)\right]   \nonumber
        \\&-\int_0^1M(z,0,y)\breve{\varphi}(y)dy, \label{phi(z,0)_1}
        \quad  \quad \mathrm{if}~ z < 1,  \\    \nonumber
	\phi(z,0) =& -\frac{1}{\kappa_2}\left[g_2(z)K^{\mathrm{T}}+h_2(z)\right]  \\   
        &-\int_0^1M(z,0,y)\breve{\varphi}(y)dy, \quad  \quad  \mathrm{if}~ z \ge 1,  \label{phi(z,0)_2}
       \end{align}
    where 
    \begin{align}  
		\varPi[k,l] =& \frac{1}{\kappa_1}\!\!\int_0^1k(z,s,y)g_1(y)dy \!+\!
		\frac{1}{\kappa_2}\!\!\int_1^2l(z,s,y)g_2(y)dy, \label{Pi} \\
        \varPhi[k,l] =& \frac{D}{\kappa_1}\!\!\!\int_0^1k(z,s,y)h_1(y)dy \!+\!\frac{D}{\kappa_2}\!\!\!\int_1^2l(z,s,y)h_2(y)dy, \label{Phi}
    \end{align}   	    
   and $\breve{\varphi}(\cdot)$ is the inverse kernel function appearing in \eqref{w trans}, whose explicit expression will be given in the next subsection. 

The dependencies among these kernel functions are summarized as follows. 
The first group, $\{\varphi, \psi, \eta, q\}$, defined by \eqref{varphi, psi} and \eqref{eta, p, q}, exhibits a sequential dependency: 
$\varphi$ is independent; $\psi$ depends on $\varphi$; $\eta$ depends on $\varphi$; and $q$ depends on both $\eta$ and $\psi$. 
Accordingly, they can be solved sequentially in the order $\varphi \rightarrow \psi \rightarrow \eta \rightarrow q$.
    The kernel equation for $\eta$ can be transformed, via a change of variables, into the same form as the kernel equation studied in \cite{krstic2008backstepping} and \cite{anfinsen2019adaptive}, which yields $\eta\in C(\mathcal{D}_1)$. Moreover, by introducing a similar change of variables, the kernel equation for $q(s,y)$ can be transformed into a simplified version of the equation governing $u$ given by Equation (12) in \cite{sano2019boundary}, which gives $q\in C([0,1];L^2(0,1))$ (see Appendix~\ref{appendix:well-posedness of eta and q equation} for details). Consequently, the transformations \eqref{u trans} and \eqref{r trans} can map the state variables $u\in L^\infty(0,1)$ and $r\in L^\infty(0,1)$ to $w\in L^\infty(0,1)$ and $\mu\in L^\infty(0,1)$, respectively.
   
In the second group $\{l, k, M, \phi\}$, defined by \eqref{k, l}--\eqref{phi(z,0)_2}, the function $l$ is independent; $k$ depends on $l$; and $M$ and $\phi$ are coupled functions depending on $k$ and $l$, with $\phi$ additionally depending on $\breve{\varphi}$.

\begin{remark}\label{boundary constraint}
Note that $M(z,s,0) = 0$ is required to ensure that the transformations \eqref{u trans}--\eqref{v trans} consistently map the original system to the target system, i.e.,  $M(z,s,0)u(0,t) = 0$. However, this condition, together with the boundary condition at $y=1$, does not uniquely determine the solution, because there exists a region not governed by any boundary condition, namely, the gray area between the dashed and dash-dotted characteristic lines in Fig.~\ref{M line}. To close the system, an additional boundary (initial) condition at $s=0$ is introduced, given by 
\begin{align}\label{M(z,0,y)}
	M(z,0,y) \!=\! 
	\begin{cases}
		m_0(y)\delta(y\! -\! z), & 0 \le y \le Da, \\
		0, & Da < y \le 1,
	\end{cases}
\end{align}
where 
\begin{equation}
           m_0(y) =\frac{D\zeta_1}{\kappa_1}\int_0^{\frac{y}{Da}}\mathrm{e}^{Dc_1\varsigma}f(y-Da\varsigma)d\varsigma.\label{m0}
\end{equation}
This boundary condition guarantees that the solution of the $M$-equation satisfies $M(z,s,0)=0$, thereby establishing its well-posedness, as shown in the subsequent section.
\end{remark}

    \subsection{Inverse Kernel Equations}\label{inverse kernel}  
    The inverse kernel equations are derived through transforming \eqref{target system} to the original system \eqref{transport_cascadingsystem}, given by
      \begin{align}
            \breve{\psi}(s)^{\mathrm{T}} =& \breve{\varphi}(0)^{\mathrm{T}}\mathrm{e}^{(A+BK)\tau s}, \label{breve psi}\\
            \breve{\varphi}(z)^{\mathrm{T}} =&  K\mathrm{e}^{(A+BK)\frac{1-z}{a}}\!-\!\int_z^1\frac{\theta}{a}K\mathrm{e}^{(A+BK)\frac{y-z}{a}}dy.\label{breve varphi} 
    \end{align}
    \begin{subequations}\label{breve eta q}
    \begin{align}
        \breve{\eta}_z(z,y) =&-\breve{\eta}_y(z,y)  , \\
        \breve{\eta}(z,1) =& \frac{1}{a}(\breve{\varphi}(z)^{\mathrm{T}}B-\theta), \\
        \breve{q}_s(s,y) =& \tau a\breve{q}_y(s,y), \\
        \breve{q}(0,y) =& \breve{\eta}(0,y), \qquad \breve{q}(s,1) = \frac{1}{a}\breve{\psi}(s)^{\mathrm{T}}B. \qquad\quad~
    \end{align}
    \end{subequations}
    \begin{align}\label{breve k, l}
    \left\{
    \begin{aligned}
		\breve{k}_s(z,s,y) =&  Da\breve{k}_y(z,s,y) -D\bar{c}_1\breve{k}(z,s,y),
		\\
		\breve{k}(z,s,1) =& \frac{b\kappa_1}{a\kappa_2}\breve{l}(z,s,1)-\frac{d_1}{a}\int_1^2\breve{l}(z,s,y)dy, \\
		\breve{k}(z,0,y) =& -d_1\delta(y-1)-\zeta_1\delta(y-z),  \\
		\breve{l}_s(z,s,y) =& Db\breve{l}_y(z,s,y) -D\bar{c}_2\breve{l}(z,s,y), 
		 \\
        \breve{l}(z,s,2) =& 0,  \\
		\breve{l}(z,0,y) =& -d_2\delta(y\!-\!2)H(z-1)\!-\!\zeta_2\delta (y-z). \quad
    \end{aligned}\right.
    \end{align}
    \begin{subequations}\label{breve M}
        \begin{align}
		\breve{M}_s(z,s,y) =& Da\breve{M}_y(z,s,y) +\frac{D}{\kappa_1}\breve{k}(z,s,y)\breve{f}(y)+\breve{\varUpsilon}[\breve k],  \\
		\breve{M}(z,s,1) =& \frac{1}{a}(\breve{\varPi}[\breve k,\breve l](z,s)+\breve{\phi}(z,s)^{\mathrm{T}}B), \\
		\breve{M}(z,0,y) =& \varrho(z,y),
    \end{align}
    \end{subequations}
    where 
    \begin{align}
    \breve{\varPi}[\breve k,\breve l] = &\frac{1}{\kappa_1}\!\int_0^1\breve{k}(z,s,y)g_1(y)dy \!+\!
		\frac{1}{\kappa_2}\!\int_1^2\breve{l}(z,s,y)g_2(y)dy,  \label{breve varPi}\\
		\breve{\varUpsilon} [\breve k]= &\!\!\int_0^y\frac{D}{\kappa_1}\breve{k}(z,s,\xi)\breve{f}(\xi)\breve{\eta}(\xi,y)d\xi, \label{breve varUpsilon} \\
        \varrho(z,y) =& M(z,0,y)+\int_{0}^{y}M(z,0, \xi)\breve{\eta}(\xi,y)d\xi, \label{varrho}
    \end{align}
    with $\breve{f}$ defined in  \eqref{breve f}. 
    Substituting \eqref{M(z,0,y)} into \eqref{varrho}, we further obtain
    \begin{align}
        \varrho(z,y)=& m_0(y)\delta(y-z)H(Da-y)\notag \\
        &+ m_0(z)\breve \eta (z,y)H(\min\{y,Da\}-z).
    \end{align}
    
    And
    \begin{align}\label{breve N equation}
    \left\{
    \begin{aligned}
		\breve{N}_s(z,s,y)=& \frac{D}{\tau}\breve{N}_y(z,s,y), \qquad\qquad \qquad\qquad\quad\\
		\breve{N}(z,s,1) =& \tau a\breve{M}(z,s,0), \\
		\breve{N}(z,0,y) =& 0.
    \end{aligned}\right.
    \end{align}
    \begin{align}\label{breve phi}
        \left\{\begin{aligned}
            \breve{\phi}_s(z,s) =& D(A+BK)^{\mathrm{T}}\breve{\phi}(z,s), \\
            \breve{\phi}(z,0)\!=& -\frac{1}{\kappa_1}[f(z)\breve{\varphi}(z)\!\!+\!\!g_1(z)K^{\mathrm{T}}\!\!\!+\!h_1(z)], \:\;z \!<\!1, \\
	\breve{\phi}(z,0) =& -\frac{1}{\kappa_2}[g_2(z)K^{\mathrm{T}}+h_2(z)],  \qquad \qquad~ z \!\ge\! 1. 
        \end{aligned}\right.
    \end{align}

    \subsection{\texorpdfstring{Well-posedness of the Kernel Equations for $\mathit{k}$ and $\mathit{l}$}{Well-posedness of the kernel equations for k and l}}\label{sec_well_k_l}
    
    To define the branching structure of the kernel functions, wherein different expressions apply in different subdomains, we first introduce the following auxiliary functions:
  \begin{align}
        &\lambda_1(s) \!=\! -Das+\!1, ~~ \lambda_2(s) \!=\! -Dbs+\!\!2,~~\lambda_3(s)=-\frac{D}{\tau}s+1,\\
        & \vartheta_1(z,s) \!=\! \min\{\lambda_1(s)\!+\!z,1\}, \label{lambda s} \\
        & \vartheta_2(z,s) \!=\! \min\{\lambda_1(s)\!\!+\!\!\frac{a}{b}(z\!\!-\!\!1),1\},\\
        &\vartheta_3(z,s) \!=\! \min\{2,\max\{\lambda_2(s),\lambda_2(s)+z-1\}\},\label{vartheta_3}\\
        &\vartheta_4(z,s)\!=\!\min\{1, \lambda_1(s)+\max\{0,\frac{a}{b}(z-1)\}\},\label{vartheta_4}\\
         &\vartheta_5(z,s)\!=\!\min\{\lambda_1(s), \max\{0, \lambda_1(s)+z-1\}\} .\label{vartheta_5}
    \end{align}

\begin{figure*}[htp]
		\begin{center}
                \subfigure[]{\includegraphics[width=0.32\linewidth]{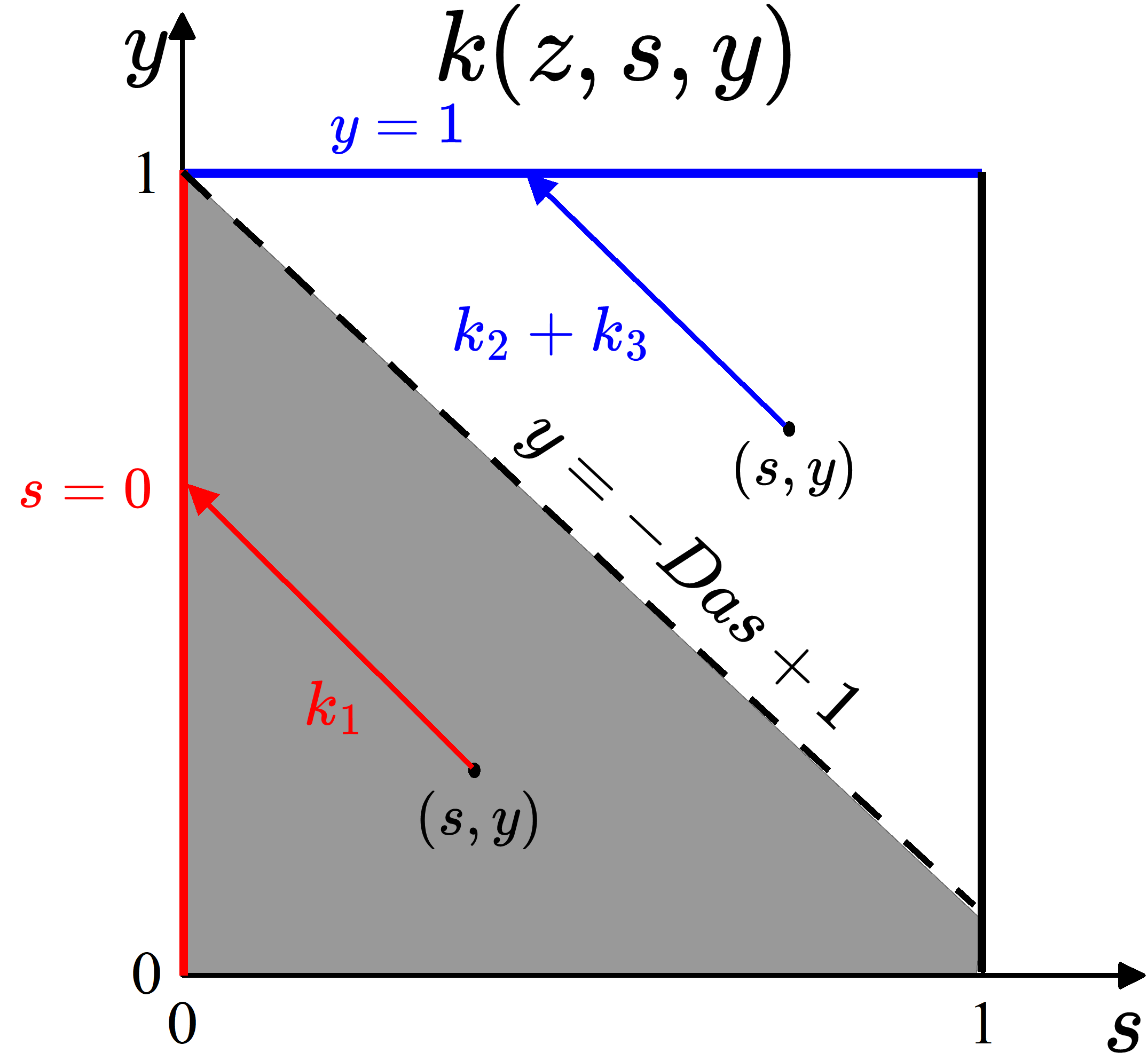}\label{k line}
			}
                \subfigure[]{\includegraphics[width=0.32\linewidth]{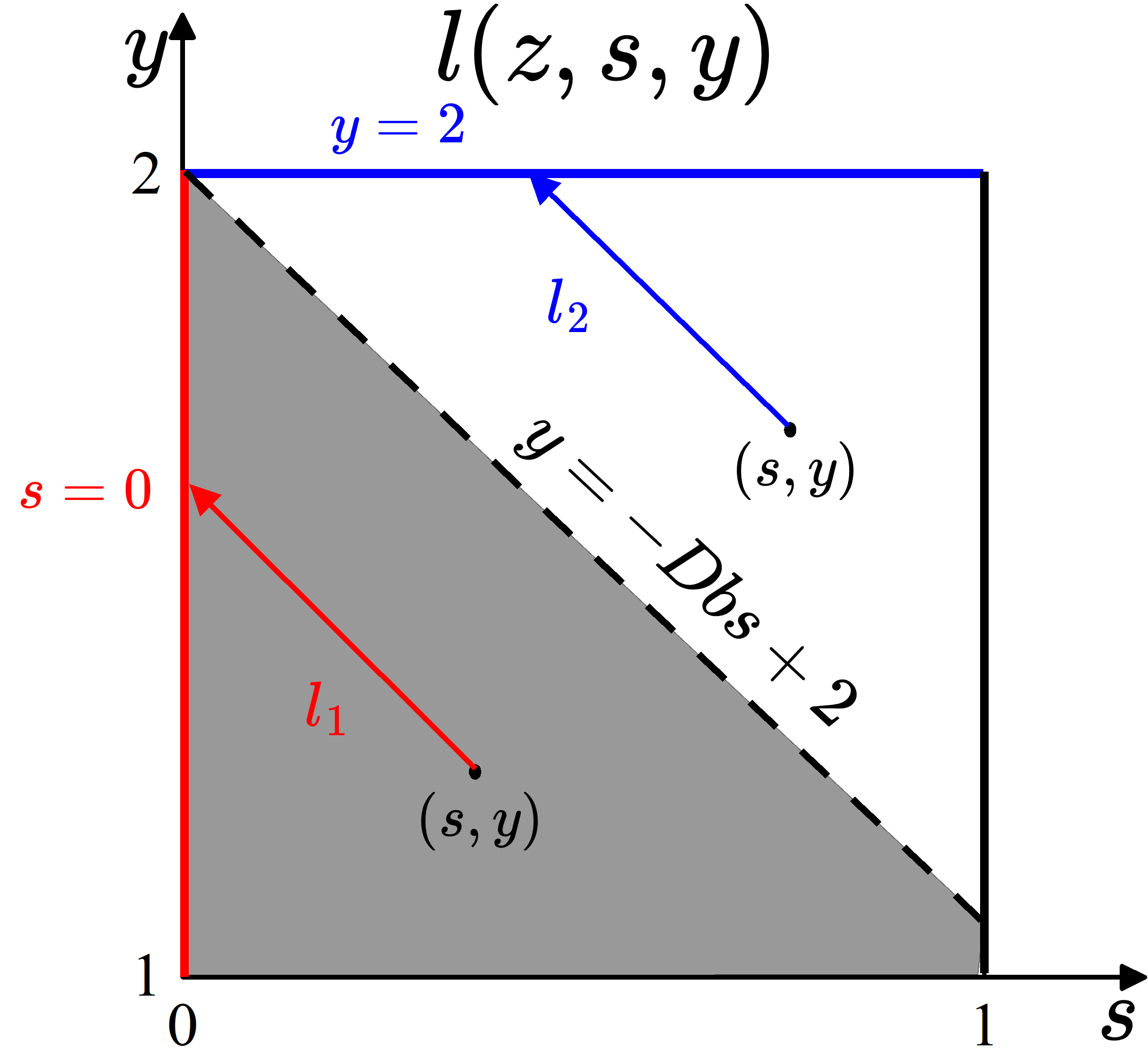}\label{l line}
			}
			\subfigure[]{\includegraphics[width=0.32\linewidth]{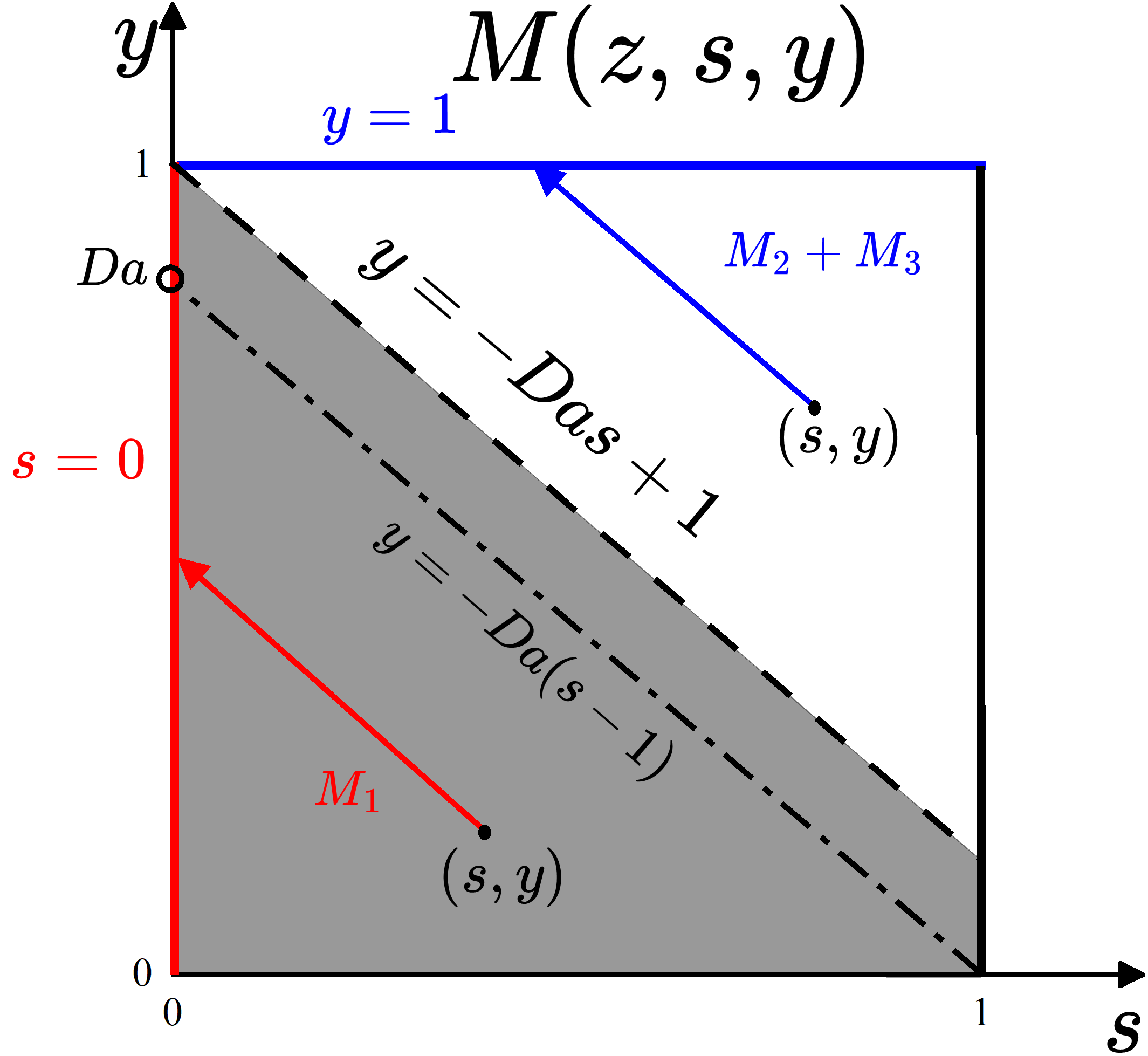}\label{M line}
			}
		\end{center}
        \vspace{-10pt}
        \caption{Schematic of the domain of $k$, $l$ and $M$, showing subregion division and characteristic lines. \subref{k line} Domain $\mathcal{D}_3$ for $k(z,s,y)$. \subref{l line} Domain $\mathcal{D}_4$ for $l(z,s,y)$. \subref{M line} Domain $\mathcal{D}_3$ for $M(z,s,y)$. 
        }
        \label{figure:characteristics}
    \end{figure*}

    \begin{lemma}
    The system of coupled kernel equations \eqref{k, l}, with the initial condition involving delta functions at $s=0$, admits a unique solution $(k,l)$ on the domain $\mathcal{D}_3 \times \mathcal{D}_4$.
    \end{lemma}
    
\begin{proof}
   We first remove the source terms from \eqref{k, l} via the change of variables:
    \begin{align}
        l  \!=\! \mathrm{e}^{Dc_2s}\bar{l},~~\quad  \;
        k  \!=\! \mathrm{e}^{Dc_1s}\bar{k},\label{change of variables}
    \end{align}
   which yields 
        \begin{align}
		\bar{k}_s(z,s,y) =& Da\bar{k}_y(z,s,y),\label{k-1}\\
		\bar{k}(z,s,1) =&\frac{d_1}{a}\int_0^1\bar{k}(z,s,y)dy\!+\!\frac{b\kappa_1}{a\kappa_2}\mathrm{e}^{-Dc_1s}l(z,s,1), \label{k-2}\\
		\bar{k}(z,0,y) =& -{d}_1\delta(y-1)-{\zeta}_1\delta(y-z), \label{k-3}\\
		\bar{l}_s(z,s,y) =& Db\bar{l}_y(z,s,y),\label{l-1}\\
		\bar{l}(z,s,2) =& \frac{d_2}{b}\int_1^2 \bar{l}(z,s,y)dy, \label{l-2}\\
		\bar{l}(z,0,y) =& -{d}_2\delta(y-2)H(z-1)-{\zeta}_2\delta (y-z).\label{l-3}
	\end{align}
    
   We first consider the PDE associated with $\bar{l}$, which is solvable via the method of characteristics. For $1 \le y \le \lambda_2(s)$, it is given explicitly by 
    \begin{align}
        \bar{l}(z,s,y)=&-{d}_2\delta(y-\lambda_2(s))H(z-1)\notag \\
       & -{\zeta}_2\delta(y-\lambda_2(s)+2-z). \label{bar l first part}
    \end{align}
In contrast, for $\lambda_2(s) < y \le 2$, which corresponds to the white region in Fig.~\ref{l line}, the solvability analysis requires the combined use of the method of characteristics and the method of successive approximations. We first reformulate the equation in integral form via the method of characteristics as
\begin{align}
        \bar{l}(z,s,y) 
        =&\int_{\chi_2(s,y)}^2\frac{d_2}{b}\bar{l}(z,\chi_1(s,y),\xi)d\xi\nonumber \\       
        -\frac{d_2^2}{b}&H(z-1)\int_1^{\chi_2(s,y)}\delta(\xi-\chi_2(s,y))d\xi  \notag \\
        -\frac{d_2}{b}&{\zeta}_2\int_{1}^{\chi_2(s,y)}\!\delta(\xi-\chi_2(s,y)+2-z)d\xi, \label{l integral the first}
    \end{align}
    where $\chi_1(s,y)=s-\frac{2-y}{Db}$, $\chi_2(s,y)=-Dbs-y+4$.
    Applying the shifting property of the delta function \cite{tang2007mathematical}, we have 
    \begin{align}
       \bar{l}(z,s,y) =&\varGamma_{11}(z) \!+\!\varOmega[\bar{l}](z,s,y),&& \mathrm{if}~ \lambda_2(s)\!<y\!\le \vartheta_3(z,s),\label{with zeta2} \\
       \bar{l}(z,s,y) =&\varGamma_{12}(z) \!+\!\varOmega[\bar{l}](z,s,y), &&\mathrm{if}~ \vartheta_3(z,s) \!<\! y \!\le\! 2,\label{without zeta2}
    \end{align}
    where $\varOmega$, $\varGamma_{11}(z)$, and $\varGamma_{12}(z)$ are defined as
    \begin{align}
        &\varOmega[\bar{l}](z,s,y)=\int_{\chi_2(s,y)}^2\frac{d_2}{b}\bar{l}(z,\chi_1(s,y),\xi)d\xi, \label{operator varOmega}\\
        &\varGamma_{11}(z)\!=\!\frac{d_2^2}{-b}H(z\!-\!1)-\frac{d_2}{b}{\zeta}_2,~\varGamma_{12}(z)\!=\!-\frac{d_2^2}{b}H(z\!-\!1).\label{varGamma11}
    \end{align}
   First, consider \eqref{with zeta2}. We apply the recursive iteration starting with $\bar{l}^{0}(z,s,y)=0$,
        \begin{align}
            \bar{l}^{n+1}(z,s,y) =& \varGamma_{11}(z)\!+\!\varOmega[\bar{l}^n](z,s,y). 
        \end{align}
Further, define $\Delta{\bar{l}}^{n}=\bar{l}^{n+1}-\bar{l}^{n}$ for $n=0, 1, 2, \cdots$, which gives  $\Delta \bar{l}^{0}(z,s,y)=\varGamma_{11}(z)$  and 
\begin{align}
            \Delta \bar{l}^{n+1}(z,s,y)=\varOmega[\Delta\bar{l}^n](z,s,y).\label{l: n+1 and n}
        \end{align}
We can obtain
\begin{align}
        &\Delta \bar{l}^{n}(z,s,y)=\varGamma_{11}(z)\frac{1}{n!}\Bigl(\frac{d_2}{b}\bigl(y+Dbs-2\bigr)\Bigr)^{n},\label{assume on l: n+1 and n}
    \end{align}
which gives     
    \begin{align}
        \bar{l}(z,s,y) 
        =&\sum\limits_{n=0}^{\infty}\varGamma_{11}(z)\frac{1}{n!}(\frac{d_2}{b})^{n}(y+Dbs-2)^{n} . 
    \end{align}
The solution to \eqref{without zeta2} is obtained similarly. Consequently, 
    \begin{align}
        \bar{l}(z,s,y) =&\varGamma_{11}(z)\mathrm{e}^{\frac{d_2}{b}(y-\lambda_2(s))}, ~\mathrm{if}~ \lambda_2(s)\!< \! y \!\le \!\vartheta_3(z,s), \label{bar l second part} \\
        \bar{l}(z,s,y) =&\varGamma_{12}(z)\mathrm{e}^{\frac{d_2}{b}(y-\lambda_2(s))}, ~\mathrm{if}~ \vartheta_3(z,s) \!<\! y \!\le\! 2. \label{bar l third part} 
    \end{align}
    The uniqueness of the solution, which relies on the linearity of the kernel equations, is established via a contradiction argument as presented in \cite[Thm.~1]{krstic2008backstepping}.
    
    Analogously, the explicit solution for  $\bar k$   is obtained after a  lengthy computation (see the detailed computations in Appendix~\ref{appendix:derivation of explicit solution for bar k}). Using the change of variables \eqref{change of variables}, we revet to the original functions $l$ and $k$:
    \begin{align}\label{l solution}
        l(z,s,y)\!=\!\left\{\begin{aligned}
        &\!l_1(z,s,y), \qquad\qquad\quad \:~\;\mathrm{if}\: 1\!\le \!y\!\le\! \lambda_2(s), \\
        &\!l_2(z,s,y), \qquad\qquad\quad \: ~\;\mathrm{if}\: \lambda_2(s)\!<\!y\le 2, 
        \end{aligned}\right. 
    \end{align}
    \begin{align}\label{k solution}
        k(z,s,y)\!=\!\left\{\begin{aligned}
        &\!k_1(z,s,y), \qquad\qquad\quad \:\mathrm{if}\: 0\!\le\!y\!\le\! \lambda_1(s), \\
        &\!k_2(z,s,y)\!+\!k_3(z,s,y), \; \mathrm{if}\: \lambda_1(s)\!<\!y\!\le\! 1,
        \end{aligned}\right. 
    \end{align}   
   where
    \begin{align}
        &l_1(z,s,y) = -d_2\varrho_2(s)\delta\bigl(y\!-\!\lambda_2(s)\bigr)H(z - 1) \notag \\
        &\qquad \qquad  \quad -\zeta_2\varrho_2(s)\delta\bigl(y-\lambda_2(s)+2-z\bigr), \\
        &k_1(z,s,y) = -\zeta_1\varrho_1(s)\delta\bigl(y-\lambda_1(s)+1-z\bigr)\notag \\
        &\qquad \qquad ~ \quad -d_1\varrho_1(s)\delta\bigl(y-\lambda_1(s)\bigr), \label{k1}\\
        &k_2(z,s,y) =\bar{\zeta}_2\varrho_2(s)\mathrm{e}^{\frac{c_2-c_1}{a}(y-1)}\delta\bigl(3\!-\!z\!-\!\lambda_2(s)\!-\!\frac{b}{a}(1\!-\!y)\bigr),\label{k2}
    \end{align}     
    and $\varrho_1(s)\!=\!\mathrm{e}^{Dc_1s}$, $\varrho_2(s)\!=\!\mathrm{e}^{Dc_2s}$, $\bar{\zeta}_2=-\frac{b\kappa_1}{a\kappa_2}\zeta_2$. The bounded function $l_2(z,s,y)$ and $k_3(z,s,y)$ are piecewise defined as
    \begin{align}\label{l2}
        l_2(z,s,y)=\left\{\begin{aligned}
            &\varGamma_{11}(z)\varrho_2(s)\mathrm{e}^{\frac{d_2}{b}(y-\lambda_2(s))}, ~\mathrm{if}\,\lambda_2(s)\!\!< \!\! y \!\!\le\! \!\vartheta_3(z,s), \\
            &\varGamma_{12}(z)\varrho_2(s)\mathrm{e}^{\frac{d_2}{b}(y-\lambda_2(s))}, ~\mathrm{if}\, \vartheta_3(z,s) \!<\! y \!\le\! 2.
        \end{aligned}\right.
    \end{align}
    \begin{subequations}\label{k3}
    \begin{align}
        &\mathrm{For}~z\le  1, \notag \\
        &\; k_3(z,s,y)\!=\!\!\left\{\begin{aligned}
            & \left(-\frac{d_1^2}{a}-\frac{d_1}{a}\zeta_1\right)\varrho_1(s)\mathrm{e}^{\frac{d_1}{a}(y-\lambda_1(s))}, \qquad\qquad  \; \\
            &\qquad \qquad\qquad \qquad \qquad \:\mathrm{if}~ \lambda_1(s) \!\!< \!\!y \!\!\le \!\!\vartheta_1(z,s), \\
            & \!-\!\frac{d_1^2}{a}\varrho_1(s)\mathrm{e}^{\frac{d_1}{a}(y-\lambda_1(s))},\,\: \mathrm{if}~ \vartheta_1(z,s)<y\le 1.
        \end{aligned}\right.\\
        &\mathrm{For}~z> 1, \notag \\
        &\; k_3(z,s,y)\!=\!\!\left\{\begin{aligned}
            & \!-\!\frac{d_1^2}{a}\varrho_1(s)\mathrm{e}^{\frac{d_1}{a}(y-\lambda_1(s))},\,\: \mathrm{if}~\lambda_1(s) \!\!< \!\!y \!<\! \vartheta_2(z,s), \\
            & \left(-\frac{d_1^2}{a}+\varGamma_2(z)\right)\varrho_1(s)\mathrm{e}^{\frac{d_1}{a}(y-\lambda_1(s))},  \\&\qquad \qquad\qquad \qquad \qquad \:\mathrm{if}~ \vartheta_2(z,s)\le y\le 1,
        \end{aligned}\right.
    \end{align}   
    \end{subequations}
    where $\varGamma_2(\!z\!)\!\!=\!\!\frac{d_1}{b}\bar\zeta_2\mathrm{e}^{\frac{c_2\!-\!c_1}{b}(z\!-\!1)}$, and $\varGamma_{11}(z)$, $\varGamma_{12}(z)$ are given by \eqref{varGamma11}. Subregion divisions of  $k$  and  $l$ are shown in Fig.~\ref{figure:characteristics}.
    \end{proof}

    \subsection{\texorpdfstring{Well-posedness of the Kernel Equations for $\phi$ and $\mathit{M}$}{Well-posedness of the kernel equation for phi and M}} \label{wellposedness part 2}     
To proceed, we need to prove the boundedness of the first three integrals defined in \eqref{definition of components in trans}.
\begin{lemma}\label{direct norm}
For each $j \in \{v, \alpha, \beta\}$, the integral operator $\mathcal{T}^{j}$ defined in \eqref{definition of components in trans} is bounded, i.e., there exists a constant $C_j > 0$ such that
\begin{align}\label{direct norm estimate}
    \|\mathcal{T}^{j}[f]\|_\infty \leq C_j \|f\|_\infty.
\end{align}
\end{lemma}

The detailed proof is provided in Appendix~\ref{appendix:proofs of lemmas 2 and 4}. Mainly, Dirac delta functions do not affect the boundedness of the transformations, as their integration yields bounded terms and simplifies integrals.

By Lemma~\ref{direct norm}, the operators $\varPi[k,l]$ and $\varPhi[k,l]$, defined in \eqref{Pi} and \eqref{Phi}, respectively, are bounded, since $\varPi[k,l]=\mathcal{T}^{\alpha}[g_1]+\mathcal{T}^{\beta}[g_2]$ and $\varPhi[k,l]=D\mathcal{T}^{\alpha}[h_1]+D\mathcal{T}^{\beta}[h_2]$.

The solution for $\phi$, satisfying \eqref{phi(z,s)}--\eqref{phi(z,0)_2}, can be obtained explicitly:
        \begin{align}
            \phi(z,s) = \mathrm{e}^{A^{\mathrm{T}}Ds}\phi(z,0) \!\!+\!\! \int_0^s \mathrm{e}^{A^{\mathrm{T}}D(s-\varsigma)} \varPhi[k,l](z,\varsigma) \, d\varsigma,
        \end{align}
        where $\varPhi$ is defined in \eqref{Phi}, and $\phi(z,0)$ depends on $M(z,0,y)$, which is explicitly defined in \eqref{M(z,0,y)}.

\begin{lemma}\label{well-posedness M}
   The kernel equation \eqref{M}, with boundary conditions \eqref{M(z,s,1)} and \eqref{M(z,0,y)}, admits a unique solution $M$ on the domain $\mathcal{D}_3$, which satisfies $M(z,s,0)=0$.
\end{lemma}

Before proceeding to the proof of the above lemma, we define 
 $m_1(s,y)$ on $\{(s,y)\in[0,1]^2:0\le y \le \lambda_1(s)\}$, and $m_2$, $\chi_3$, $\chi_4$ on $\{(s,y)\in[0,1]^2: \lambda_1(s)< y \le 1\}$ as
    \begin{align}
   m_1(s,y) =& \frac{D\zeta_1}{-\kappa_1}\int_0^s\mathrm{e}^{Dc_1\varsigma}f(y+Da(s-\varsigma))d\varsigma, \label{m1} \\
     m_2(s,y)\!=&\frac{b\zeta_2D}{-a\kappa_2}\varrho_2(\chi_3(s,y))\int_0^{\frac{1-y}{Da}}\!\mathrm{e}^{Dc_1\varsigma}f(-Da\varsigma +1)d\varsigma,\label{m2}  \\
     \chi_3(s,y)=&s-\frac{1-y}{Da}, \quad \chi_4(s,y)=-Das-y+2. \label{s3, s4} 
    \end{align}

    \begin{proof}
        We employ the method of characteristics to solve for $M$, and obtain the piecewise defined function
        \begin{align}\label{M solution}
        M(z,s,y)\!\!=\!\!\left\{\begin{aligned}
        &\!M_1(z,\!s,\!y), \qquad\qquad\quad \:\:\mathrm{if}\: 0\!\le\!y\!\le\! \lambda_1(s), \\
        &\!M_2(z,\!s,\!y)\!+\!\!M_3(z,\!s,\!y), \;\: \mathrm{if}\: \lambda_1(s)\!<\!y\!\le\! 1.
        \end{aligned}\right. 
    \end{align}
        For the branch corresponding to the gray region in Fig.~\ref{M line}, i.e., $0 \le y \le \lambda_1(s)$, we denote the solution by $M_1$, which is solved as follows:
    \begin{align}\label{M1}
     M_1(z,s,y)\!&=\!\left(m_1(s,y)+m_0(y+Das)H(Da(1\!-\!s)\!-\!y)\right)  \notag\\
             \times\delta\Bigl(y-&\lambda_1(s)+1\!-\!z\Bigr)\!+\!\frac{d_1}{\zeta_1}m_1(s,y)\delta\Bigl(y\!-\!\lambda_1(s)\Bigr),
\end{align}
where we have used $k_1$ expressed in \eqref{k1} and boundary condition \eqref{M(z,0,y)}. To verify the boundary condition $M(z,s,0)=0$ in Remark~\ref{boundary constraint}, we evaluate \eqref{M1} at $y=0$, which yields $M(z,s,0)=0$, where we have used $H(Da(1-s))=1$ and $\delta(\lambda_1(s))=0$. The latter follows from Assumption~\ref{assumption on delay}, which ensures that $\lambda_1(s)=-Das+1>0$ for $s\in[0,1]$, provided that $D<1/a$.

For $M$ with $\lambda_1(s)<y\le 1$, corresponding to the white region in Fig.~\ref{M line}, the solution for $M$ becomes more involved due to the boundary condition $M(z,s,1)$, which contains a self-integral term. By applying $k_2$ and $k_3$ defined in \eqref{k2} and \eqref{k3}, respectively, and using  \eqref{M1}, the solution for $M$ on $\lambda_1(s)<y\le 1$  can be decomposed into two components, $M_2$ and $M_3$, which yields 
\begin{align}
M_2 + M_3 &= Y_2(z,s,y)+Y_3(z,s,y) + \varPsi [M_2 + M_3]
, \label{M integral no split}
\end{align}
where
\begin{align}
 Y_2(z,s,y) =&~ \Theta(z,\chi_3(s,y))+\int_0^{\frac{1-y}{Da}}\frac{D}{\kappa_1}f(-Da\xi+1)\notag \\
        &\times k_3(z,\xi+\chi_3(s,y),-Da\xi+1)d\xi,  \label{Y2}\\
        Y_3(z,s,y)=&~ m_2(s,y)\delta\bigl(3-z-\lambda_2(s)-\frac{b}{a}(1-y)\bigr),   \label{Y3}\\
\varPsi [M]
 = &\int_{\chi_4(s,y)}^{1} \frac{\theta}{a}M(z,\chi_3(s,y),\xi)d\xi, \label{operator varPsi}
\end{align}
with $\Theta(z,s)$ defined piecewise as
\begin{subequations}
\begin{align}
    \Theta(z,s)=&\Theta_1
+ \dfrac{\theta d_1}{a\zeta_1} m_1\!\left(s,\lambda_1(s)\right)
+ \dfrac{\theta}{a}\!\bigl(m_1\!\left(s,z-Das\right) \notag \\
&+m_0(z)\bigr) H(Da-z),\quad
\text{if } Das \le z \le 1,\\
\Theta(z,s)=&\Theta_1
+ \dfrac{\theta d_1}{a\zeta_1} m_1\!\left(s,\lambda_1(s)\right),
\quad  \quad\text{otherwise},
\end{align}
\end{subequations}
and $\Theta_1(z,s)=\bigl(\varPi[k,l](z,s)+\phi(z,s)^{\mathrm T}B\bigr)/a$.

We consider $M_i$  separately: 
 \begin{align}
        M_i(\!z,s,y\!)\!=&Y_{i}(z,s,y)+\varPsi[M_i](\!z,s,y\!), \,i=2,3.\label{Mi integral}
\end{align}
Accordingly, the iterative form is given by
 \begin{align}
        M^{n+1}_i(\!z,s,y\!)\!=&Y_{i}(z,s,y)+\varPsi[M_i^n](\!z,s,y\!), \,i=2,3. 
\end{align}
For $n=0,1,2,\ldots,$ 
let $M_i^0=0$ and $\Delta M_i^n(z,s,y)=M_i^{n+1}(z,s,y)-M_i^{n}(z,s,y)$, which gives  
    \begin{align}
    \Delta M_i^0(z,s,y)=&Y_i(z,s,y),\label{Delta M_0}\\
        \Delta M_i^{n+1}(z,s,y)=&\varPsi[\Delta M_i^n](z,s,y). \label{Delta M n+1 and n relation}
    \end{align}
Assume that for $n\ge 0$, 
    \begin{align}
        \lvert \Delta M_2^n(z,s,y) \rvert \le \bar{Y}_2\frac{1}{n!}\left(\frac{| \theta |}{a}(y+Das-1)\right)^n, \label{Delta M n inequality}
    \end{align}
    where $\bar{Y}_2=\sup\lvert Y_2(z,s,y)\rvert$. Substituting \eqref{Delta M n inequality} into \eqref{Delta M n+1 and n relation} gives
    \begin{align*}
        &\lvert \Delta M_2^{n+1}(z,s,y) \rvert  \\
        \le& \frac{| \theta |}{a}\!\int_{\chi_4(s,y)}^1\!\bar{Y}_2 \frac{1}{n!}\left(\frac{| \theta |}{a}(\xi+Da\chi_3(s,y)-1)\right)^n  d\xi \notag \\
        \le& \bar{Y}_2\frac{1}{(n+1)!}\left(\frac{|\theta|}{a}(y+Das-1)\right)^{n+1}.        
    \end{align*}
    Hence, \eqref{Delta M n inequality} holds, which yields    
    \begin{align}
        | M_2(z,s,y) | \le \bar{Y}_2\mathrm{e}^{\frac{\lvert \theta \rvert}{a}(y+Das-1)}. \label{M2 bound}
    \end{align}
 
 The delta function in   $Y_3$  vanishes upon applying the integral operator \eqref{operator varPsi}, resulting in
    \begin{subequations}
    \begin{align}
               &\Delta M_3^1(z,s,y)= 0,\qquad \qquad \:\mathrm{if}~ \lambda_1(s)\!<\! y< \vartheta_4(z,s), \\
           &\Delta M_3^1(z,s,y)= \frac{\theta}{b}m_2(\chi_3(s,y),\xi_1(z,s,y))H(z-1), \notag \\
                &\qquad \qquad\qquad\qquad\qquad\qquad\mathrm{if}~\vartheta_4(z,s) \le y\le 1, 
    \end{align}
    \end{subequations}
    where $\xi_1(z,s,y)=\lambda_1(s)+\frac{a}{b}(z-1)-y+1$.
Applying a similar mathematical induction, we can prove for $n\ge 1$ that 
    \begin{align*}
        \left| \Delta M_3^n(z,s,y) \right| \!\le\! \frac{\lvert \theta \rvert \bar{m}_2}{b}\frac{1}{(n\!-\!1)!}\left(\frac{\lvert \theta \rvert}{a}(y\!+\!Das\!-\!1)\right)^{n-1},
    \end{align*}
    where $\bar{m}_2\!=\!\sup \lvert m_2(s,y)\rvert$.
    Consequently, we have 
     \begin{align}
         \left|\sum\limits_{n=1}^{\infty}\Delta M_3^n(z,s,y)\right| \le \frac{\lvert \theta \rvert \bar{m}_2}{b}\mathrm{e}^{\frac{\lvert \theta \rvert}{a}(y+Das-1)}.  \label{M4 bound} 
    \end{align}
    Combining \eqref{Delta M_0}, \eqref{Mi integral} and \eqref{M4 bound}, we obtain 
    \begin{align}
        M_3(z,s,y)\!=& Y_3(z,s,y)+\sum\limits_{n=1}^{\infty}\Delta M_3^n(z,s,y). \label{M sum expression}
    \end{align}
   Note that the second term of $M_3$ is bounded; hence, $M_2+M_3$ can be decomposed into a Dirac delta function and a bounded component. 
    This ends the proof.
    \end{proof}
    
    \begin{lemma}
           \label{direct norm u}    The operator $\mathcal{T}^u$ defined in \eqref{definition of components in trans} is bounded, that is, there exists a positive constant $C_u$ such that
\begin{align}\label{T_u}
 \|\mathcal{T}^{u}[f]\|_\infty \leq C_u \|f\|_\infty.
\end{align}
   \end{lemma}

\subsection{Well-posedness of the Inverse Kernel Equations}

The ODE \eqref{breve phi} for the inverse kernel function $\breve \phi$ can be expressed by its boundary condition $\breve \phi(z,0)$:  
\begin{align}
    \breve\phi(z,s) =& \mathrm{e}^{(A+BK)^{\mathrm{T}}Ds}\breve \phi(z,0).
\end{align}

\begin{lemma}\label{lemma: inverse kernel solution}
    The system of inverse kernel equations \eqref{breve eta q}--\eqref{breve N equation} admits a unique solution on the domain $\mathcal{D}_1\!\times \mathcal{D}_2 \!\times \mathcal{D}_4 \!\times \mathcal{D}_3^{3}$.
\end{lemma}
\begin{proof}
To proceed, we first define the following functions:
\begin{align}
   \breve{m}_1(s,y) =& \frac{D\zeta_1}{-\kappa_1}\int_0^s\mathrm{e}^{-D\bar{c}_1\varsigma}\breve{f}(y+Da(s-\varsigma))d\varsigma, \\
   \breve m_2(s,y)=&\frac{b\zeta_2D}{-a\kappa_2}\breve\varrho_2(\chi_3(s,y))\int_0^{\frac{1-y}{Da}}\!\mathrm{e}^{-D\bar c_1\varsigma}\breve f(-Da\varsigma +1)d\varsigma, \\
   \breve{\varrho}_1(s)=&\mathrm{e}^{-D\bar{c}_1s},  \qquad \breve{\varrho}_2(s)=\mathrm{e}^{-D\bar{c}_2s}. 
\end{align}
Due to space limitations, the detailed proof of well-posedness is omitted, and only the key steps and main differences from the direct kernel equations are summarized.

First, the equations for $\breve{\eta}$, $\breve{q}$, $\breve{k}$, and $\breve{l}$ are simplified versions of those for $\eta$, $q$, $k$, and $l$, respectively; hence, their well-posedness can be established in a straightforward manner.
The main differences arise from the presence of an additional kernel $\breve{N}$ in the inverse equations and an extra integral term, denoted by $\breve{\varUpsilon}$, in the governing equation for $\breve{M}$.

As $\breve{\varUpsilon}$ depends on $\breve{k}$, we first present
     \begin{align}\label{breve k}
        &\breve{k}(z,s,y)\!=\!\!\left\{\begin{aligned}
        &\!\breve{k}_1(z,s,y), &&\mathrm{if}\: 0\!\le\!y\!\le\! \lambda_1(s), \\
        &\!\breve{k}_2(z,s,y)\!\!+\!\!\breve{k}_3(z,s,y), && \mathrm{if}\: \lambda_1(s)\!<\!y\!\le\! 1,
        \end{aligned}\right.
        \end{align}
        with
        \begin{align*}
            &\breve{k}_1(z,s,y) = -\zeta_1\breve{\varrho}_1(s)\delta\bigl(y-\lambda_1(s)+1-z\bigr)\notag \\
        &\qquad \qquad ~ \quad -d_1\breve{\varrho}_1(s)\delta\bigl(y-\lambda_1(s)\bigr), \\
        &\breve{k}_2(z,s,y) =\bar{\zeta}_2\breve{\varrho}_2(s)\mathrm{e}^{\frac{\bar{c}_1-\bar{c}_2}{a}(y-1)}\delta\bigl(3\!-\!z\!-\!\lambda_2(s)\!-\!\frac{b}{a}(1\!-\!y)\bigr),
        \end{align*}
        and the bounded function $\breve k_3(z,s,y)$ is piecewise defined as
        \begin{align*}
         \breve k_3(z,s,y)=\left\{\begin{aligned}
             &\left(\varGamma_3(z)+\frac{d_1\zeta_2}{a}\right)\breve{\varrho}_2(s)\mathrm{e}^{\frac{\bar{c}_1-\bar{c}_2}{a}(y-1)}, \\&\qquad\qquad\qquad\qquad\qquad \,\mathrm{if}~ \lambda_1(s)\!\!<\!\!y\!\!\le \!\!\vartheta_4(z,s), \\
             & \varGamma_3(z)\breve{\varrho}_2(s)\mathrm{e}^{\frac{\bar{c}_1-\bar{c}_2}{a}(y-1)},~ \mathrm{if}~ \vartheta_4(z,s) < y\le 1,
             \end{aligned} \right.
        \end{align*}
        where $\varGamma_3(z)=\frac{d_1d_2}{a}H(z-1)$.

  We now consider $\breve M$ in the region $0\le y\le\lambda_1(s)$, where we denote its solution by $\breve M_1$. We obtain
    \begin{align}\label{breve M1}
        \breve{M}_1(z,s,y)\!&=\!\varrho(z,y\!+\!Das)\!+\!\int_0^s\breve{\varUpsilon}[\breve k](z,\xi,y\!+\!Da(s\!-\!\xi))d\xi  \notag\\
             \!+\breve{m}_1(s,y&)\Bigl(\delta(y\!-\!\lambda_1(s)\!+\!1\!-\!z)\!+\!\frac{d_1}{\zeta_1}\delta(y\!-\!\lambda_1(s))\Bigr),
    \end{align}
   where $\varrho$ is defined in \eqref{varrho} and
     \begin{align*}
\breve{\varUpsilon}[\breve k] =\left\{ 
\begin{aligned}
&0, \qquad\qquad\qquad\qquad\qquad\mathrm{if}~0 < y < \vartheta_5(z,s), \\
&\begin{aligned}
&\frac{\zeta_1 D}{-\kappa_1} \breve{\varrho}_1(s) \breve{f}(z \!-\! D a s)\breve{\eta}(z \!-\! D a s, y) H(z\!-\! D a s),
\end{aligned} \\
&\qquad\qquad\qquad\qquad\qquad \quad \mathrm{if}~\vartheta_5(z,s) \!\le\! y \!\le\! \lambda_1(s),
\end{aligned}\right.
\end{align*}
 For $\lambda_1(s)< y\le1$, we rewrite $\breve M$ as
    \begin{align}
        \breve{M}(z,s,y)=\breve{M}_2(z,s,y)+\breve M_3(z,s,y),
    \end{align}
    where
    \begin{align*}
        \breve M_2(z,s,y)=& \breve\Theta(z,\chi_3(s,y))+\int_0^{\frac{1-y}{Da}}\frac{D}{\kappa_1}\breve f(-Da\xi+1)\notag \\
        &\times \breve k_3(z,\xi+\chi_3(s,y),-Da\xi+1)d\xi \notag\\
        &+ \int_0^{\frac{1-y}{Da}}\breve \varUpsilon[\breve k](z,\xi+\chi_3(s,y),-Da\xi+1)d\xi,\\
        \breve M_3(z,s,y)=& \breve m_2(s,y)\delta\bigl(3-z-\lambda_2(s)-\frac{b}{a}(1-y)\bigr),
    \end{align*}
    and $\breve \Theta(z,s)=\frac{1}{a}\left(\breve\varPi[\breve k,\breve l](z,s)\!+\!\breve\phi(z,s)^{\mathrm{T}}B\right)$. Recall that $\breve f$ and $\breve \varPi[\breve k, \breve l]$  are defined in \eqref{breve f} and \eqref{breve varPi}, respectively. We can also prove that $\breve{\varUpsilon}[\breve k]$ is bounded for $y>\lambda_1(s)$ through substituting $\breve k$, $\breve f$ and $\breve \eta$ into \eqref{breve varUpsilon}.    
    Since the right-hand side of $\breve M_2$ only involves known and bounded functions, 
it follows directly that $\breve M_2$ is bounded. 

Therefore, the equation for $\breve M$ is well-posed.

    Using $\breve{M}(z,s,0)$ from \eqref{breve M1},  we obtain 
    \begin{align*}
        &\breve{N}(z,s,y)=\left\{\begin{aligned}
        &0, & &\mathrm{if}~0\le y\le \lambda_3(s),\\
        &\tau a\breve{M}_1(z,\chi_5(s,y),0),& & \mathrm{if}~ \lambda_3(s)<y\le 1,
        \end{aligned}\right.
    \end{align*}
    where $\chi_5(s,y)=s-\frac{\tau}{D}(1-y)$.  Hence, we have established the well-posedness of the inverse kernel system \eqref{breve eta q}--\eqref{breve N equation}.   
\end{proof}

\section{Stability Analysis}\label{stability}
    In this section, we will prove the main result stated in Theorem~\ref{main result}. The proof is organized into two parts. First, the stability of the target system is proved in Section~\ref{section target system stability}. Second, the norm equivalence between the original controlled system and the target system is demonstrated in Section~\ref{section norm equivalence}.

\subsection{Stability of the Target System}\label{section target system stability}    
 \begin{proposition}\label{stability of target system}
Consider the target system \eqref{target system}, with initial conditions $(w_0, \mu_0, \alpha_0, \beta_0, \gamma_0, X_0)\in \mathcal{S}_c$ satisfying the corresponding compatibility conditions. Then, the system admits a unique weak solution $(w, \mu, \alpha, \beta, \gamma, X) \in C([0,\infty); \mathcal{S})$ and its zero equilibrium is exponentially stable in $\mathcal{S}$, i.e., there exist constants $M_0, \upsilon_0>0$ such that
  \begin{align}
    V_2(t) \le M_0 \mathrm{e}^{-\upsilon_0 t} V_2(0), \label{V0 inequality}
\end{align}
where  
\begin{align}
    V_2(t) = &\| w \|_\infty^2+ \| \mu \|_\infty^2 + \| \alpha \|_\infty^2 + \| \beta \|_\infty^2  + \| \gamma \|_\infty^2 + |X|^2.\label{V0}
    \end{align}
\end{proposition}

\begin{proof}
First, we consider the well-posedness of the target system. Since the ODE component is standard, the analysis focuses on the PDE subsystem. 
    Applying the method of characteristics, we obtain
    \begin{align}
        &\gamma(z,s,t)\!=\!\gamma_0(z,s\!+\!t/D),\; t \!\le\! D(1 \!- \!s); \gamma\!=\!0,\; t \!>\! D(1\!- \! s),\label{explicit_gamma}\\
        &\mu(s,t)\!=\!\mu_0(s \!+\! t/\tau),\; t \!\le\! \tau(1\!- \! s); \mu\!=\!0,\; t \!>\! \tau(1 \!- \! s),\label{mu expression}\\
        &w(z,t)\!=\!w_0(z \!- \! at), \: t \le \frac{z}{a};w(z,t)\!=\!\mu(0, t \!- \! z/a), \; t > \frac{z}{a}.\label{w expression}
    \end{align}
     Given the initial conditions, the above explicit solution representations imply that $\gamma, \mu, w \in C([0, \infty); L^\infty)$. Employing the method of characteristic, we also derive the explicit solutions for $\alpha$ and $\beta$:
    \begin{align}
        \alpha(z,t)=& \mathrm{e}^{-\bar c_1t}\Bigl(\alpha_0(z-at)\!+\!\!\int_0^t\mathrm{e}^{\bar c_1\xi}\bigl[\kappa_1\gamma(z-a(t-\xi),0,\xi)\notag \\
        &\quad+F_1[w](z-a(t-\xi),\xi)\bigr]d\xi\Bigr),\quad t\le \frac{z}{a},\label{alpha expression 1}\\
        \alpha(z,t)=&\mathrm{e}^{\frac{-\bar c_1}{a}z}\int_0^{\frac{z}{a}}\mathrm{e}^{\bar c_1\xi}\bigl[\kappa_1\gamma(a\xi,0,\xi+t-\frac{z}{a})\notag \\
        &\quad+F_1[w](a\xi,\xi+t-\frac{z}{a})\bigr]d\xi, \:\qquad t> \frac{z}{a},\label{alpha expression 2}\\
        \beta(z,t)=&\mathrm{e}^{-\bar c_2t}\Bigl(\int_0^t\mathrm{e}^{\bar c_2\xi}\bigl[\frac{\kappa_2d_1}{-\kappa_1}\alpha(1,\xi)+\kappa_2\gamma(z_1+b\xi,0,\xi)\notag \\
        +F_2[&w](z_1+b\xi,\xi)\bigr]d\xi+\beta_0(z_1)\Bigr),\quad t\le \frac{z-1}{b},\label{beta expression 1}\\
        \beta(z,t)=&\mathrm{e}^{\frac{-\bar c_2}{b}(z\!-\!1)}\Bigl(\alpha(1,z_2)\!+\!\int_0^{\frac{z-1}{b}}\mathrm{e}^{\bar c_2\xi}\bigl[\frac{\kappa_2d_1}{-\kappa_1}\alpha(1,\xi+z_2)\notag \\
        +\kappa_2\gamma&(b\xi\!+\!1,0,\xi\!+\!z_2)\!+\!F_2[w](b\xi\!+\!1,\xi\!+\!z_2)\bigr]d\xi\Bigr), \notag \\
        &\qquad \qquad \qquad \qquad\qquad\qquad\quad ~t> \frac{z-1}{b}, \label{beta expression 2}
    \end{align}
    where $z_1(z,t)=z-bt$, $z_2(z,t)=t-\frac{z-1}{b}$. 
    
    From \eqref{F1 rewrite} and \eqref{F2 rewrite}, we get  
    \begin{align}
        |F_1[w]|\le R_1\|w\|_{\infty},\quad
        |F_2[w]|\le R_2\|w\|_{\infty},
    \end{align}
    where  $R_1, R_2>0$.  

    Since $F_1$ and $F_2$ are bounded on the admissible $L^\infty$ states, all the integrals in \eqref{alpha expression 1}--\eqref{beta expression 2} are well defined. Together with the continuity of the initial conditions and the boundary compatibility conditions, the characteristic representations imply that  $\alpha, \beta \in C([0, \infty); L^\infty)$. Therefore, the system \eqref{target system} admits a unique weak solution in $C([0,\infty);\mathcal{S})$.

 Next, we consider the stability of the target system. From \eqref{explicit_gamma}--\eqref{w expression}, we have 
      \begin{align}
        \|\mu\|_{\infty} \!\le& \|\mu_0\|_\infty, \qquad \qquad \|\gamma\|_{\infty} \!\le \|\gamma_0\|_\infty, \label{mu gamma norm estimate}\\ 
        \|w\|_{\infty} \!\le& \max\{\|w_0\|_{\infty},\|\mu_0\|_{\infty}\}\!\le\! \|w_0\|_{\infty}\!+\!\|\mu_0\|_{\infty}. \label{w norm estimate}
    \end{align}

Additionally, from \eqref{alpha expression 1} and \eqref{alpha expression 2}, it follows that
\begin{align}
        \|\alpha\|_{\infty}&\le \|\alpha_0\|_{\infty}+\int_0^{\frac{1}{a}}\mathrm{e}^{\frac{\bar c_1}{a}}(|\kappa_1|\|\gamma\|_{\infty}+R_1\|w\|_{\infty})d\xi, \notag \\
        &\le \|\alpha_0\|_{\infty}\!+\! P_1(\|\gamma_0\|_{\infty}\!+\! \|w_0\|_{\infty}\!+\!\|\mu_0\|_{\infty}),\label{alpha norm estimate}
    \end{align}
where
$P_1\!\!=\!\!\max\{\frac{|\kappa_1|}{a}\mathrm{e}^{\frac{\bar c_1}{a}},\!
\frac{R_1}{a}\mathrm{e}^{\frac{\bar c_1}{a}}\}\!\!>\!\!0$.
Similarly, it follows from \eqref{beta expression 1} and \eqref{beta expression 2} that there exists a constant $P_2>0$ such that
\begin{align}
\|\beta\|_{\infty}
\le \|\beta_0\|_{\infty}
\!+\! P_2(\|\alpha_0\|_{\infty}
+ \|\gamma_0\|_{\infty}
+ \|w_0\|_{\infty}
+ \|\mu_0\|_{\infty}).
\label{beta norm estimate}
\end{align}

Combining \eqref{mu gamma norm estimate}--\eqref{beta norm estimate}, one can verify that there exist constants $P_3>0$ and $\upsilon_1>0$ such that
\begin{align}
        &\| w\|_{\infty}^2 \!+\! \| \mu\|_{\infty}^2 \!+\! \| \alpha \|_{\infty}^2 
    + \| \beta \|_{\infty}^2  + \| \gamma \|_{\infty}^2 \notag \\
    \le& P_3\mathrm{e}^{\upsilon_1t_\mathrm{F}}\mathrm{e}^{-\upsilon_1t}(\| w_0\|_{\infty}^2 \!+\! \| \mu_0\|_{\infty}^2 \!+\! \| \alpha_0 \|_{\infty}^2 
    + \notag \\
    &\qquad\qquad  \| \beta_0 \|_{\infty}^2  + \| \gamma_0 \|_{\infty}^2),\label{no X inequality}
    \end{align}
where $t_F=D+\tau+\frac{2}{a}+\frac{1}{b}$ denotes the finite time
after which the PDE subsystem becomes identically zero.

The solution of the ODE \eqref{target 4} is given by
\begin{align}
X(t)=\mathrm{e}^{(A+BK)t}X_0\!+\!\int_0^t \mathrm{e}^{(A+BK)(t-\varsigma)}B w(1,\varsigma)\,d\varsigma.
\end{align}
Since $A+BK$ is Hurwitz, there exist constants $M_1\ge1$ and
$\upsilon_2>0$ such that
$|\mathrm{e}^{(A+BK)t}|\le M_1 \mathrm{e}^{-\upsilon_2 t}$.
Hence,
\begin{align}
|X(t)|
\le M_1 \mathrm{e}^{-\upsilon_2 t}|X_0|
+ \int_0^t M_1 \mathrm{e}^{-\upsilon_2 (t-\varsigma)}
|B||w(1,\varsigma)|\,d\varsigma.
\end{align}
Using the estimate
$|w(1,t)|\le \|w(\cdot,t)\|_{\infty}
\le M_2 \mathrm{e}^{-\upsilon_3 t}
(\|w_0\|_{\infty}+\|\mu_0\|_{\infty})$,
with $M_2=\mathrm{e}^{(\tau+\frac{1}{a})\upsilon_3}$
and any $\upsilon_3>\upsilon_2$,
we obtain
\begin{align}
|X(t)|
\le M_3 \mathrm{e}^{-\upsilon_2 t}
(\|w_0\|_{\infty}
+ \|\mu_0\|_{\infty}
+ |X_0|),
\label{X inequality}
\end{align}
where
$M_3=\max\{M_1,\frac{M_1M_2|B|}{\upsilon_3-\upsilon_2}\}$.
Combining \eqref{no X inequality} and \eqref{X inequality}
yields \eqref{V0 inequality} with
$\upsilon_0=\min\{\upsilon_1,2\upsilon_2\}$.
\end{proof}

\subsection{Norm Equivalence} \label{section norm equivalence}    \begin{lemma}\label{inverse norm}  
    For each $j \in \{\gamma, \alpha, \beta, w, \mu\}$, the integral operator $\mathcal{\breve T}^{j}[f]$ defined in \eqref{definition of components in inverse trans} is bounded, i.e., there exists a constant $\breve C_j > 0$ such that
        \begin{align}\label{inverse norm estimate}
            \|\mathcal{\breve T}^{j}[f]\|_{\infty} \leq \breve C_j \|f\|_{\infty}.
        \end{align}
\end{lemma}
   The proof, analogous to that of Lemma~\ref{direct norm}, is omitted here for brevity. It can be found in Appendix~\ref{appendix:proof of lemma 6}.
    \begin{proposition}(Norm equivalence) \label{norm equivalance}
        The Lyapunov functions $V_1(t)$ and $V_2(t)$ defined in \eqref{infty norm of the original states} and \eqref{V0}, respectively, are norm-equivalent on $\mathcal{S}$, i.e., there exist $\varsigma_1>0$, $\varsigma_2>0$ such that
        \begin{align}
        \varsigma_1 V_1(t) \le V_2(t) \le \varsigma_2 V_1(t). \label{V0 L relation} 
        \end{align}
    \end{proposition} 
    \begin{proof}
        It follows directly from the transformations \eqref{u trans} and \eqref{r trans} that there exist positive constants $\theta_1$, $\theta_2$ such that
        \begin{align}
            \| w\|_{\infty}^2 \le&  \theta_1 (\| u\|_{\infty}^2+\lvert X \rvert^2), \label{w estimate} \\
         \| \mu \|_{\infty}^2 \le& \theta_2(\| r\|_{\infty}^2+\| u\|_{\infty}^2+\lvert X\rvert^2).\label{mu estimate}
        \end{align}
        
       Applying Lemmas~\ref{direct norm} and~\ref{direct norm u} to the transformation \eqref{v trans} yields
        \begin{align}
            \| \gamma\|_{\infty}^2
            \le& 6\|v\|_{\infty}^2+6\varSigma_{j\in \{v,\alpha,\beta,u\} }\|\mathcal{T}^{j}[j]\|_{\infty}^2 + 6\bar{\phi}^2\lvert X\rvert^2\notag \\
            \le& 6(1+C_v^2)\| v\|_{\infty}^2+6C_{\alpha}^2\| \alpha\|_{\infty}^2+6C_{\beta}^2\| \beta\|_{\infty}^2 \notag \\
            &+6C_u^2\| u \|_{\infty}^2+6\bar{\phi}^2\lvert X\rvert^2,
             \label{gamma norm estimate}
        \end{align}
        where $\bar\phi=\sup_{(z,s)\in[0,2]\times[0,1]}|\phi(z,s)|$. The right-hand inequality in \eqref{V0 L relation} follows from \eqref{w estimate}, \eqref{mu estimate}, and \eqref{gamma norm estimate}. The left-hand inequality is obtained via the inverse transformations \eqref{w trans}--\eqref{gamma trans} and Lemma~\ref{inverse norm}.
        \end{proof}

Since the transformations \eqref{u trans}--\eqref{v trans} and their inverses \eqref{w trans}--\eqref{gamma trans} are bounded operators on
$L^\infty$, they establish an isomorphism between the state spaces of
the original and target systems. Therefore, the well-posedness of the
target system implies that the closed-loop original system admits a unique
weak solution in $C([0,\infty);\mathcal{S})$.

    Theorem~\ref{main result} follows directly from Propositions~\ref{stability of target system} and~\ref{norm equivalance}.
    \begin{remark}
        The robustness against delay mismatches can be investigated by applying the backstepping transformations constructed for the nominal delays and analyzing the resulting perturbed target system via a suitable Lyapunov functional. In-domain delay mismatches are more challenging due to the Dirac-delta-type terms in the Fredholm transformation, which may require higher-order spatial regularity. Alternatively, filtering-based robustification approaches can be employed to improve robustness against delays and uncertainties~\cite{auriol2023robustification}.
    \end{remark}

  \begin{figure}[pb]
		\begin{center}
			\subfigure{\includegraphics[width=0.85\linewidth]{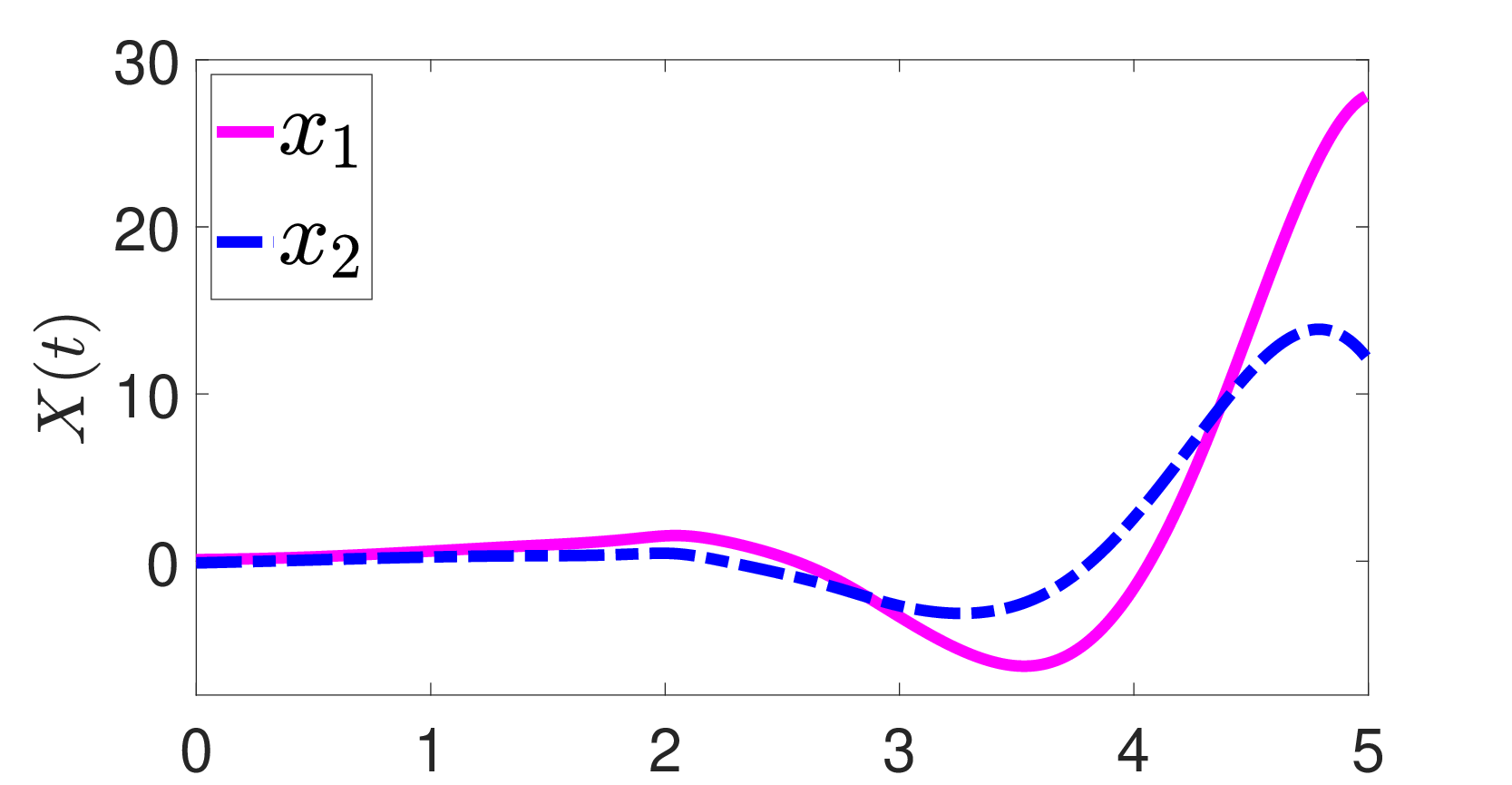}
			}
            
            \vspace{-12pt}
			\subfigure{\includegraphics[width=0.85\linewidth]{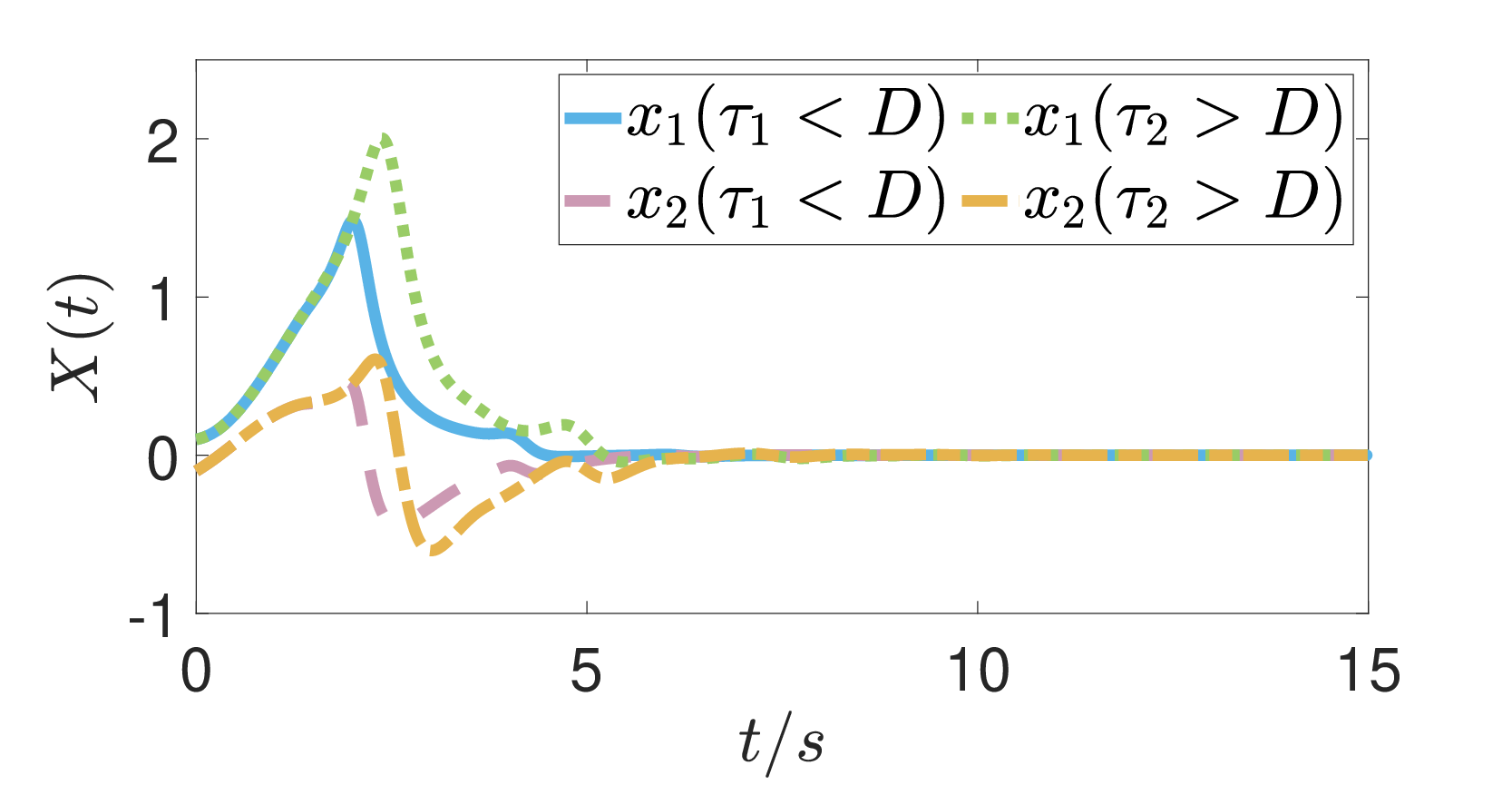}
			}
		\end{center}
        \vspace{-15pt}
        \caption{Top: State $X(t)$ without delay compensation for $\tau=0.8$ and $D=1$. Bottom: State $X(t)$ under the respective delay compensation for two cases: $(\tau_1,D)=(0.8,1)$ and $(\tau_2,D)=(1.2,1)$.}
        \label{time evolution of ODE}
    \end{figure}

\section{Numerical Simulation}\label{simulation}
    To validate  the theoretical results, we present two numerical examples. The first is an unstable linear system with a 2D ODE, and the second is a screw extrusion model for the melt spinning process, which is linearized around its equilibrium point (operating point).

    \subsection{Unstable Linear System}
    The coefficients are set as follows: $a=0.8$, $b=0.4$, $A=[1,0; 1, -2]$, $B=[1, 1]^{\mathrm{T}}$, $c_1=-0.5$, $c_2=0.5$, $d_1=2$, $d_2=1.5$, $f(z)=\mathrm{e}^z$, $g_1(z)=2z$, $h_1(z)=[\mathrm{e}^{-z}, -\mathrm{e}^{2z}]^{\mathrm{T}}$, $g_2(z)=5(z-1)$, $h_2(z)=[\mathrm{e}^{z-1}, \mathrm{e}^{2-z}]^{\mathrm{T}}$, $\theta=0$, $\kappa_1=2$, $\kappa_2=4$, and $D=1$. Note that we perform the simulation under two distinct boundary input delays: $\tau_1=0.8$ and $\tau_2=1.2$. We set the control gain to $K=[-4, -1]$ and the design parameters in the target system to $\zeta_1=3$ and $\zeta_2=5$ (corresponding to $\bar c_1$ and $\bar c_2$, respectively).
    
    The coupled PDE-ODE system is solved numerically using the finite difference method with temporal and spatial step sizes $\Delta t=0.01$, $\Delta z=0.01$, and $\Delta s=0.02$. The initial conditions are chosen as $X_0=[0.1, -0.1]^\mathrm{T}$, $u_0(z)=0.3\sin(\pi z)$, $\alpha_0(z)=3\sin(2\pi z)$, $\beta_0(z)=3\sin(2\pi z)$, $r_0(s)=0.3\cos(\frac{3\pi}{2}s)$, and $v_0(z,s)=2\sin(\pi zs)$.
    \begin{figure*}[htp]
		\begin{center}
            \subfigure[]{\includegraphics[width=0.30\linewidth]{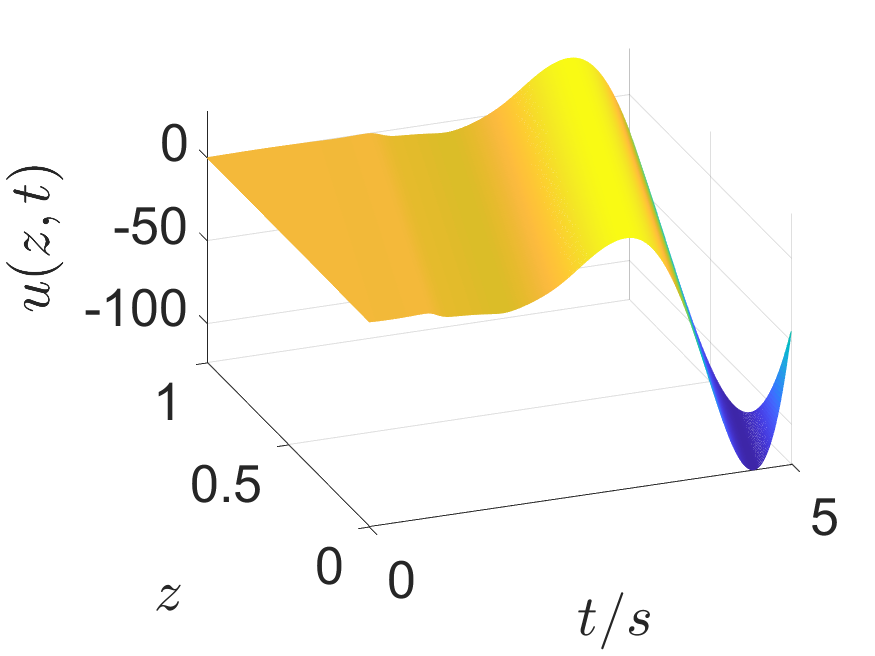}\label{u evolution uncompensated}
			}
			\subfigure[]{\includegraphics[width=0.30\linewidth]{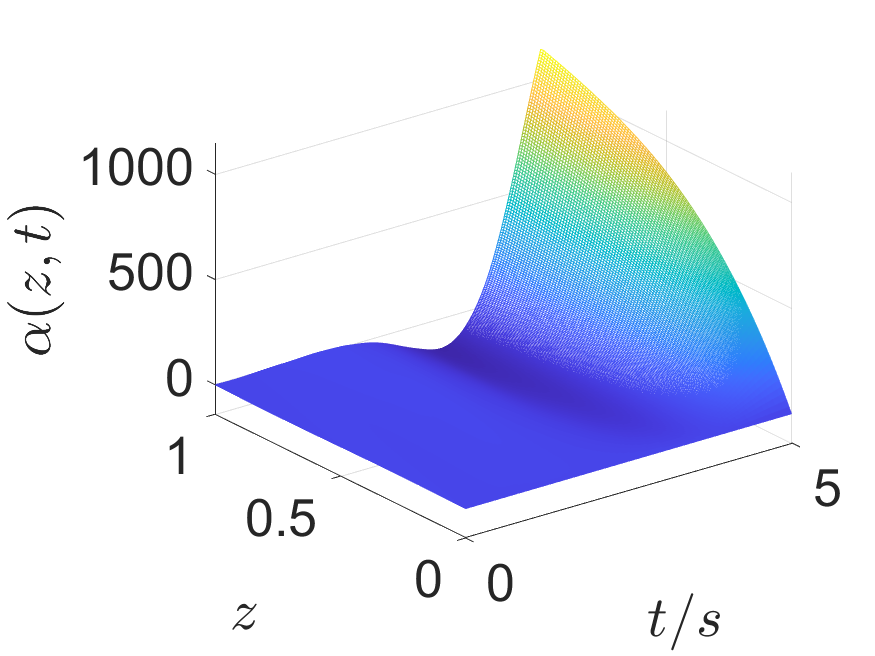}\label{alpha evolution uncompensated}
			}
			\subfigure[]{\includegraphics[width=0.30\linewidth]{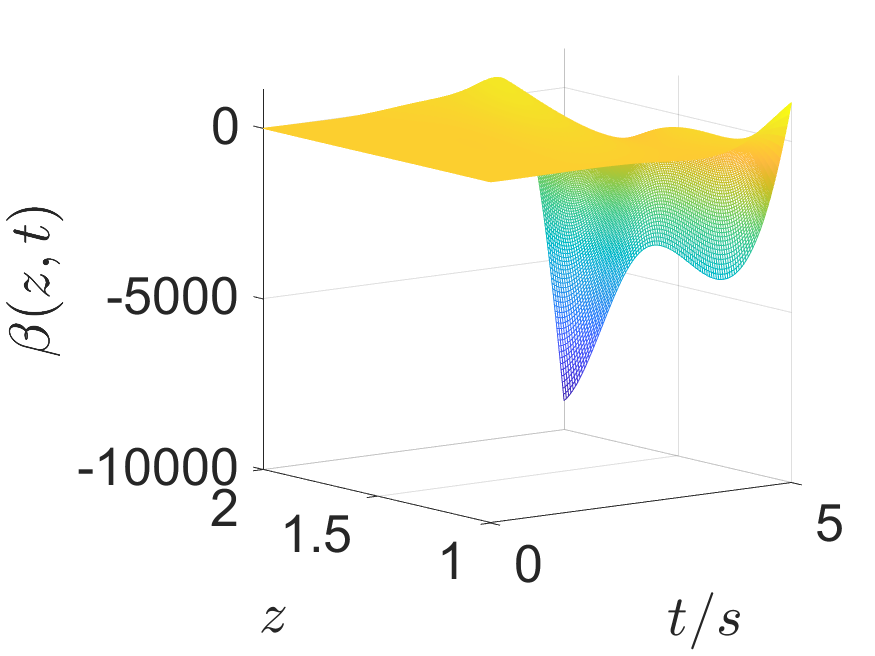}\label{beta evolution uncompensated}
			}
            \vspace{-10pt}
            
			\subfigure[]{\includegraphics[width=0.30\linewidth]{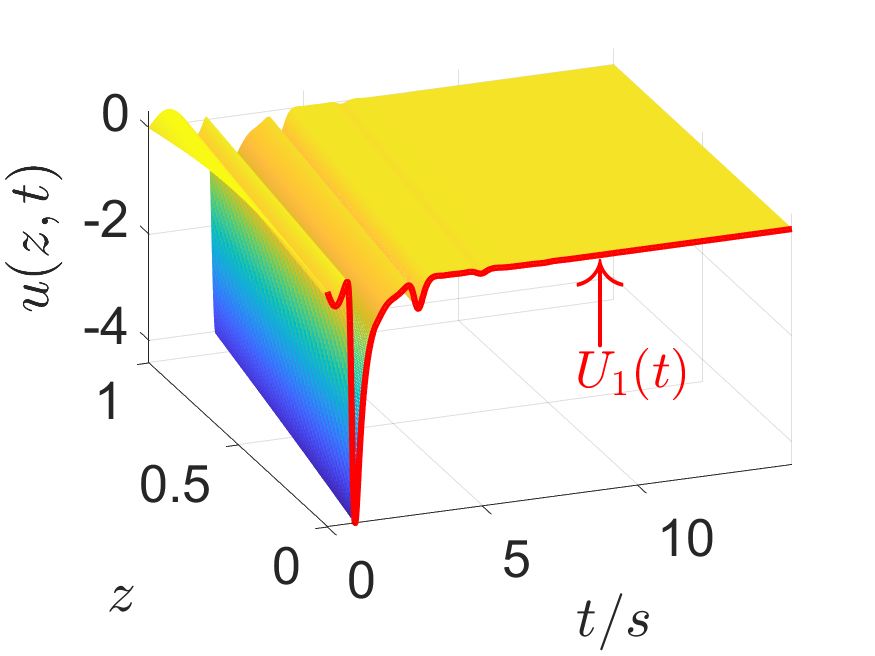}\label{u evolution}
			}
			\subfigure[]{\includegraphics[width=0.30\linewidth]{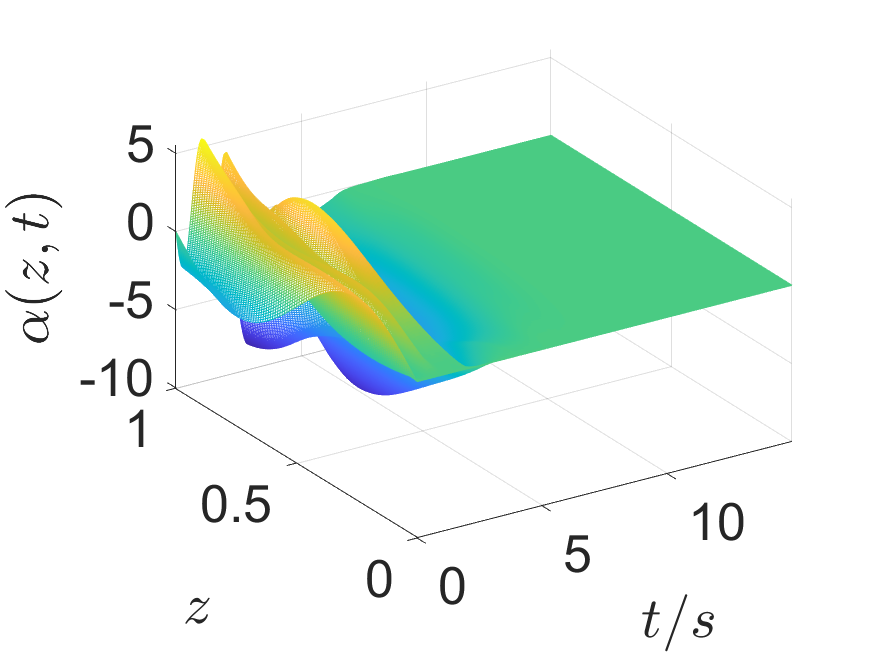}\label{alpha evolution}
			}
			\subfigure[]{\includegraphics[width=0.30\linewidth]{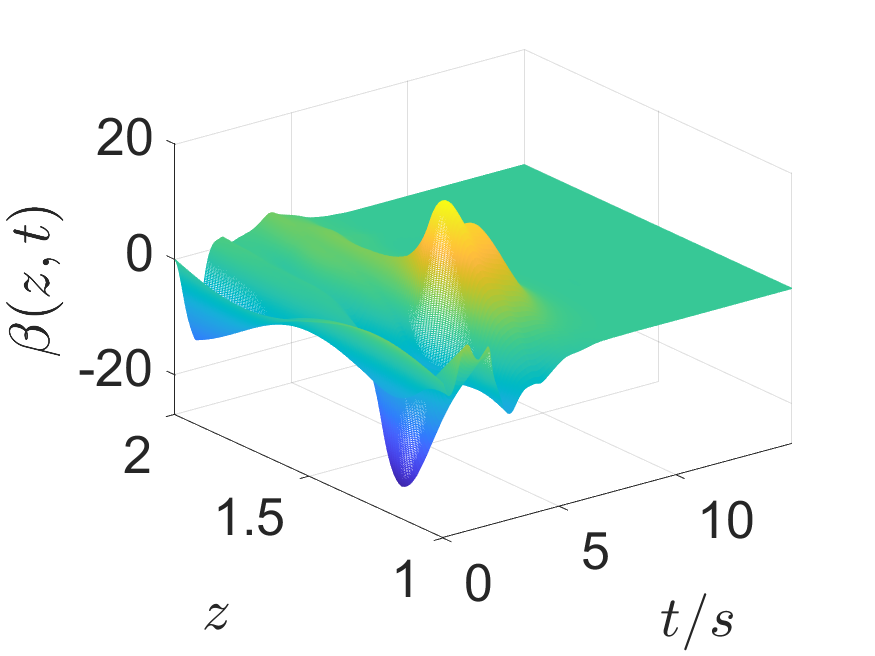}\label{beta evolution}
			}            
			\vspace{-8pt}
            
			\subfigure[]{\includegraphics[width=0.30\linewidth]{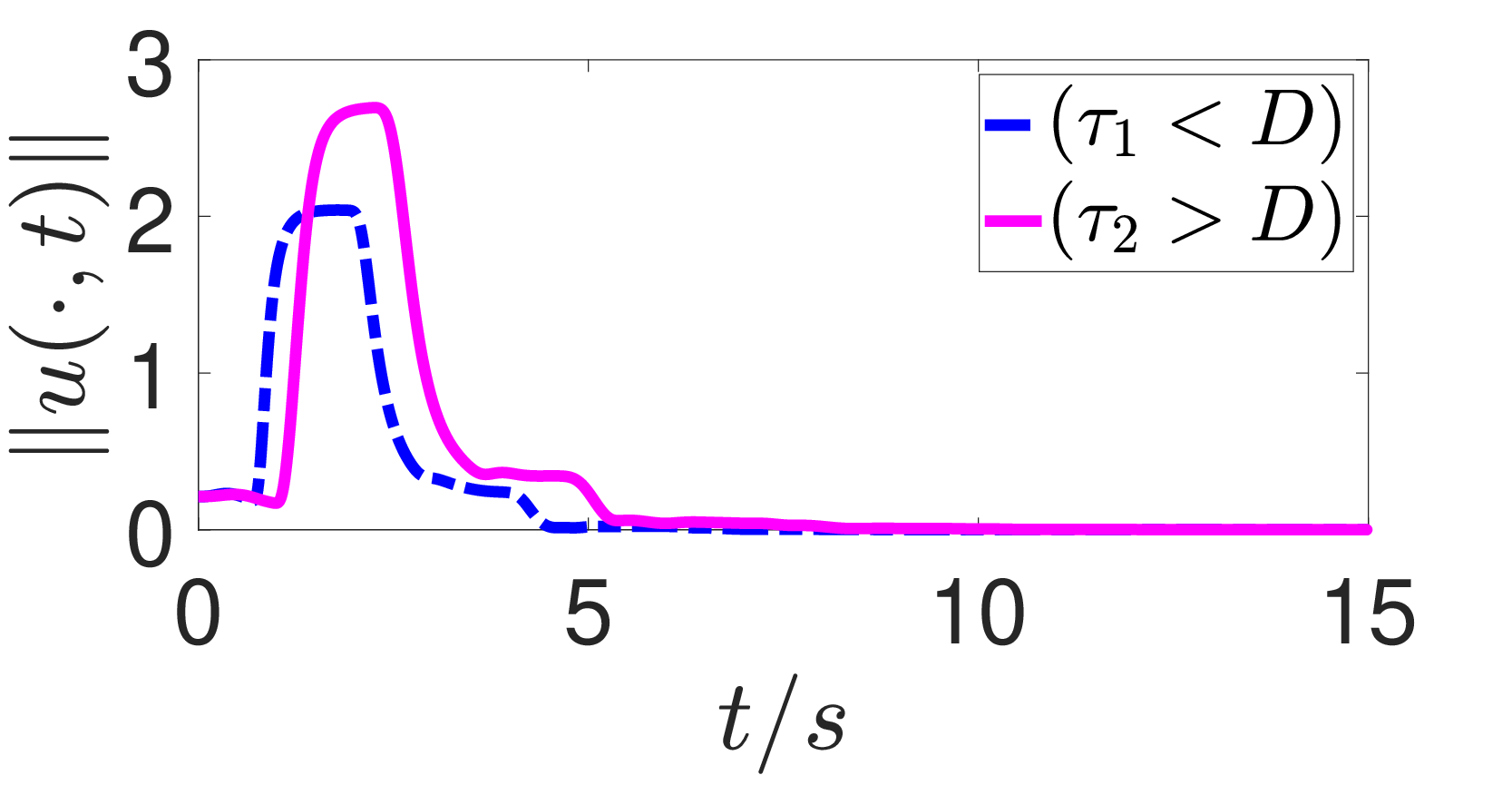}\label{u norm}
			}
			\subfigure[]{\includegraphics[width=0.30\linewidth]{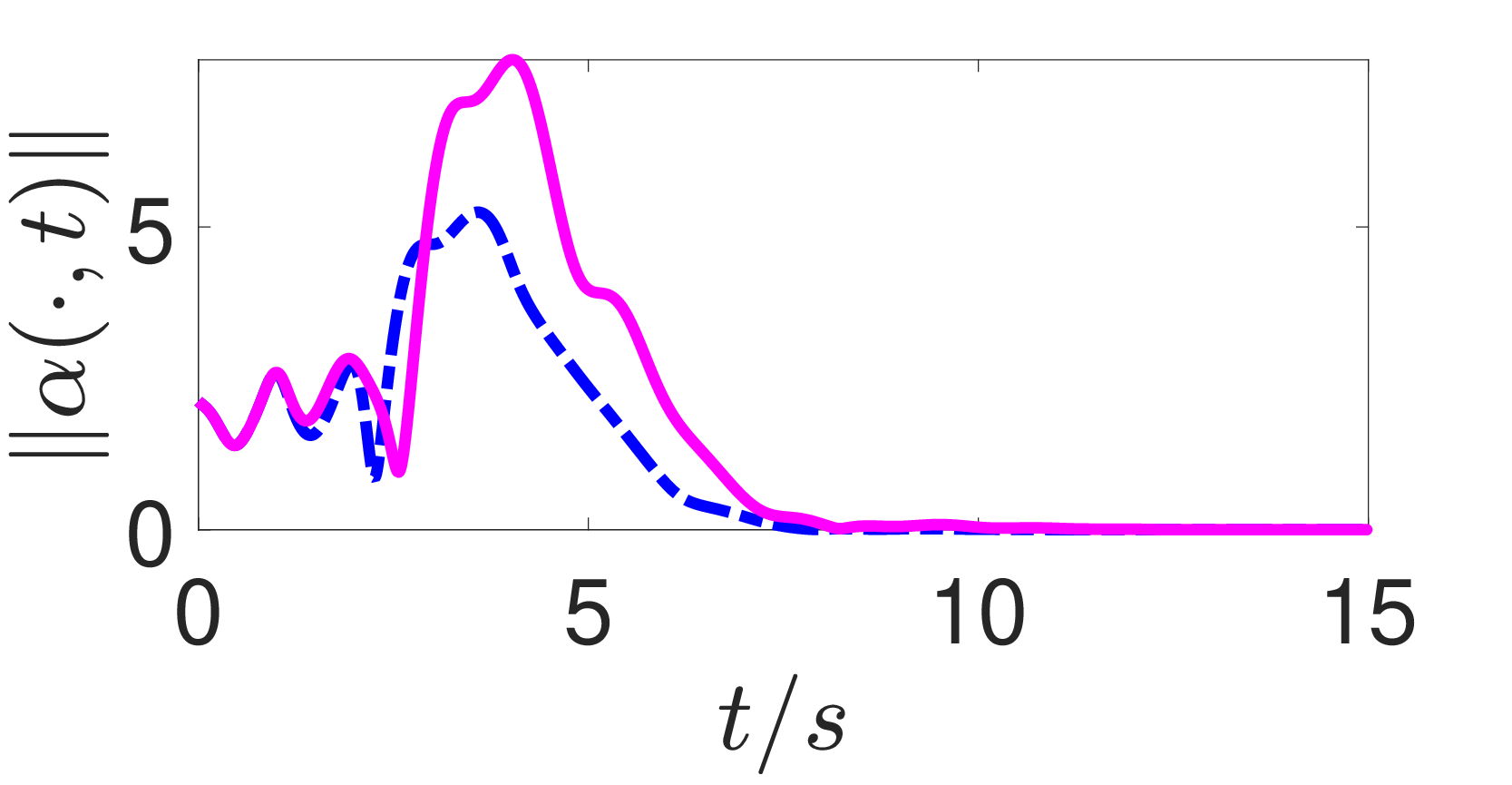}\label{alpha norm}
			}
			\subfigure[]{\includegraphics[width=0.30\linewidth]{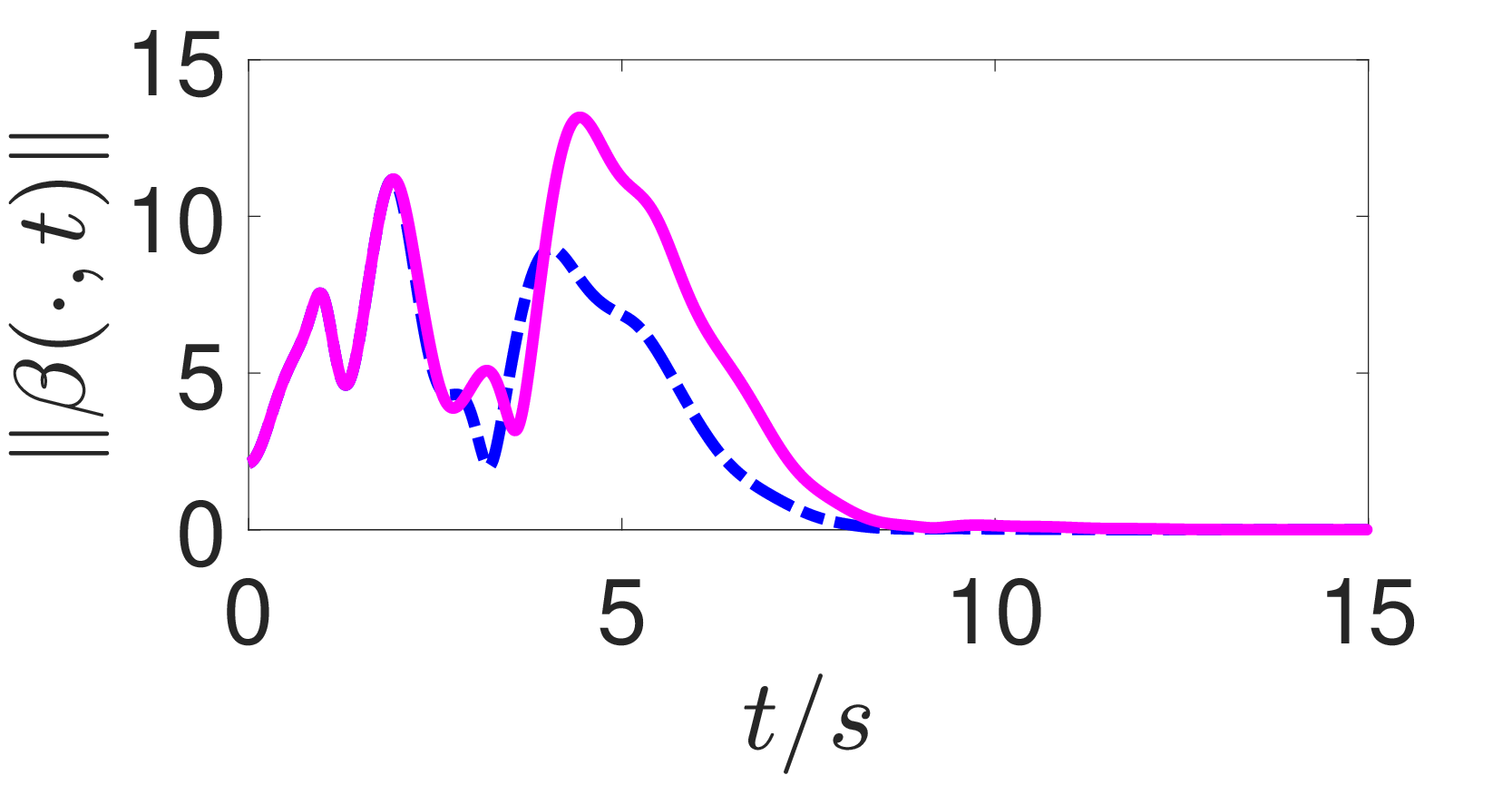}\label{beta norm}
			}            
		\end{center}
            \vspace{-10pt}
            \caption{Closed-loop responses without delay compensation and with delay-compensated. \subref{u evolution uncompensated}-\subref{beta evolution uncompensated} show the dynamics of the states without delay compensation for $\tau=0.8$ and $D=1$. \subref{u evolution}-\subref{beta evolution} show the dynamics of the states under the delay compensation for $\tau=0.8$ and $D=1$. \subref{u norm}-\subref{beta norm} show the $L^2$ norm of the states under the respective delay compensation for two cases: $(\tau_1,D)=(0.8,1)$ and $(\tau_2,D)=(1.2,1)$.
            }
            \label{Comparison of states between delay-compensated and delay-uncompensated}
    \end{figure*}
    
    \begin{figure}[pb]
	\begin{center}
    \subfigure{\includegraphics[width=0.90\linewidth]{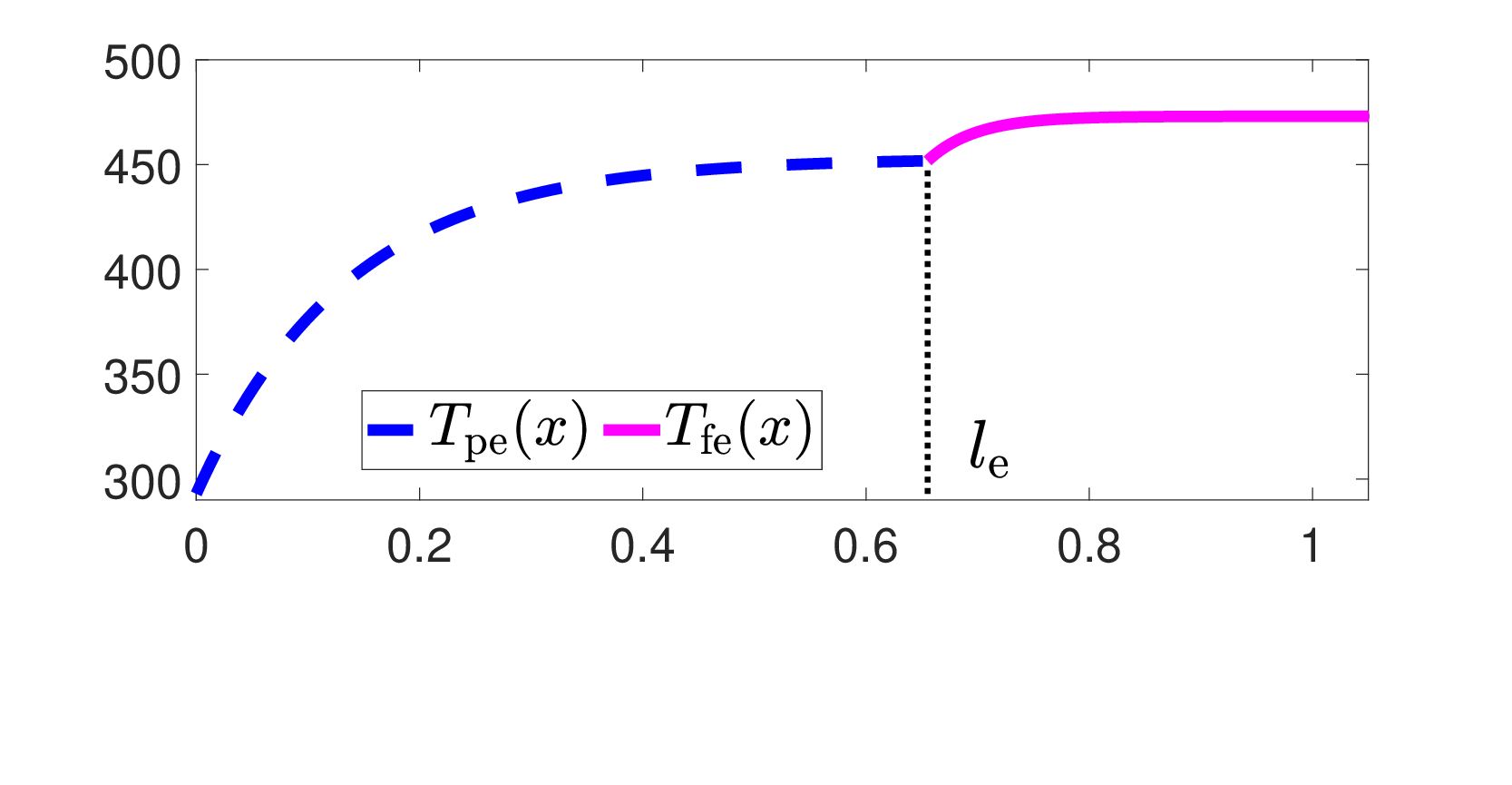}}	
	\end{center}
        \vspace{-30pt}
        \caption{Equilibrium temperature profile.}
        \label{Equilibrium temperature}
    \end{figure}
    For the case without delay compensation, we use the  following nominal controllers:
    \begin{subequations}
		\begin{align}
			U_1(t) =& \int_{0}^{1}\eta (0,y)u(y,t)dy +\varphi (0)^{\mathrm{T}}X(t),\\
			U_2(z,t) =&-\frac{d_1}{\kappa_1}\alpha(1,t)-\frac{\zeta_1}{\kappa_1}\alpha(z,t),   \quad z\in[0,1), \\
			U_2(z,t)
			=&-\frac{d_2}{\kappa_2}\beta(2,t)-\frac{\zeta_2}{\kappa_2}\beta(z,t),   \quad z\in[1,2],
		\end{align}
    \end{subequations}
   where $\varphi$, $\eta$ are solutions of \eqref{varphi, psi} and \eqref{eta, p, q}.

    Although the control laws \eqref{U1} and \eqref{U2} involve kernel functions containing Dirac delta distributions, the corresponding singular terms can be evaluated analytically through integration. For the kernel functions $k$ and $l$, substituting their explicit expressions into the control laws eliminates the delta terms through integration, resulting in explicit control gains. For the kernel function $M$, which is decomposed into a Dirac-delta-related part and a bounded part, the Dirac-delta-related part is treated analytically using the shifting property of the Dirac delta function, while the bounded component is computed by truncating the convergent successive approximation series established in Lemma~\ref{well-posedness M}. In the simulations, the series is truncated after three iterations, since the remaining terms contribute less than $10^{-5}$.

   We compare the system responses under the nominal controllers (without delay compensation) and the delay-compensated controllers, as shown in 
    Figs.~\ref{time evolution of ODE} and~\ref{Comparison of states between delay-compensated and delay-uncompensated}. The simulation results illustrate that the delay-compensated controllers successfully stabilize the system. In contrast, the nominal controllers are unable to achieve stabilization in the presence of input delays.
    \begin{figure*}[tpb]
		\begin{center}
            \subfigure[]{\includegraphics[width=0.30\linewidth]{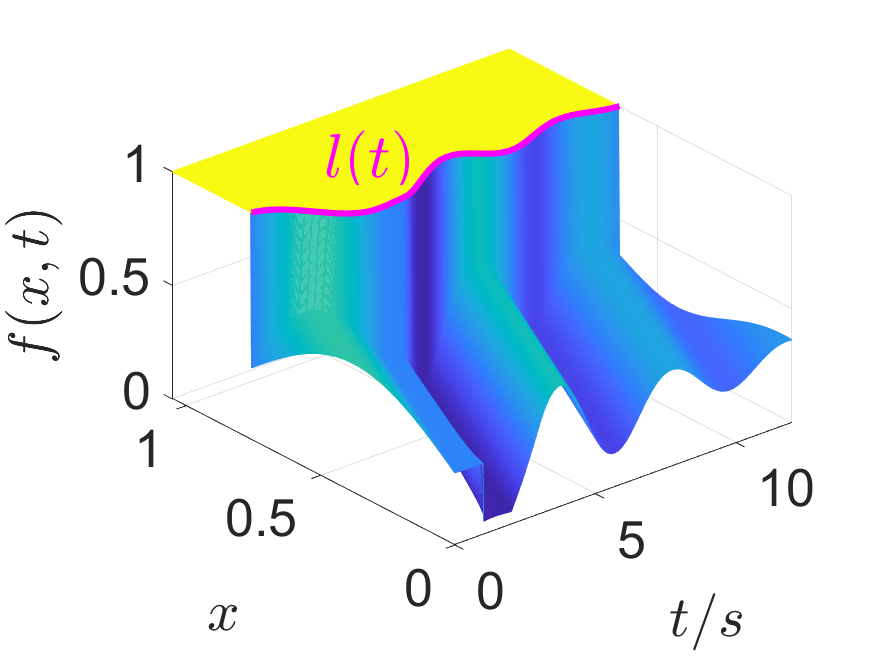}\label{nonlinear filling ratio evolution uncompensated}
			}
			\subfigure[]{\includegraphics[width=0.30\linewidth]{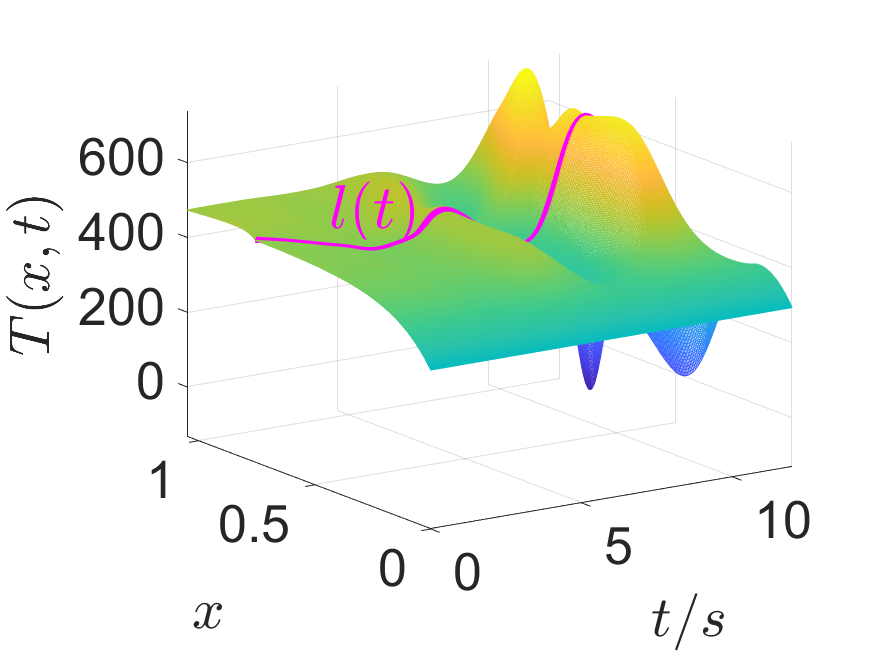}\label{nonlinear temperature evolution uncompensated}
			}
			\subfigure[]{\includegraphics[width=0.30\linewidth]{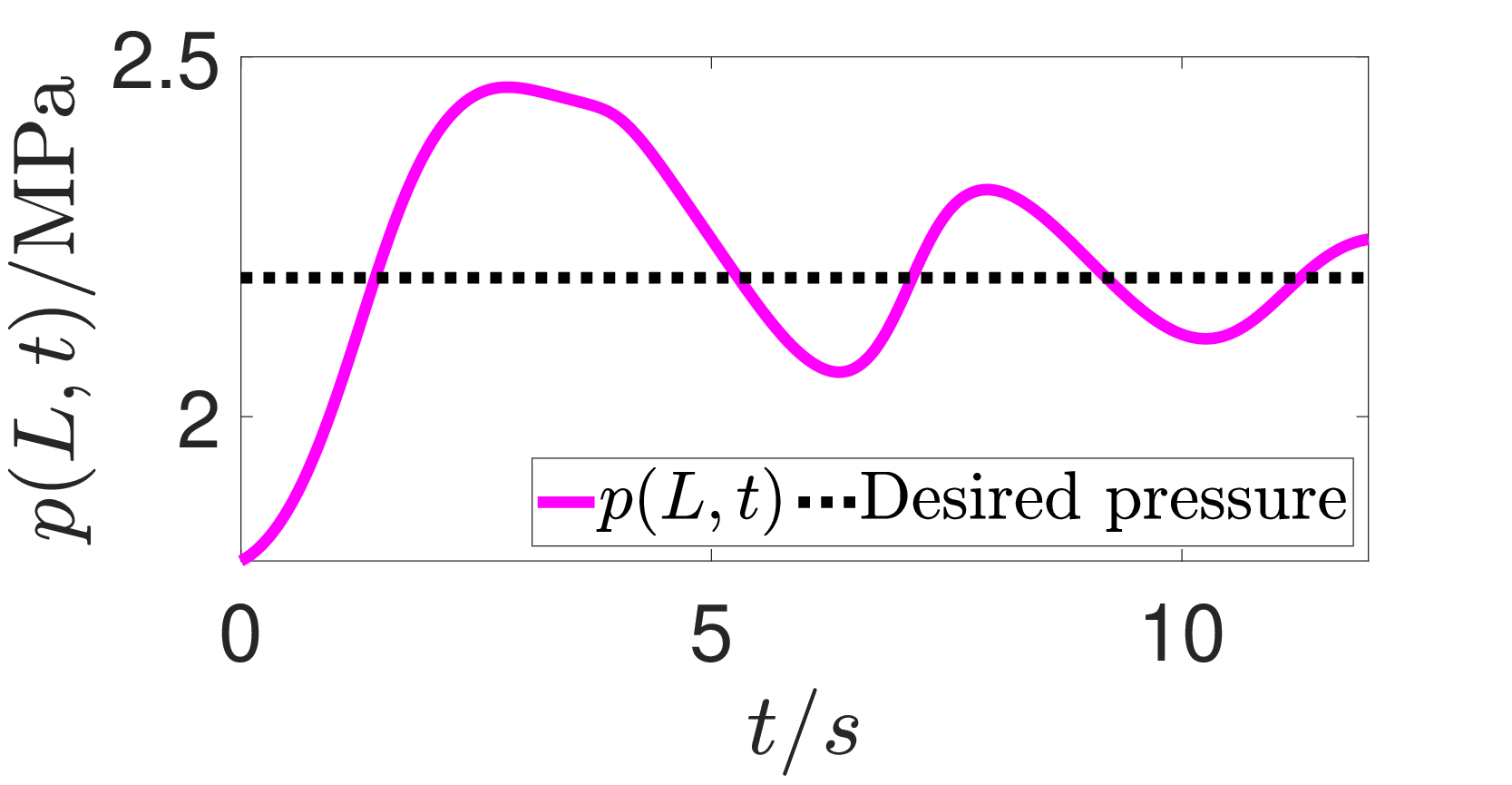}\label{nonlinear p evolution uncompensated}
			}            
            \vspace{-8pt}
            
            \subfigure[]{\includegraphics[width=0.30\linewidth]{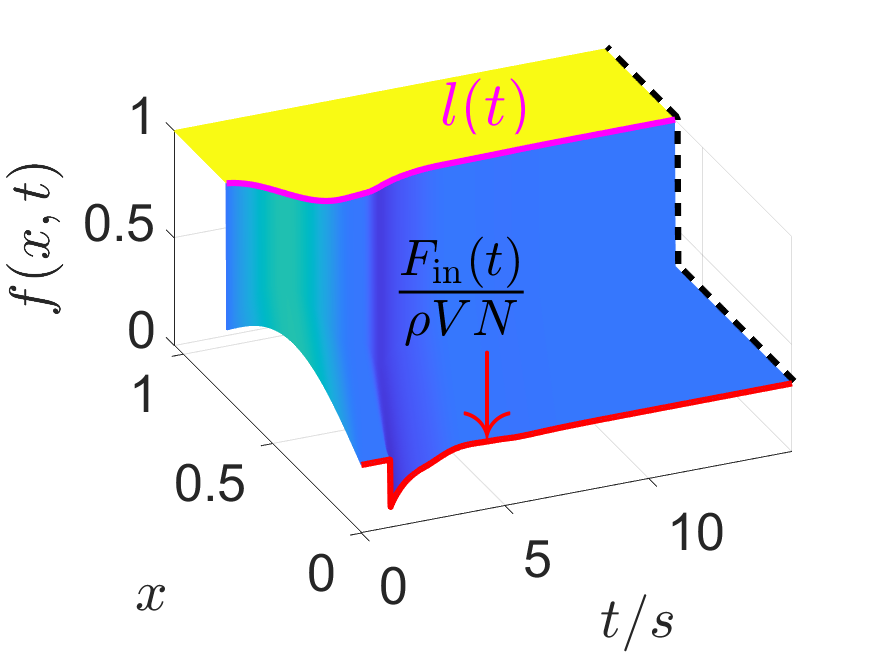}\label{nonlinear filling ratio evolution}
			}
			\subfigure[]{\includegraphics[width=0.30\linewidth]{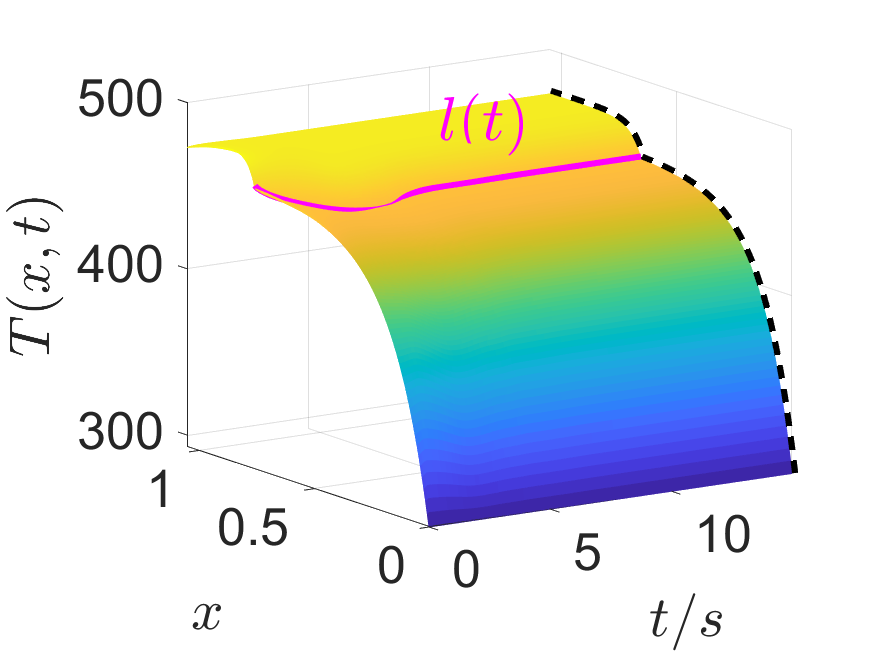}\label{nonlinear temperature evolution}
			}
			\subfigure[]{\includegraphics[width=0.30\linewidth]{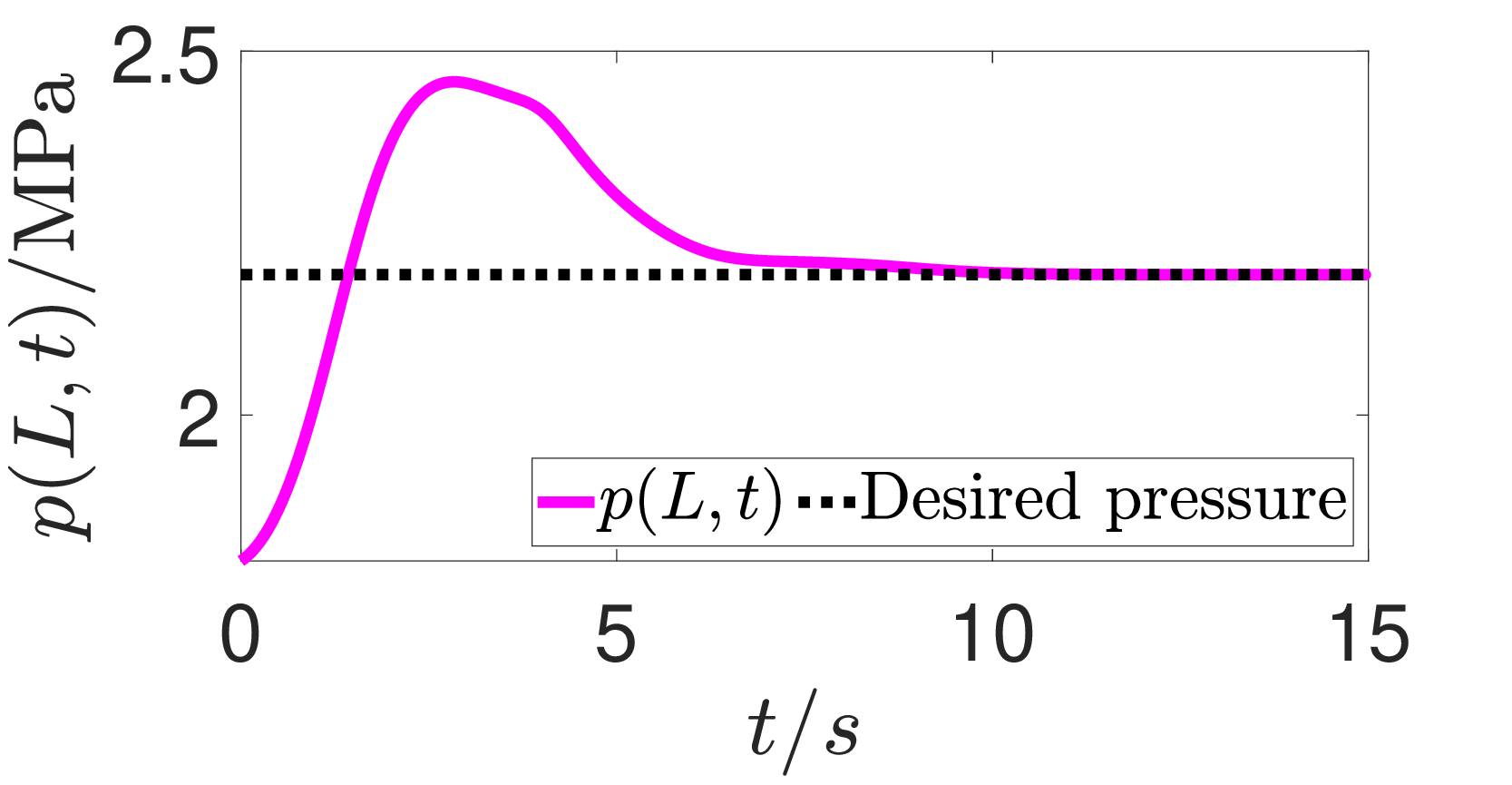}\label{nonlinear p evolution}
			}            
            \vspace{-8pt}
			
			\subfigure[]{\includegraphics[width=0.30\linewidth]{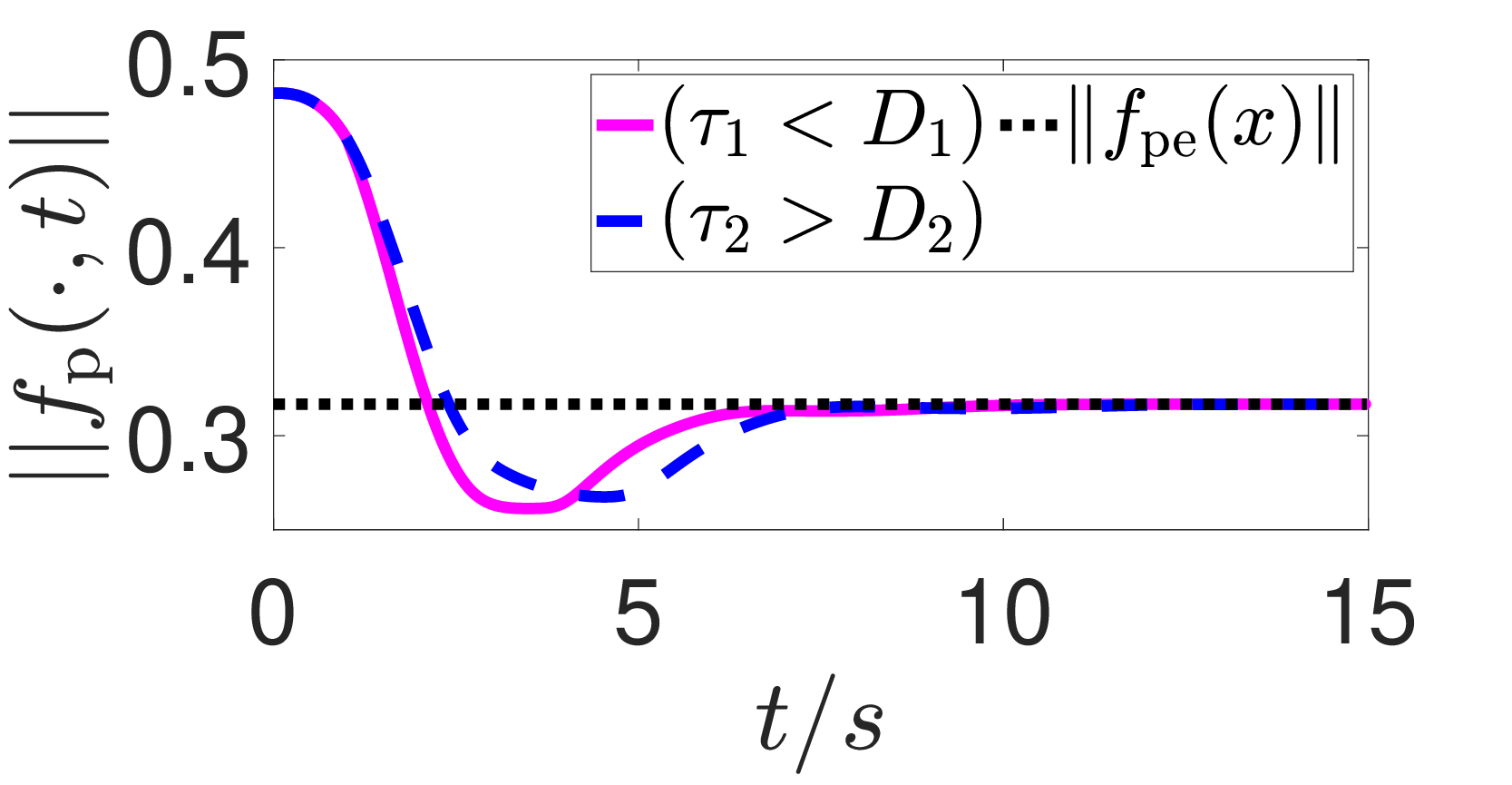}\label{nonlinear fp norm}
			}
			\subfigure[]{\includegraphics[width=0.30\linewidth]{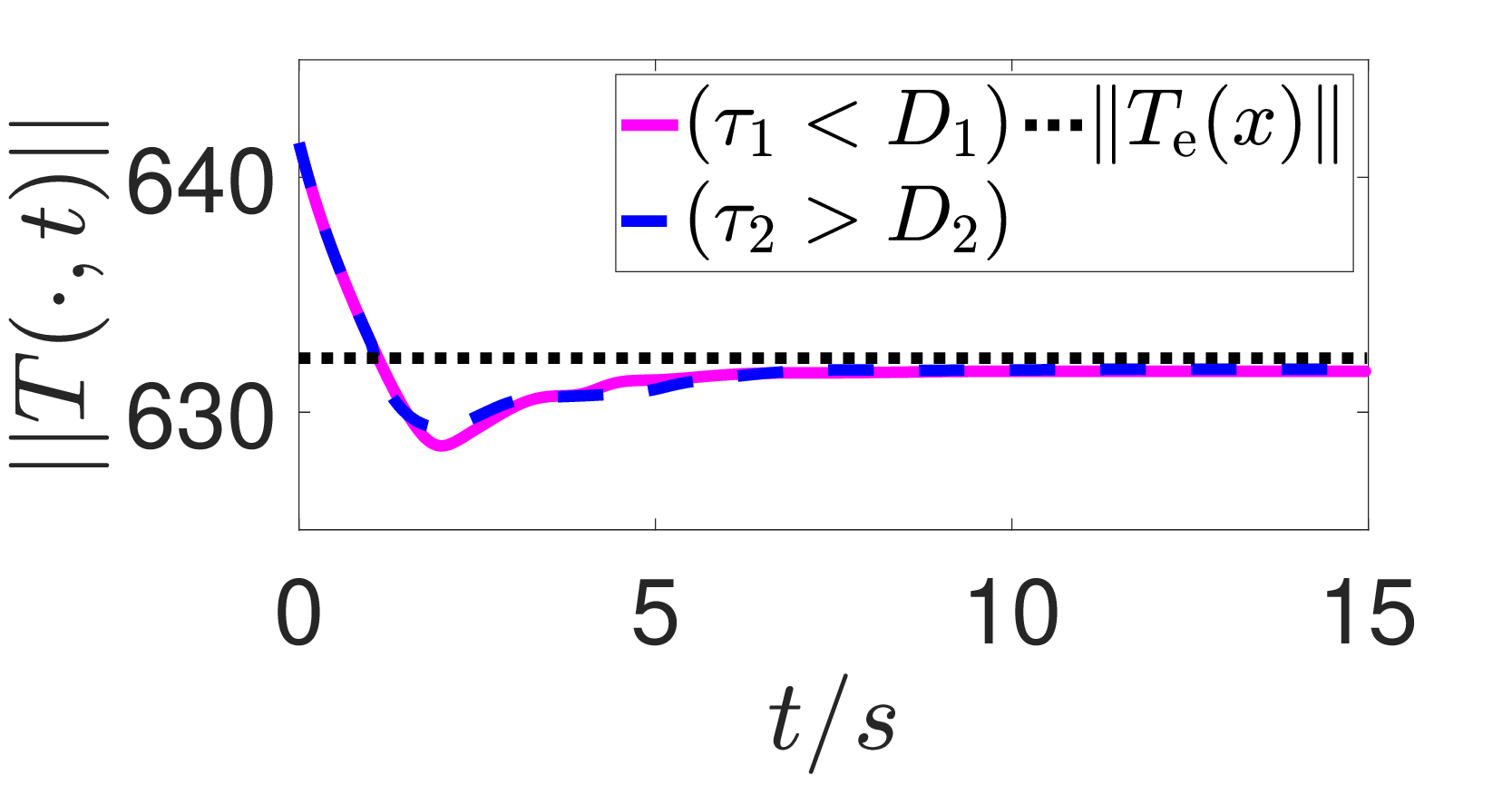}\label{nonlinear T norm}
			}
			\subfigure[]{\includegraphics[width=0.30\linewidth]{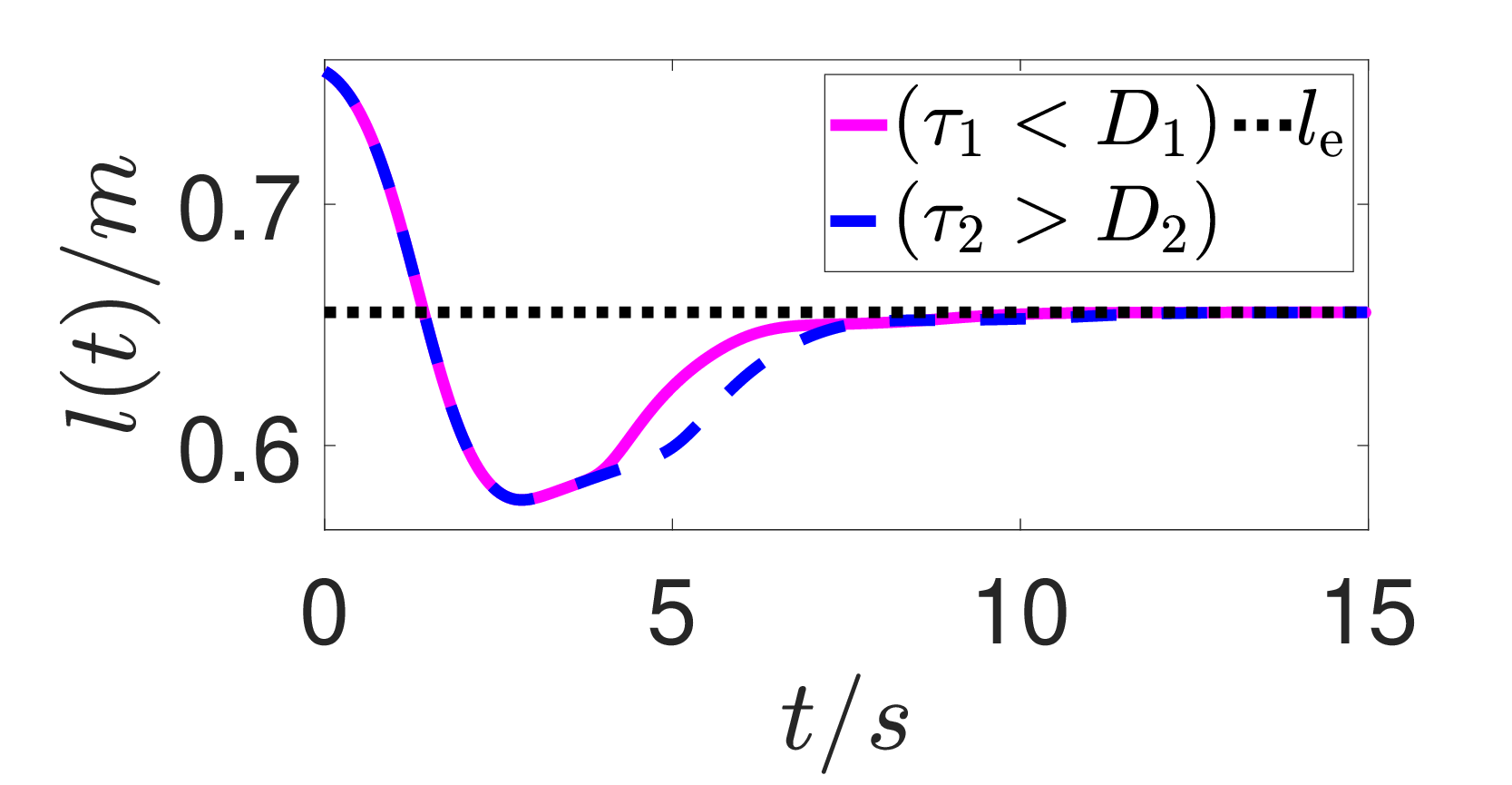}\label{nonlinear l norm}
			}            
		\end{center}
            \vspace{-10pt}
            \caption{Closed-loop response of the screw extrusion system without delay compensation and with delay-compensated. $f:$ filling ratio. $T:$ temperature. \subref{nonlinear filling ratio evolution uncompensated}-\subref{nonlinear p evolution uncompensated} show the dynamics of the states without delay compensation for $\tau=1$ and $D=2$. \subref{nonlinear filling ratio evolution}-\subref{nonlinear p evolution} show the dynamics of the states under the delay compensation for $\tau=1$ and $D=2$. The black dashed lines show the equilibrium profile. \subref{nonlinear fp norm}-\subref{nonlinear l norm} show the $L^2$ norm of the states under the respective delay compensation for two cases: $(\tau_1,D_1)=(1,2)$ and $(\tau_2,D_2)=(2,1)$.}
            \label{fig:closed-loop dynamics of extrusion}
    \end{figure*}
    \vspace{-3pt}
    
    \subsection{Screw Extrusion Process}
    We consider a  melt-spinning screw extrusion model, where the barrel temperature and the feed rate serve as manipulated inputs, with the objective of achieving the desired temperature and pressure at the die exit of the extrusion chamber. The state variables include the incompletely melting polymer temperature $T_\mathrm{p}(x,t)$, the molten polymer temperature $T_\mathrm{f}(x,t)$, the filling ratios $f_\mathrm{p}(x,t)$  and  $f_\mathrm{f}(x,t)$ in the partially filled zone (PFZ) and  in the fully filled zone (FFZ), respectively, the interface position $l(t)$ between the PFZ and the FFZ, and the chamber pressure $p(x,t)$ in the FFZ. Note that $f_\mathrm{f} \equiv 1$ in the FFZ, as the polymer is fully melted.
    
    The mathematical model is presented as follows \cite{diagne2011well, dos2013introduction}:
    \begin{subequations}\label{screw extrusion model}
    \begin{align}
		\frac{\partial f_\mathrm{p}}{\partial t}(x,t) = &-\xi N\frac{\partial f_\mathrm{p}}{\partial x}(x,t), \qquad \quad  \quad x\in (0,l(t)),
		\label{eq:mass balance of PFZ}\\
		\frac{\partial T_\mathrm{p}}{\partial t}(x,t) = 
		&-\xi N\frac{\partial T_\mathrm{p}}{\partial x}(x,t)+\frac{\mu_\mathrm{p}\eta N^2}{f_\mathrm{p}(x,t)\rho Vc_\mathrm{p}}  \notag \\
		+\frac{S_\mathrm{ech}\hbar}{\rho Vc_\mathrm{p}}\Bigl(&T_\mathrm{b}(x, t-D)-T_\mathrm{p}(x,t)\Bigr), ~x\in (0,l(t)),
		\\
		\frac{\partial T_\mathrm{f}}{\partial t}(x,t) =&-\frac{F_\mathrm{d}(t)\xi}{\rho V}\frac{\partial T_\mathrm{f}}{\partial x}(x,t)+\frac{\mu_\mathrm{f} \eta N^2}{\rho Vc_\mathrm{f}}
		\notag \\
		+\frac{S_\mathrm{ech}\hbar}{\rho Vc_\mathrm{f}}\Bigl(&T_\mathrm{b}(x, t-D)-T_\mathrm{f}(x,t)\Bigr), ~x\in (l(t), L),
		 \\
		\frac{\partial p}{\partial x}(x,t) = &\eta \frac{\rho VN-F_\mathrm{d}(t)}{B\rho}, \quad \qquad \quad x\in (l(t), L),\label{eq:mass balance of FFZ}\\
		\frac{dl(t)}{dt} =& \frac{F_\mathrm{d}(t)-\rho VNf_\mathrm{p}(l(t),t)}
		{\rho S_\mathrm{eff}(1-f_\mathrm{p}(l(t),t))} ,  \label{eq:coupling interface}
	\end{align}
    \end{subequations}
	with boundary conditions:
	\begin{align*}
		f_\mathrm{p}(0,t) \!=\!\frac{F_{\rm{in}}(t\!-\!\tau)}{\rho VN}, 
		T_\mathrm{p}(0,t) \!=\! T_\mathrm{in}, 
		T_\mathrm{f}(l(t),t) \!=\! T_\mathrm{p}(l(t),t),
	\end{align*}
    where the feed rate $F_{\rm{in}}(t)$ and the barrel temperature $T_\mathrm{b}(z,t)$ are control inputs; $\tau$ and $D$ are time delays in the feed rate and barrel temperature, respectively, and $F_\mathrm{d}(t)= \frac{k_\mathrm{d}}{\eta}(p(L,t)-P_0)$ is the net flow rate at the die. 
    We use the same parameters as in \cite{lotero2017distributed, zhong2022adaptive}, and let $\alpha=400$, $\mu_\mathrm{p}, \mu_\mathrm{f}=10^{-4}$, and $\eta=0.2$. $L=1.05\,\mathrm{m}$ represents the extruder length.

    The desired  temperature and  pressure at the die exit are set to  $T_\mathrm{f}(L,t)=473\,\mathrm{K}$, $p(L,t)=2.2\,\mathrm{MPa}$, thus the steady-state operating point of  \eqref{screw extrusion model} is given by  $F_\mathrm{in}^{\star}=1.5\,\mathrm{kg/h}$; $T_\mathrm{b}^{\star}(x)=453\,\mathrm{K}$ for PFZ, $T_\mathrm{b}^{\star}(x)=473\,\mathrm{K}$ for FFZ, and the equilibrium: $l_\mathrm{e}=0.6552\,\mathrm{m}$, $f_\mathrm{pe}(x)=0.3168~(x\in[0,l_\mathrm{e}])$, $T_\mathrm{pe}(x)=453.01-160.01\mathrm{e}^{-7.41x}/\mathrm{K}~(x\in[0,l_\mathrm{e}])$, and $T_\mathrm{fe}(x)=473-9.59\times10^{7}\mathrm{e}^{-23.39x}/\mathrm{K}~(x\in[l_\mathrm{e}, L])$. The equilibrium temperature profile is shown in Fig.~\ref{Equilibrium temperature}. 

     The model \eqref{screw extrusion model} is related to the system \eqref{main model} by first mapping the time-varying spatial domain to a fixed reference domain and then linearizing the resulting dynamics about the equilibrium point. Standard procedures for the fixed-domain transformation and linearization can be found in \cite{dos2013introduction}. More details are provided in Appendix~\ref{appendix:transformation of screw extrusion system}.

    This yields the linear system \eqref{main model} with coefficients $a=0.2925$, $b=0.1538$, $c_1=c_2=-1.4202$, $d_1=d_2=0$, $f=-0.0471$, $\theta=0$, $A=-0.1538$, $B=-0.2805$, and $\kappa_1=\kappa_2=1.4202$. The control parameters are chosen as $\zeta_1=5$, $\zeta_2=5$, and $K=8A/B$. We consider two heterogeneous-delay cases: $(\tau_1, D_1)=(1, 2)$ and $(\tau_2, D_2)=(2, 1)$.

    Nevertheless, we design the controllers \eqref{U1} and \eqref{U2} for the linear system \eqref{main model}, and apply the controllers directly to the nonlinear model \eqref{screw extrusion model}, with results shown in Fig.~\ref{fig:closed-loop dynamics of extrusion}.    
    The simulation is performed using the finite difference method after converting the time-varying domain PDEs to fixed-domain PDEs via introducing coordinate transformations as illustrated in \cite{dos2013introduction,zhong2022adaptive}. The step sizes are set to $\Delta t=0.02$, $\Delta z=0.01$, and $\Delta s=0.02$. The initial conditions in the fixed domains are: $\bar{f}_\mathrm{p}(z,0)=0.25\sin(\pi z)+\bar{f}_\mathrm{pe}(z)$, $\bar{T}_\mathrm{p}(z, 0)=10\sin(\frac{\pi}{2} z)+\bar{T}_\mathrm{pe}(z)$, $\bar{T}_\mathrm{f}(z, 0)=10\sin(\frac{\pi}{2} z)+\bar{T}_\mathrm{fe}(z)$, and $l(0)=0.1+l_\mathrm{e}$. Additionally, let $F_\mathrm{in}(t)=F_\mathrm{in}^{\star}$ for $-\tau \le t<0$, and $T_\mathrm{b}(z, t)=T_\mathrm{b}^{\star}$ for $-D \le t<0$.

    Figure~\ref{fig:closed-loop dynamics of extrusion} demonstrates that the delay-compensated controllers effectively stabilize the nonlinear extrusion system. Conversely, the uncompensated nominal controllers fail to reach the desired equilibrium, exhibiting oscillations and divergence as evidenced in Figs.~\ref{nonlinear filling ratio evolution uncompensated}--\ref{nonlinear p evolution uncompensated}. Under the proposed controller, all states converge to their desired values (black dashed lines) within 10 seconds, as shown in Figs.~\ref{nonlinear filling ratio evolution}--\ref{nonlinear p evolution}. This convergence is further confirmed by the $L^2$ norms in Figs.~\ref{nonlinear fp norm}--\ref{nonlinear T norm}, which vanish within 10 seconds. Finally, Fig.~\ref{control effort in extrusion} presents the control inputs.

    The controller used in the simulation is designed based on the linearized model and applied to the original nonlinear screw extrusion system. The effectiveness of this approach is illustrated through the considered numerical example. A rigorous local stability analysis is not included in this work. Following the approach in \cite{coron2013local,hu2019boundary}, such an analysis can be developed by applying the linear backstepping transformations to the original nonlinear system. This leads to a target system consisting of an exponentially stable linear subsystem and higher-order nonlinear perturbation terms. With suitable regularity assumptions and Lyapunov estimates, local exponential stability might be established for sufficiently small initial perturbations.

    \begin{figure}[htpb]
	\begin{center}
        \subfigure{\includegraphics[width=0.84\linewidth]{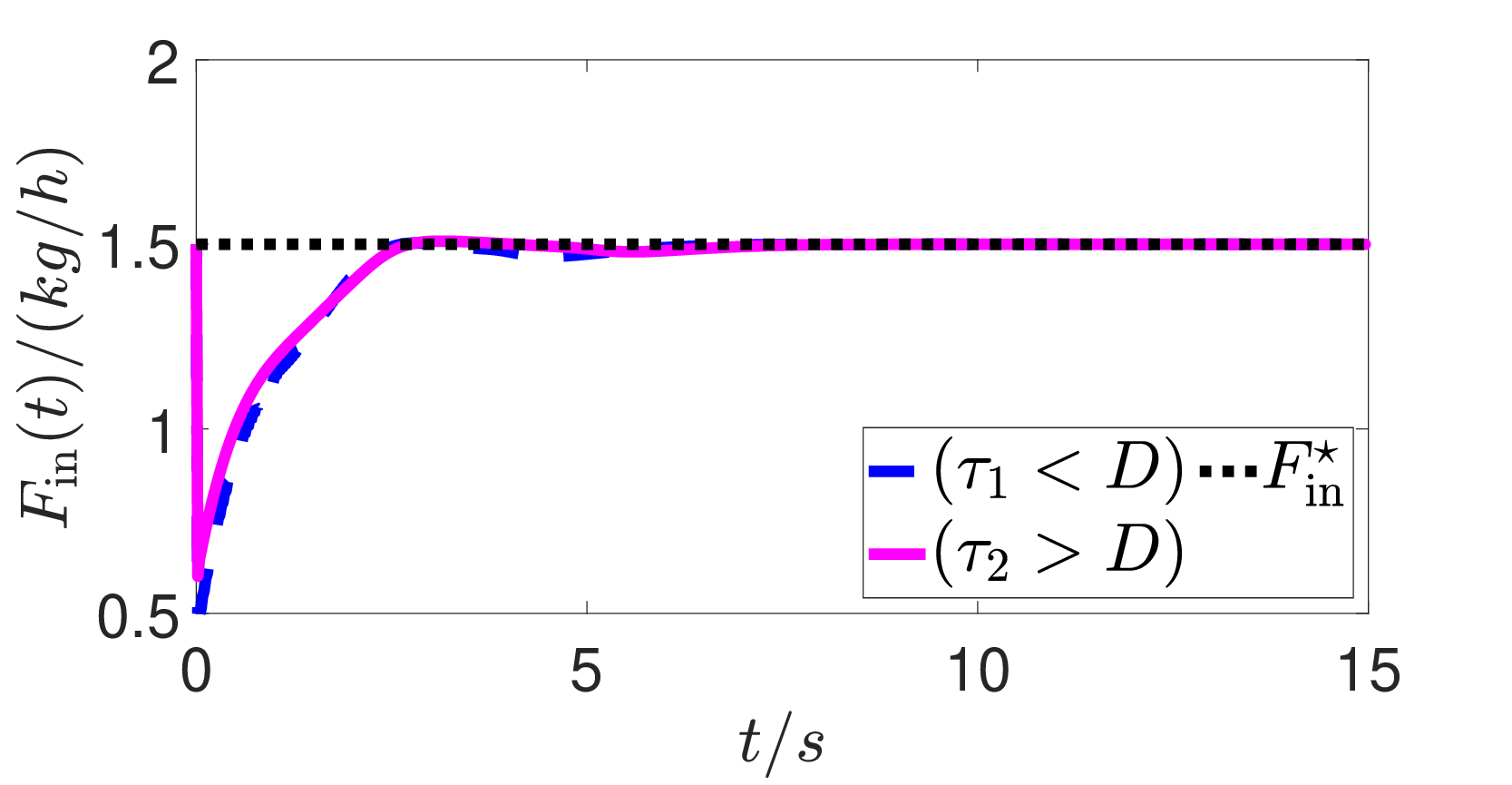}
		}

        \vspace{-12pt}
        \subfigure{\includegraphics[width=0.84\linewidth]{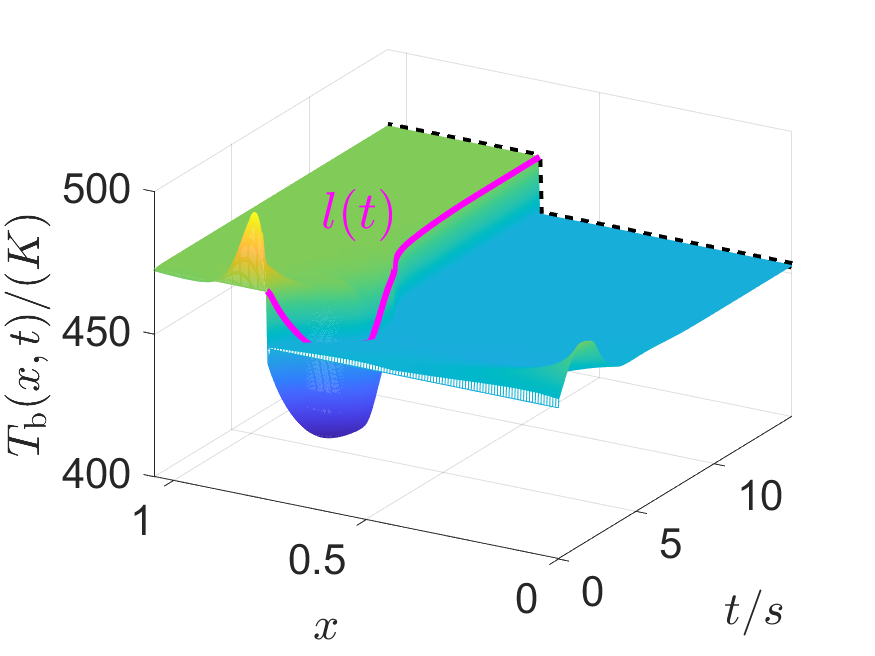}
		}		
	\end{center}
        \vspace{-15pt}
        \caption{Control inputs. Top: feed rate $F_\mathrm{in}(t)$ under the respective delay compensation for two cases: $(\tau_1,D_1)=(1,2)$ and $(\tau_2,D_2)=(2,1)$. Bottom: barrel temperature $T_\mathrm{b}(x,t)$ under the delay compensation for $\tau=1$, $D=2$.}
        \label{control effort in extrusion}
    \end{figure}
    \vspace{-10pt}
	
\section{Conclusion}\label{conclusion}
    This paper presents a control design methodology to compensate for both boundary and distributed input delays occurring in a two-phase advective unidirectional first-order PDE-ODE system. The proposed method can be applied to distinct delays without any ordering constraints. We design three invertible backstepping transformations, which map the delayed unstable system to a stable target system, resulting in three groups of coupled kernel equations. We establish the well-posedness of these kernel equations--which contain delta functions at the boundary--and prove that the direct and inverse transformations are bounded. Finally, the exponential stability of the closed-loop system in the $L^\infty$ norm is established by analyzing the explicit solution of the target system. We apply the controllers designed for the linear system to a nonlinear screw extruder model in simulation. The results show that the proposed controllers can stabilize the nonlinear system with a fast convergence rate. 
    
    Future research will focus on more practical segmented barrel temperature control strategies, in which the entire temperature control region is divided into a finite number of segments with piecewise constant inputs, replacing the distributed control at infinitely many points.

\appendices 
\section{Proof of Proposition \ref{invertibility of v trans}} \label{proof of invertibility}
    Since the transformations \eqref{u trans}--\eqref{v trans} possess a lower-triangular structure, and the transformations \eqref{u trans} and \eqref{r trans} are of affine Volterra type, whose invertibility is straightforward \cite[Theorem 1.4]{anfinsen2019adaptive}, it suffices to prove the invertibility of the transformation \eqref{v trans}, with $\gamma$, $\alpha$, $\beta$, $u$, and $X$ regarded as given functions.

To this end, we interpret the transformation  \eqref{v trans}  as an integral equation with respect to the unknown function $v$, while treating the remaining states $\gamma$, $\alpha$, $\beta$, $u$, and $X$ as known inputs. The invertibility of the transformation then follows from the existence and uniqueness of the solution to this integral equation. We therefore begin by analyzing the corresponding homogeneous integral equation as follows:
\begin{align}
    \gamma(z,s,t) =&v(z,s,t)-\mathcal{T}^v[v](z,s). \label{homogeneous integral transformation}
\end{align}
Let 
\begin{align}
    \rho(z,s,t) = \mathcal{T}^v[v](z,s), \label{rho definition}
\end{align}
then 
\begin{align}
     v(z,s,t)=&\gamma(z,s,t)+\rho(z,s,t). \label{v in terms of gamma and rho}
\end{align}
Substituting  \eqref{v in terms of gamma and rho} into \eqref{rho definition} yields 
\begin{align}
    \rho(z,s,t) = \mathcal{T}^v[\gamma](z,s)+\mathcal{T}^v[\rho](z,s). \label{rho integral equation}
\end{align}
Next, we will prove the existence of a unique solution $\rho$ to the integral equation \eqref{rho integral equation}. Let 
\begin{align}
    \rho^0(z,s,t) =& \mathcal{T}^v[\gamma](z,s), \label{rho 0} \\
    \rho^n(z,s,t) =& \mathcal{T}^v[\rho^{n-1}](z,s), \quad n=1,2,\cdots \label{rhon iteration}
\end{align}

Substituting \eqref{l solution} and \eqref{k solution} into $\mathcal{T}^v$ defined in \eqref{definition of components in trans}, and after a lengthy computation, we obtain the explicit expression for 
 $z<1$ as follows:
\begin{align}
    \mathcal{T}^v&[\rho](z,s) = \!\!\int_{0}^s\int_{\lambda_1(s)+Da\epsilon}^{1}\!\!Dk_3(z,s-\epsilon,y)\rho(y,\epsilon)dyd\epsilon \notag \\ 
        &-\!\int_{0}^{s}Dd_1\varrho_1(s-\epsilon)\rho(\lambda_1(s)+Da\epsilon,\epsilon)d\epsilon \notag \\
        &-\! \int_{\max\{0,s-\frac{z}{Da}\}}^{s}\!\!D\zeta_1\varrho_1(s\!-\!\epsilon)\rho(Da(\epsilon \!-\! s) \!+\! z,\epsilon)d\epsilon. \label{T v rho z<1}
\end{align}
It then follows from \eqref{rho 0} and \eqref{T v rho z<1} that
\begin{align}
    |\rho^0(z,s,t)| \le & C_0\|\gamma\|_{\infty}, 
\end{align}
where $C_0=D\bar k_3+D(d_1+\zeta_1)\bar \varrho>0$, with $\bar k_3=\sup|k_3(z,s,y)|$, $\bar\varrho=\max_{s\in[0,1],i=1,2}\{\mathrm{e}^{Dc_is}\}$. Assume 
\begin{align}
    |\rho^n(z,s,t)|\le \frac{C_0^{n+1}s^n}{n!}\|\gamma\|_{\infty}. \label{rhon estimate}
\end{align}
Inserting \eqref{rhon estimate} into \eqref{rhon iteration} yields
\begin{align*}
    |\rho^{n+1}(z,s,t)|\le &\int_{0}^s\int_{\lambda_1(s)+Da\epsilon}^{1}\!\!D\bar k_3\frac{C_0^{n+1}\epsilon^n}{n!}\|\gamma\|_{\infty} dyd\epsilon \notag \\ 
    &+\int_{0}^{s}Dd_1\bar\varrho \frac{C_0^{n+1}\epsilon^n}{n!}\|\gamma\|_{\infty} d\epsilon \notag \\
    &+\int_{\max\{0,s-\frac{z}{Da}\}}^{s}D\zeta_1\bar\varrho \frac{C_0^{n+1}\epsilon^n}{n!}\|\gamma\|_{\infty} d\epsilon \notag \\
    \le &\frac{C_0^{n+2}\|\gamma\|_{\infty}}{n!}\int_0^s\epsilon^nd\epsilon \le \frac{C_0^{n+2}s^{n+1}}{(n+1)!}\|\gamma\|_{\infty}.
\end{align*}
Hence, \eqref{rhon estimate} is proved. Similarly, when $z\ge 1$, 
we can prove 
\begin{align}
    |\rho^n(z,s,t)|\le \frac{C_1^{n+1}s^n}{n!}\|\gamma\|_{\infty},
\end{align}
where $C_1=D(\bar k_3+\bar l_2+(d_1+d_2+\zeta_2+\frac{a}{b}\mathrm{e}^{|c_2-c_1|/b}\bar\zeta_2)\bar\varrho)$, and $\bar l_2=\sup|l_2(z,s,y)|$. Therefore, the series
\begin{align}
    \rho(z,s,t)=\sum\limits_{n=0}^{\infty}\rho^n(z,s,t)
\end{align}
converges uniformly with a bound $|\rho|\le C_2\|\gamma\|_{\infty}$, where $C_2=\max\{C_0\mathrm{e}^{C_0},C_1\mathrm{e}^{C_1}\}$, which results in $\|\rho\|_{\infty}\le C_2 \|\gamma\|_{\infty}$. Hence, \eqref{rho integral equation} admits a solution. To prove uniqueness, suppose that \eqref{rho integral equation} has two solutions $\rho_1$ and $\rho_2$, and define their difference as $\tilde{\rho}:=\rho_1-\rho_2$. Then  $\tilde\rho$ satisfies the homogeneous equation $\tilde\rho(z,s,t) = \mathcal{T}^v[\tilde\rho](z,s)$. Applying the same contraction argument as above, we obtain that $\tilde{\rho}\equiv 0$, which implies $\rho_1=\rho_2$. Therefore, the solution to \eqref{rho integral equation} is unique.

Therefore, there exist a bounded linear operator $\Lambda$: $L^\infty([0,2]\times[0,1])\to L^\infty([0,2]\times[0,1])$ such that
\begin{align}
    v(z,s,t)=(I+\Lambda)[\gamma](z,s), \label{v operator representation}
\end{align}
where $I$ is the identity operator. The invertibility of \eqref{homogeneous integral transformation} is thereby established.

We now consider the transformation \eqref{v trans} with the affine terms $\mathcal{T}^\alpha$, $\mathcal{T}^\beta$, $\mathcal{T}^u$, and $\phi(z,s)^{\mathrm{T}}X(t)$.  Let  $\rho$ be defined as in \eqref{rho definition}, which yields
\begin{equation}
    \rho=\mathcal{T}^{v}[\rho]+\mathcal{T}^{v}[\gamma+\mathcal{T}^{\alpha}[\alpha]+\mathcal{T}^{\beta}[\beta]
    +\mathcal{T}^{u}[u]+\phi^{\mathrm{T}}X(t)].
\end{equation}
Recalling Lemmas~\ref{direct norm} and~\ref{direct norm u}, it follows that
\begin{align}
    \left|\mathcal{T}^{v}[\gamma+\mathcal{T}^{\alpha}[\alpha]+\mathcal{T}^{\beta}[\beta]
    +\mathcal{T}^{u}[u]+\phi^{\mathrm{T}}X(t)]\right|\le \bar G \varXi,
\end{align}
where $\varXi: =\|\gamma\|_{\infty}+\|\alpha\|_{\infty}+\|\beta\|_{\infty}+\|u\|_{\infty}+|X(t)|$, and $\bar G=C_v\max\{1,C_\alpha,C_\beta,C_u,\bar\phi\}>0$, with $\bar\phi=\sup|\phi(z,s)|$. $C_v=C_0$ when $z<1$, and $C_v=C_1$ when $z\ge1$. Proceeding in a way analogous to the homogeneous case, we obtain
\begin{align}
    |\rho^n(z,s,t)|\le& \frac{C_0^{n}s^n}{n!}\bar G\varXi(t), \quad z<1, \\
    |\rho^n(z,s,t)|\le &\frac{C_1^{n}s^n}{n!}\bar G\varXi(t), \quad z\ge 1.
\end{align}
Consequently, $\|\rho\|_\infty\le C_3  \varXi$ with $C_3=\max\{\bar G\mathrm{e}^{C_0},\bar G\mathrm{e}^{C_1}\}$. Therefore, the Fredholm transformation \eqref{v trans} is invertible. This completes the proof of Proposition~\ref{invertibility of v trans}.

\section*{References}
\bibliographystyle{IEEEtranS}
\bibliography{IEEEabrv, refs.bib}

\clearpage
\appendices
\renewcommand{\Roman}[1]{\Alph{#1}}
\setcounter{lemma}{0}
\renewcommand{\thelemma}{\thesection.\arabic{lemma}}
\makeatletter
\@addtoreset{lemma}{section} 
\makeatother

\setcounter{equation}{0}
\renewcommand{\theequation}{\Alph{section}.\arabic{equation}} 
\makeatletter
\@addtoreset{equation}{section} 
\makeatother

\section{\texorpdfstring{Well-posedness of the Equations for $\eta$ and $q$}{Well-posedness of eta and q}}\label{appendix:well-posedness of eta and q equation}
The well-posedness result is summarized in the following lemma.
   \begin{lemma}
    The PDE system~\eqref{eta, p, q} admits a unique solution $(\eta,q) \in C(\mathcal{D}_1)\times C([0,1];L^2(0,1))$.
    \end{lemma}
    \begin{proof}
First, we consider the equation for $\eta$, which is independent of $q$. Introduce the following change of variables: 
\begin{align}
     \bar z=&1-z, & \bar y=&1-y, 
\end{align}
and let
\begin{align}
    \bar\eta(\bar z,\bar y)=&\eta(z,y), 
\end{align}
where $\bar \eta(\bar z, \bar y)$ is defined on the domain $\bar{\mathcal{D}}_1=\{(\bar z,\bar y)\in [0,1]^2:0\le \bar y \le \bar z\le 1\}$, which yields
\begin{align}
    \bar\eta_{\bar z}(\bar z,\bar y) =& -\bar\eta_{\bar y}(\bar z, \bar y),\label{equ_bar_eta} \\
        \bar\eta(\bar z,0) = &\int_0^{\bar z}\frac{\theta}{a} \bar\eta(\bar z,\bar y)d\bar y+\frac{1}{a}\left(\varphi(1-\bar z)^{\mathrm{T}}B-\theta\right). \label{equ_bar_eta boundary}
\end{align}
The equations~\eqref{equ_bar_eta} and \eqref{equ_bar_eta boundary} are analogous to the kernel equation studied in \cite[Thm.~1]{krstic2008backstepping}, differing only in their coefficients.
Recall that the term $\varphi(z)^{\mathrm{T}} = K \mathrm{e}^{A\frac{1-z}{a}}\in C([0,1];\mathbb{R}^{1\times n})$ appears in the boundary condition for $\eta$. Then, by the results in \cite[App.~D.1]{anfinsen2019adaptive}, the solution to \eqref{equ_bar_eta} and \eqref{equ_bar_eta boundary} belongs to $C(\bar{\mathcal{D}}_1)$; therefore, we recover $\eta \in C(\mathcal{D}_1)$.

Second, we consider the equation for $q(s,y)$, whose well-posedness can be established via $C_0$-semigroup theory. We introduce the change of variables:  
\begin{align}
    &\bar y = \frac{1}{\tau a}(1-y),  
\end{align}
and let
\begin{align}
     \bar q(s,\bar y)= q(s,y),
\end{align}
where $\bar q $ is defined on the domain $\bar{\mathcal{D}}_2=\{(s,\bar y)\in [0,1]\times [0,\frac{1}{\tau a}]\}$, which leads to
\begin{align}
    \bar q_s(s,\bar y) =&- \bar q_{\bar y}(s,\bar y), \\
    \bar q(0,\bar y) =& f(\bar y), \\
    \bar q(s,0) =& \int_0^\frac{1}{\tau a} \tau \theta\bar q(s,\bar y)d\bar y + \frac{1}{a}\psi(s)^{\mathrm{T}}B, \label{bar q boundary}
\end{align}
where $f(\bar y):=\eta(0,1-\tau a\bar y) \in C([0, \frac{1}{\tau a}]; \mathbb{R})$.

Following a method similar to that in \cite{sano2019boundary}, we remove the inhomogeneous term $\frac{1}{a}\psi(s)^{\mathrm{T}}B$ from the boundary condition~\eqref{bar q boundary}. Let
\begin{align}
    Q(s,\bar y) = \bar q(s,\bar y)-\frac{1}{a-\theta}\psi(s)^\mathrm{T}B,\quad  (\theta\neq a),\label{change of variable to Q}
\end{align}
so we obtain the equation for $Q$ as follows:
\begin{align}
    Q_s(s,\bar y) =& -Q_{\bar y}(s, \bar y)+g(s), \\
    Q(0,\bar y) =& h(\bar y),\\
    Q(s,0) =& \int_0^\frac{1}{\tau a} \tau \theta  Q(s,\bar y)d\bar y,
\end{align}
where $g(s)=\frac{1}{\theta-a}\psi^\prime(s)^{\mathrm{T}}B$, $h(\bar y)=f(\bar y)-\frac{1}{a-\theta}\psi(0)^{\mathrm{T}}B \in  C([0, \frac{1}{\tau a}]; \mathbb{R})$.

We define the unbounded operator $\mathcal{A}: D(\mathcal{A})\subset L^2(0,\frac{1}{\tau a}) \to L^2(0,\frac{1}{\tau a})$ as follows: 
\begin{align}
    (\mathcal{A}\varphi)(\bar y) =& -\varphi^\prime(\bar y), \quad \varphi\in D(\mathcal{A}), \\
    D(\mathcal{A}) =&\left\{\varphi\in H^1\left(0,\frac{1}{\tau a}\right): \varphi(0) =\int_0^\frac{1}{\tau a} \tau \theta  \varphi(\bar y)d\bar y\right\},
\end{align}
then $Q$ satisfies the abstract evolution equation
\begin{align}
    \frac{dQ}{ds} =& \mathcal{A}Q + g(s), \label{Q}\\
    Q(0) =& h. \label{Q initial}
\end{align}
We know that 
\begin{enumerate}
    \item the operator $\mathcal{A}$ generates a $C_0$-semigroup $T_\mathcal{A}(s)$ on $L^2(0,\frac{1}{\tau a})$, see \cite[Prop.~3.4]{sano2019boundary}; 
    \item the inhomogeneous term $g(s)\in C([0,1]; L^2(0,\frac{1}{\tau a}))$, recalling $\psi(s)^{\mathrm{T}} = \varphi(0)^{\mathrm{T}}\mathrm{e}^{A\tau s}$;
    \item the initial condition $h\in L^2(0,\frac{1}{\tau a})$.
\end{enumerate}
Consequently, the system~\eqref{Q}--\eqref{Q initial} admits a unique mild solution $Q\in C([0,1];L^2(0,\frac{1}{\tau a}))$, which is given by:
\begin{align}
    Q(s, \bar y) = T_{\mathcal{A}}(s)h(\bar y) + \int_0^s T_\mathcal{A}(s-\tau)g(\tau)d\tau,\quad \forall~ s\in [0,1].
\end{align}

By the inverse transformation, it follows that $q\in C([0,1];L^2(0,1))$. This establishes the well-posedness of System~\eqref{eta, p, q}. 
\end{proof}

\section{\texorpdfstring{Derivation of the Explicit Solution for $\bar k$}{Derivation of the Explicit Solution for Bar k}}\label{appendix:derivation of explicit solution for bar k}
Similar to $\bar l$, the solution to Equations~\eqref{k-1}--\eqref{k-3} exhibits a two-branch structure determined by its two boundary conditions \eqref{k-2} and \eqref{k-3}. For $0\le y\le\lambda_1(s)$, we denote this branch by $\bar k_1$. By applying the boundary condition~\eqref{k-3} alongside the method of characteristics, we obtain the explicit expression
    \begin{align}
        \bar{k}_1(z,s,y)=&-d_1\delta(y-\lambda_1(s)) \notag \\
        &-\zeta_1\delta(y-\lambda_1(s)+1-z). \label{bar k first part}
    \end{align}
    For $\lambda_1(s)<y\le 1$, obtaining $\bar k$ becomes more involved due to the boundary condition \eqref{k-2}, which contains a self-integral term. In this region, we decompose the solution $\bar k$ into two components as $\bar k=\bar k_2+\bar k_3$. By substituting the explicit expression for $l$ in \eqref{l solution} and $\bar k_1$ in \eqref{bar k first part}, integrating \eqref{k-1} along the characteristics, and applying the boundary condition \eqref{k-2}, we arrive at the following integral equations:
    \begin{align}
        \bar{k}_2+\bar{k}_3=&-\frac{d_1^2}{a}-\frac{d_1}{a}\zeta_1\int_0^{\chi_4(s,y)}\delta(\xi-\chi_4(s,y)+1-z)d\xi \notag \\
        &+\varTheta(z,s,y)+\varLambda[\bar{k}_2+\bar{k}_3](z,s,y), \label{k integral equation} 
    \end{align} 
    where 
    \begin{align}
    \varLambda[\bar k](z,s,y) =&\frac{d_1}{a}\int_{\chi_4(s,y)}^1\bar k(z,\chi_3(s,y),\xi)d\xi,\\
    \varTheta(z,s,y) =& \bar \zeta_2\mathrm{e}^{D(c_2-c_1)\chi_3(s,y)}\delta(3\!-\!z\!-\!\lambda_2(s)\!-\!\frac{b}{a}(1-y)). 
    \end{align}
    In view of the linearity of the operator $\varLambda$, \eqref{k integral equation} reduces to
    \begin{align}
        \bar{k}_2(z,s,y) =& -\frac{d_1}{a}\zeta_1\int_0^{\chi_4(s,y)}\delta(\xi-\chi_4(s,y)+1-z)d\xi
        \notag \\
        &-\frac{d_1^2}{a}+\varLambda[\bar{k}_2](z,s,y), \label{k2 integral equation} \\
        \bar{k}_3(z,s,y)   =& \varTheta(z,s,y)+\varLambda[\bar{k}_3](z,s,y). \label{k3 integral equation} 
    \end{align}
    
    We now proceed to solve for $\bar k_2$ and $\bar k_3$ separately. Consider $\bar k_2$ first. Applying the shifting property of the delta function, we rewrite \eqref{k2 integral equation} as
    \begin{subequations}
    \begin{align*}
        \bar k_2(z,s,y) =& \varGamma_{21}(z) + \varLambda[\bar k_2](z,s,y),\mathrm{if}~\lambda_1(s)\!<\!y\!\le\! \vartheta_1(z,s), \\
        \bar k_2(z,s,y) =& \varGamma_{22}(z) + \varLambda[\bar k_2](z,s,y), \mathrm{if}~\vartheta_1(z,s)\!<\!y\!\le \!1, 
    \end{align*}
    \end{subequations}
    where  
    \begin{align}
        \varGamma_{21}(z)=&-\frac{d_1^2}{a}-\frac{d_1}{a}\zeta_1H(1-z), & \varGamma_{22}(z)=&-\frac{d_1^2}{a}.
    \end{align}
    Following a procedure similar to that for solving Equations~\eqref{with zeta2} and \eqref{without zeta2}, we obtain the solution for $\bar k_2$ as follows:
    \begin{subequations}\label{bar k2 piecewise expression}
       \begin{align}
           \bar k_2(z,s,y) =&\varGamma_{21}(z)\mathrm{e}^{\frac{d_1}{a}(y-\lambda_1(s))}, & &\mathrm{if}~\lambda_1(s)<y\le \vartheta_1(z,s),\label{bar k2 piecewise expression branch 1} \\
           \bar k_2(z,s,y) =& \varGamma_{22}(z)\mathrm{e}^{\frac{d_1}{a}(y-\lambda_1(s))}, &&\mathrm{if}~\vartheta_1(z,s)\!<\!\!y\!\!\le \!1. \label{bar k2 piecewise expression branch 2}
       \end{align}
    \end{subequations}

    The solution for $\bar k_3$ is addressed subsequently. Set the following recursive formula with $\bar k_3^0=0$, 
    \begin{align}
        \bar{k}_3^{n+1}(z,s,y)   =& \varTheta(z,s,y)+\varLambda[\bar{k}_3^n](z,s,y). 
    \end{align}
    Define $\Delta \bar k_3^n=\bar k_3^{n+1}-\bar k_3^{n}$, for $n=0,1, 2\cdots$, which yields $\Delta \bar k_3^0(z,s,y)=\varTheta(z,s,y)$ and   
    \begin{align}
        \Delta\bar{k}_3^{n+1}(z,s,y) =& \varLambda[\Delta\bar{k}_3^n](z,s,y). \label{n+1 and n relation}
    \end{align}
    Substituting $\Delta\bar k_3^0$ into \eqref{n+1 and n relation} and repeating this iterative procedure, we can prove by a similar inductive argument that for $n > 1$,
    \begin{align*}
        \Delta\bar{k}_3^{n}(z,s,y) =& 0, \qquad \qquad\quad \mathrm{if}~\lambda_1(s) <y <\vartheta_4(z,s), \\
        \Delta\bar{k}_3^{n}(z,s,y) =& \frac{d_1}{b}\bar\zeta_2\mathrm{e}^{\frac{c_2-c_1}{b}(z-1)}\frac{\left(\frac{d_1}{a}(y\!-\!\lambda_1(s))\right)^{n-1}}{(n-1)!}H(z\!-\!1), \notag \\
        &\qquad \qquad \qquad\mathrm{if}~\vartheta_4(z,s)\le y\le 1, 
    \end{align*}
    which ensures the convergence of the series $\bar k_3(z,s,y)=\sum_{n=0}^{\infty}\Delta\bar k_3^{n}(z,s,y)$, resulting in the following explicit expression for $\bar k_3$:
    \begin{subequations}\label{bar k3}
    \begin{align}
        \bar k_3(z,s,y) =& \varTheta(z,s,y) ,\qquad\mathrm{if}~\lambda_1(s)<\!y\!< \vartheta_4(z,s),\label{bar k3 branch 1}\\
        \bar k_3(z,s,y) =& \varTheta(z,s,y)+\varGamma_{23}(z)\mathrm{e}^{\frac{d_1}{a}(y-\lambda_1(s))}, \notag  \\
        &\qquad\qquad \qquad~\mathrm{if}~\vartheta_4(z,s)\le y\le 1, \label{bar k3 branch 2}
    \end{align}
    \end{subequations}
    where 
    \begin{align}
        \varGamma_{23}(z) =& \frac{d_1}{b}\bar\zeta_2\mathrm{e}^{\frac{c_2-c_1}{b}(z-1)}H(z-1).
    \end{align}
    For $\lambda_1(s)<y\le 1$, when $z>1$, since $\vartheta_1(z,s)=1$, $\bar k_2$ is determined by its first branch~\eqref{bar k2 piecewise expression branch 1}; when $z\le 1$, we have $\vartheta_4(z,s)=\lambda_1(s)$, and $\bar k_3$ is determined by its second branch~\eqref{bar k3 branch 2}. Consequently, combining \eqref{bar k2 piecewise expression} and \eqref{bar k3} yields $\bar k=\bar k_2+\bar k_3$: 
    \begin{subequations}
    \begin{align}
        &\mathrm{For}~z\le 1,\notag \\
       &\qquad\begin{aligned}
           \bar k(z,s,y) =&\varTheta(z,s,y)+\varGamma_{21}(z)\mathrm{e}^{\frac{d_1}{a}(y-\lambda_1(s))}, \notag \\ &\qquad \qquad \qquad \qquad\mathrm{if}~\lambda_1(s)<y\le \vartheta_1(z,s),\\
           \bar k(z,s,y) =& \varTheta(z,s,y)+\varGamma_{22}(z)\mathrm{e}^{\frac{d_1}{a}(y-\lambda_1(s))},\notag \\
           &\qquad\qquad \qquad \qquad\mathrm{if}~\vartheta_1(z,s)<y\le 1.
       \end{aligned}\\
       &\mathrm{For}~z> 1,\notag\\
       &\qquad\begin{aligned}
        \bar k(z,s,y) =& \varTheta(z,s,y)+\varGamma_{21}(z)\mathrm{e}^{\frac{d_1}{a}(y-\lambda_1(s))}, \notag \\
        &\qquad \qquad \qquad \qquad
        \mathrm{if}~\lambda_1(s)<y< \vartheta_2(z,s), \\
        \bar k(z,s,y) =& \varTheta(z,s,y)+\left(\varGamma_{21}(z)+\varGamma_{23}(z)\right)\mathrm{e}^{\frac{d_1}{a}(y-\lambda_1(s))}, \notag \\
        &\qquad \qquad \qquad \qquad\mathrm{if}~\vartheta_2(z,s)\le y\le 1,
    \end{aligned}
    \end{align}   
   \end{subequations}
    Note that for $z>1$, we have $\vartheta_4=\vartheta_2$. This completes the derivation of the explicit solution for $\bar k$.

\section{Proofs of Lemmas~\ref{direct norm} and~\ref{direct norm u}}\label{appendix:proofs of lemmas 2 and 4}
\begin{proof}
Since integral operators with bounded kernels are bounded, it suffices to handle the singular terms involving the Dirac delta function, such as $l_1$, $k_1$, and $k_2$. Substituting the solutions for $l$ and $k$ given in \eqref{l solution} and \eqref{k solution} into $\mathcal{T}^{\beta}[f]$ and $\mathcal{T}^{\alpha}[f]$, respectively, we obtain
\begin{align}
    \mathcal{T}^{\alpha}[f](z,s) =&\frac{d_1\varrho_1(s)}{-\kappa_1}f(\lambda_1(s))+J_1[f](z,s)+J_2[f](z,s)\notag \\
    &+\int_{\lambda_1(s)}^1\frac{1}{\kappa_1}k_3(z,s,y)f(y)dy, \label{t alpha}\\
    \mathcal{T}^{\beta}[f](z,s) =&\frac{d_2\varrho_2(s)}{-\kappa_2}f(\lambda_2(s))H(z-1)+J_3[f](z,s)\notag \\
    &+\int_{\lambda_2(s)}^2\frac{1}{\kappa_2}l_2(z,s,y)f(y)dy, \label{t beta}
\end{align}
where $l_2$ and $k_3$ are bounded functions satisfying \eqref{l2} and \eqref{k3}, respectively. For $Das\le z\le 1$, $J_1[f](z,s)=\frac{\zeta_1}{-\kappa_1}\varrho_1(s)f(z-Das)$, and $J_1=0$ otherwise. For $1\le z\le 1+Dbs$, $J_2[f](z,s)=\frac{a\bar \zeta_2\varrho_1(s)}{b\kappa_1}\mathrm{e}^{\frac{c_1-c_2}{b}(Dbs+1-z)}f(\lambda_1(s)+\frac{a}{b}(z-1))$, and $J_2=0$ otherwise. For $Dbs+1\le z\le 2$, $J_3[f](z,s)=\frac{\zeta_2}{-\kappa_2}\varrho_2(s)f(z-Dbs)$, and $J_3=0$ otherwise. Since $J_1$, $J_2$, and $J_3$ are bounded functions, the estimate \eqref{direct norm estimate} follows directly from \eqref{t alpha} and \eqref{t beta} for $j=\alpha$, $\beta$.

Then, for $\mathcal{T}^v[f]$, we denote $\mathcal{T}^v[f]=\mathcal{T}^v_1[f]+\mathcal{T}^v_2[f]$, where  
    \begin{align}
        &\mathcal{T}^v_1[f](z,s) \stackrel{\triangle}{=} D\!\!\int_0^1\!\!\!\!\int_0^s \!\!k(z,s\!-\!\epsilon,y)f(y,\epsilon)d\epsilon dy, \label{T v k} \\
        &\mathcal{T}^v_2[f](z,s) \stackrel{\triangle}{=} D\!\!\int_1^2\!\!\!\!\int_0^s \!\!l(z,s\!-\!\epsilon,y)f(y,\epsilon)d\epsilon dy. \label{T v l}
    \end{align}
Substituting the solutions for $k$ and $l$ into \eqref{T v k} and \eqref{T v l}, respectively, yields 
\begin{align}
        &\begin{aligned}[b]\mathrm{for}~z<&1,\qquad\quad \\
        \mathcal{T}^v_1[f]&(z,s)=\!\!\int_{0}^s\int_{\lambda_1(s)+Da\epsilon}^{1}\!\!Dk_3(z,s-\epsilon,y)f(y,\epsilon)dyd\epsilon \\ 
        &-\int_{\max\{0,s-\frac{z}{Da}\}}^{s}D\zeta_1\varrho_1(s\!-\!\epsilon)f(Da(\epsilon\!-\!s)\!+\!z,\epsilon)d\epsilon\\
        &-\int_{0}^{s}Dd_1\varrho_1(s-\epsilon)f(\lambda_1(s)+Da\epsilon,\epsilon)d\epsilon,
        \end{aligned} \label{T v 1 z<1} \\
        &\begin{aligned}[b]\mathrm{for}~z\ge&1,\\
        \mathcal{T}^v_1[f]&(z,s)=\!\!\int_{0}^s\int_{\lambda_1(s)+Da\epsilon}^{1}\!\!Dk_3(z,s-\epsilon,y)f(y,\epsilon)dyd\epsilon  \\ 
        &-\int_{0}^{s}Dd_1\varrho_1(s-\epsilon)f(\lambda_1(s)+Da\epsilon,\epsilon)d\epsilon \\
        &+\int_{0}^{\max\{0,s-\frac{z-1}{Db}\}}\frac{Da\bar \zeta_2}{b}\mathrm{e}^{\frac{c_2-c_1}{b}(z-1)}\varrho_1(s-\epsilon)\\
        &\times f(y_0(z,s,\epsilon),\epsilon)d\epsilon,
        \end{aligned} \label{T v 1 z>1}
\end{align}
where $y_0(z,s,\epsilon)=1+Da(\epsilon-s)+\frac{a}{b}(z-1)$, and 
\begin{align}
\mathcal{T}^v_2[f]&(z,s)=\!\!\int_{0}^s\int_{\lambda_2(s)+Db\epsilon}^{1}\!\!Dl_2(z,s-\epsilon,y)f(y,\epsilon)dyd\epsilon  \notag \\ 
        &-H(z-1)\bigg[\int_{0}^{s}Dd_2\varrho_2(s-\epsilon)f(\lambda_2(s)+Db\epsilon,\epsilon)d\epsilon  \notag  \\
        &+\int_{\max\{0,s-\frac{z-1}{Db}\}}^{s}D\zeta_2\varrho_2(s\!-\!\epsilon)f(Db(\epsilon\!-\!s)+z,\epsilon)d\epsilon\bigg]. \label{T v 2}
\end{align} 
Equations~\eqref{T v 1 z<1}--\eqref{T v 2} are expressed in explicit form with bounded integral kernels, from which \eqref{direct norm estimate} follows directly for $j=v$.
    
Finally, employing the solution for $M$, we obtain
\begin{align}
       \mathcal{T}^u[f](z,s) =&\int_{\lambda_1(s)}^1\!\!\left(\sum\limits_{n=1}^{\infty}\Delta M_3^n(z,s,y)\right)f(y)dy\notag \\
       +\int_{\lambda_1(s)}^1&M_2(z,s,y)f(y)dy +J_4[f](z,s)+J_5[f](z,s)\notag \\
       + \frac{d_1}{\zeta_1}m_1&(s,\lambda_1(s))f(\lambda_1(s)),\label{M integral partition}
\end{align}
    where $M_2$ and the series $\sum_{n=1}^{\infty}\Delta M_3^n(z,s,y)$ are bounded according to \eqref{M2 bound} and \eqref{M4 bound}, respectively. For $1\le z\le 1+Dbs$, $J_4[f](z,s)=\frac{a}{b}m_2\bigl(s,\lambda_1(s)\!+\!\frac{a}{b}(z\!-\!1)\bigr)f\bigl(\lambda_1(s)\!+\!\frac{a}{b}(z\!-\!1)\bigr)$, and $J_4=0$ otherwise. For $Das\le z\le 1$, $J_5[f](z,s)=\left(m_1(s,z-Das)+m_0(z)H(Da-z)\right)f(z-Das)$, and $J_5=0$ otherwise. This directly establishes \eqref{T_u}, which completes the proof.
\end{proof}

\section{Proof of Lemma~\ref{inverse norm}}\label{appendix:proof of lemma 6}
\begin{proof}
        The proof is analogous to that of Lemma~\ref{direct norm}. Accordingly, we detail only the derivation associated with the term $\mathcal{\breve T}^{\mu}[f]$. Recall the definition of $\mathcal{\breve{T}}^{\mu}[f](z,s)$ in \eqref{T mu definition}, the inverse kernel $\breve{N}(z,s,y)$ given in Lemma~\ref{lemma: inverse kernel solution}, and the solution $\breve M_1$ in \eqref{breve M1}. Substituting the explicit expression for $\breve{N}(z,s,y)$ into $\mathcal{\breve{T}}^{\mu}[f](z,s)$, we obtain for the case $\tau\ge D$, 
        \begin{subequations}
        \begin{align}
        &\mathcal{\breve{T}}^{\mu}[f](z,s) =\bigl(\breve{m}_1(\frac{z}{Da},0)\!-\!m_1(\frac{z}{Da},0)\bigr)f(1\!+\!\frac{z\!-\!Das}{\tau a}), \notag \\
        &\qquad \qquad \qquad\qquad \qquad \qquad\qquad \mathrm{if} ~ 0\le z\le Das,\\
       &\mathcal{\breve{T}}^{\mu}[f](z,s)=0, \qquad \qquad\qquad \quad \mathrm{otherwise}.
    \end{align}
       \end{subequations}
    It follows from this explicit expression that $\|\mathcal{\breve T}^{\mu}[f]\|_\infty\le \breve C_\mu\|f\|_\infty$ holds for some positive constant $\breve C_\mu$. An analogous expression---and hence the same norm bound---can be derived for the case $\tau<D$. 
\end{proof}

\section{Transformation of the Nonlinear Screw Extrusion Model into \texorpdfstring{System~\eqref{main model}}{System (1)}}\label{appendix:transformation of screw extrusion system}
The nonlinear screw extrusion model \eqref{screw extrusion model} is transformed into system~\eqref{main model} through three successive steps. Before proceeding, note that the pressure equation~\eqref{eq:mass balance of FFZ} is a spatial ODE that can be simplified via its relationship with $F_\mathrm{d}(t)$, i.e., $F_\mathrm{d}(t) = \frac{k_\mathrm{d}}{\eta}\Delta P(t)=\frac{k_\mathrm{d}}{\eta}\left(p(L,t)-P_\mathrm{0}\right)$. Specifically, integrating \eqref{eq:mass balance of FFZ} from $l(t)$ to $L$ yields
    \begin{align}
		\Delta P(t) &= 
		\frac{\eta\rho VN(L-l(t))}{B\rho+k_\mathrm{d}(L-l(t))}.\label{eq:P(l(t))}
	\end{align} 
    Consequently, $p(x,t)$ can be explicitly expressed in terms of the state $l(t)$, bypassing the need for direct analysis of $p(x,t)$ in the subsequent derivation.
    \subsection{Model Representation in Fixed Domains}
    
    First, a coordinate transformation is introduced to map the time-varying spatial intervals from $x \in [0,l(t))$ (respectively, $x \in (l(t),L])$ onto the fixed intervals $z \in [0,1)$ (resp. $z\in(1,2]$) as follows:
    \begin{subequations}
	\begin{align}
		z &=\frac{x}{l(t)},\qquad \qquad ~~~
		\mathrm{for}~~ x\in [0,l(t)],   \\         
		z &=\frac{x+L-2l(t)}{L-l(t)}  
		\label{change of variable},\quad \mathrm{for}~~ x\in (l(t),L],
	\end{align}
    \end{subequations}
	with the inverse coordinate transformation given by
    \begin{subequations}
	\begin{align}
		x &=zl(t),&			&\mathrm{for}~~ z\in [0,1],   \\         
		x &=(2-z)l(t)+(z-1)L,  
		\label{inverse change of variable}\quad &	&\mathrm{for}~~ z\in (1,2].
	\end{align}
    \end{subequations}
	
	We denote
	\begin{align}
		\bar{f}_\mathrm{p}(z,t) &= f_\mathrm{p}(x,t), & \bar{T}_\mathrm{p}(z,t) &= T_\mathrm{p}(x,t), \\
		\bar{T}_\mathrm{b}(z,t-D) &= T_\mathrm{b}(x, t-D), & \bar{T}_\mathrm{f}(z,t) &= T_\mathrm{f}(x,t),
	\end{align}
	and let $h_\mathrm{p} := \frac{S_\mathrm{ech}\hbar}{\rho Vc_\mathrm{p}}$, $
	h_\mathrm{f} := \frac{S_\mathrm{ech}\hbar}{\rho Vc_\mathrm{f}}$, $
	\beta_\mathrm{p} := \frac{\mu_\mathrm{p}\eta N^2}{\rho Vc_\mathrm{p}}$, 
	$\beta_\mathrm{f} := \frac{\mu_\mathrm{f}\eta N^2}{\rho Vc_\mathrm{f}}$. Then, in the new coordinates, the system \eqref{screw extrusion model} can be rewritten as
    \begin{subequations}\label{fixed-domain model}
	\begin{align}
		\frac{\partial \bar{f_\mathrm{p}}}{\partial t}(z,t) = &-\alpha_\mathrm{p}\frac{\partial \bar{f_\mathrm{p}}}{\partial z}(z,t), \qquad \qquad  \quad  ~z\in(0,1),  \label{eq:f_p}  \\
		\frac{\partial \bar{T}_\mathrm{p}}{\partial t}(z,t) = &-\alpha_\mathrm{p}\frac{\partial \bar{T}_\mathrm{p}}{\partial z}(z,t)-h_\mathrm{p}\bar{T}_\mathrm{p}(z,t)+\frac{\beta_{\rm{p}}}{\bar{f_\mathrm{p}}(z,t)}\notag \\
        &+h_{\rm{p}}\bar{T}_\mathrm{b}(z, t-D), \qquad\qquad z\in(0,1), \\
		\frac{\partial \bar{T}_\mathrm{f}}{\partial t}(z,t) =
		&-\alpha_\mathrm{f}\frac{\partial \bar{T}_\mathrm{f}}{\partial z}(z,t)-h_{\rm{f}}\bar{T}_\mathrm{f}(z,t)+\beta_{\rm{f}}\notag \\
        &+h_{\rm{f}}\bar{T}_\mathrm{b}(z, t-D),\qquad\qquad z\in (1,2),\\
		\frac{dl(t)}{dt} =& \frac{\frac{k_\mathrm{d}}{\eta}\Delta P(t)-\rho VN\bar{f}_\mathrm{p}(1,t)}{\rho S_\mathrm{eff}(1-\bar{f}_\mathrm{p}(1,t))}, \label{eq:l(t)}\\
		\bar{f}_\mathrm{p}(0,t) =&\frac{F_{\rm{in}}(t-\tau)}{\rho V N},  \label{fixed domain:boundary 1}\\
		\bar{T}_\mathrm{p}(0,t) =& {T}_\mathrm{in},\label{fixed domain:boundary 2} \\
		\bar{T}_\mathrm{f}(1,t) =& \bar{T}_\mathrm{p}(1,t), 
		\label{fixed domain:boundary 3}
	\end{align}  
    \end{subequations}
	where 
	\begin{align}
		\alpha_\mathrm{p} &= \frac{1}{l(t)}[\xi N-z\frac{dl(t)}{dt}],  \label{alpha_p}\\
		\alpha_\mathrm{f} &=  \frac{1}{L-l(t)}[\frac{k_\mathrm{d}\Delta P(t)\xi}{\eta\rho V}-(2-z)\frac{dl(t)}{dt}]. \label{alpha_f}
	\end{align}

	\subsection{Equilibrium Point}
	The equilibrium point $(l_\mathrm{e}, \bar{f}_\mathrm{pe}(z), \bar{T}_\mathrm{pe}(z),\bar{T}_\mathrm{fe}(z))$ is obtained by setting the time derivative to zero, which leads to
	\begin{align*}
		0 =& -\alpha_\mathrm{pe}\bar{f}_{\mathrm{pe}}^\prime(z), \\
		0 =& -\alpha_\mathrm{pe}\bar{T}_{\mathrm{pe}}^\prime(z)+h_{\rm{p}}\Bigl(\bar{T}_\mathrm{b}^\star(z)-\bar{T}_\mathrm{pe}(z)\Bigr)
		+\frac{\beta_{\rm{p}}}{\bar{f}_{\mathrm{pe}}(z)}, \\
		0 =& -\alpha_\mathrm{fe}\bar{T}_{\mathrm{fe}}^\prime(z)+h_{\rm{f}}\Bigl(\bar{T}_\mathrm{b}^\star(z)-\bar{T}_\mathrm{fe}(z)\Bigr)
		+\beta_{\rm{f}}, \\
		0 =& \rho VN\bar{f}_\mathrm{pe}(1,t)-\frac{k_\mathrm{d}}{\eta}\Delta P_\mathrm{e},
	\end{align*}
	with conditions,
	\begin{align*}
		\bar{f}_\mathrm{pe}(0) =\frac{F_{\rm{in}}^\star}{\rho V N},\quad
		\bar{T}_{\mathrm{pe}}(0) = T_{\mathrm{in}}, \quad
		\bar{T}_{\mathrm{fe}}(1) = \bar{T}_{\mathrm{pe}}(1),
	\end{align*}
    where $F_\mathrm{in}^\star$ and $\bar{T}_\mathrm{b}^\star(z)$ are steady-state operating points, and $\Delta P_\mathrm{e}$, $\alpha_\mathrm{pe}=\frac{\xi N}{l_\mathrm{e}}$, and $\alpha_\mathrm{fe}=\frac{k_\mathrm{d}\xi N}{B\rho+k_\mathrm{d}(L-l_\mathrm{e})}$ denote the steady-state values of \eqref{eq:P(l(t))}, \eqref{alpha_p}, and \eqref{alpha_f}, respectively. Solving the above differential equations yields the steady-state solutions, 
	\begin{align}
		\bar{f}_\mathrm{pe}(z) &= \frac{F_\mathrm{in}^\star}{\rho VN}, \qquad \qquad \qquad\qquad \qquad ~z\in[0,1],  \label{eq: equilibrium 1}\\
		\bar{T}_\mathrm{pe}(z) &= \frac{r_2}{r_1}+ a_1\mathrm{e}^{-r_1l_\mathrm{e}z}, \qquad  ~~ \qquad \qquad  z\in[0,1], \label{eq: equilibrium 2}  \\
		\bar{T}_\mathrm{fe}(z) &= \frac{r_4}{r_3}+ a_2\mathrm{e}^{-r_3\bigl((L-l_\mathrm{e})z+2l_\mathrm{e}-L\bigr)}, ~~~\, z\in[1,2],
		\label{eq: equilibrium 3} \\
		l_\mathrm{e} &= L - \frac{B\rho F_{\mathrm{in}}^\star}{k_\mathrm{d}(\rho VN - F_{\mathrm{in}}^\star)}, \label{eq: equilibrium 4}  
	\end{align}
	where
	\begin{align}
		r_1 =& \frac{h_\mathrm{p}}{\xi N},&
		r_2 =&\bigl(\bar{T}_\mathrm{b}^\star(z)+\frac{\beta_\mathrm{p}}{h_\mathrm{p}\bar{f}_\mathrm{pe}}\bigr)r_1, \\
		r_3 =& \frac{h_\mathrm{f}\rho V}{F_\mathrm{in}^\star\xi},& 
		r_4 =& \bigl(\bar{T}_\mathrm{b}^\star(z)+\frac{\beta_\mathrm{f}}{h_\mathrm{f}}\bigr)r_3, \label{equilibrium parameter 2} \\
		a_1 =& T_\mathrm{in}-\frac{r_2}{r_1}, &
		a_2 =& \bigl(T_\mathrm{in}-\frac{r_2}{r_1}\bigr)\mathrm{e}^{-(r_1\!-\!r_3)l_\mathrm{e}}\!+\!\bigl(\frac{r_2}{r_1}\!-\!\frac{r_4}{r_3}\bigr)\mathrm{e}^{r_3l_\mathrm{e}}.\label{equilibrium parameter 4}
	\end{align}
	
	\subsection{Linearization Around the Equilibrium Point}\label{Sec_II-B}
	Define the error variables,
	\begin{align}
		X(t)&=l(t)-l_\mathrm{e}, \\
		u(z,t)&=\bar{f}_\mathrm{p}(z,t)-\bar{f}_\mathrm{pe}(z),	\,\quad \qquad \qquad z\in [0,1], \\
		\alpha(z,t)&=\bar{T}_\mathrm{p}(z,t)-\bar{T}_\mathrm{pe}(z), \qquad \quad \qquad z\in[0,1],\\
		\beta(z,t)&=\bar{T}_\mathrm{f}(z,t)-\bar{T}_\mathrm{fe}(z), ~\quad \qquad \qquad z\in[1,2],\\
		U_1(t) &= \frac{F_\mathrm{in}(t) - F_\mathrm{in}^\star}{\rho VN}, \\
		U_2(z,t) &= \bar{T}_\mathrm{b}(z,t) - \bar{T}_\mathrm{b}^\star(z), ~\quad\qquad \qquad z\in[0,2].
	\end{align} 
	
	Linearizing the system~\eqref{fixed-domain model} around the equilibrium point $(l_\mathrm{e}, \bar{f}_\mathrm{pe}(z), \bar{T}_\mathrm{pe}(z),\bar{T}_\mathrm{fe}(z))$ via Taylor series with first order approximation, we obtain
    \begin{subequations}\label{linearized system}
	\begin{align}
		u_t(z,t) =& -a_0u_z(z,t),   \label{linearized system 1}\\
		\alpha_t(z,t) =&-a_0\alpha_z(z,t)-h_{\mathrm{p}}\alpha(z,t)+h_{\mathrm{p}}U_2(z, t-
		D) \notag \\
		&-c_0u(z,t)+g_1(z)X(t)+g_2(z)u(1,t),  \label{linearized system 2}\\
		\beta_t(z,t)=&-b_0\beta_z(z,t) -h_{\mathrm{f}}\beta(z,t)+h_{\mathrm{f}}U_2(z, t-
		D) \notag \\
		&+g_3(z)X(t)+g_4(z)u(1,t), \label{linearized system 3}\\
		\dot{X}(t) =& -\gamma_0 X(t)-\theta_0u(1,t), \label{linearized system 4}\\
		u(0,t) =& U_1(t-\tau), \label{linearized system 5}\\
		\alpha(0,t) =& 0, \label{linearized system 6}\\
		\beta(1,t) =& \alpha(1,t), \label{linearized system 7}
	\end{align}
    \end{subequations}
    where
	\begin{align}
		a_0 &= \frac{\xi N}{l_\mathrm{e}}, \quad	
		b_0 = \frac{k_\mathrm{d}\xi N}{B\rho+k_\mathrm{d}(L-l_\mathrm{e})},\quad
		&c_0 =& \frac{\beta_\mathrm{p}}{\bar{f}_\mathrm{pe}^2(z)},\label{def_abc} \\
		\gamma_0 &= \frac{k_\mathrm{d}\xi NB\rho}{[1-\bar{f}_\mathrm{pe}(1)][B\rho+k_\mathrm{d}(L-l_\mathrm{e})]^2}, \\
        \theta_0 &= \frac{\xi NB\rho}{[1-\bar{f}_\mathrm{pe}(1)]^2[B\rho+k_\mathrm{d}(L-l_\mathrm{e})]}. \label{def_theta} 	
	\end{align}
	The coefficient functions $g_i(z)$ ($i = 1, \cdots, 4$), which depend on the steady-state solutions, are bounded and are given by
	\begin{align}
		g_1(z) =&\frac{a_0-\gamma_0 z}{l_\mathrm{e}}\bar{T}_\mathrm{pe}^\prime(z), \;~g_2(z) =-\frac{\theta_0 z}{l_\mathrm{e}}\bar{T}_\mathrm{pe}^\prime(z),  
		\label{g1 g2} \\
		g_3(z) =&-\left(\frac{k_\mathrm{d}^2\xi N}{[B\rho+k_\mathrm{d}(L-l_\mathrm{e})]^2}+\frac{\gamma_0(2-z
			)}{L-l_\mathrm{e}}\right)\bar{T}_\mathrm{fe}^\prime(z), \label{g4} \\
		g_4(z) =&-\frac{\theta_0 (2-z)}{L-l_\mathrm{e}}\bar{T}_\mathrm{fe}^\prime(z).
		\label{g5}
	\end{align}

    Therefore, the system \eqref{linearized system} represents a simplified version of the system~\eqref{main model}.

\end{document}